\documentclass[10pt, journal, letterpaper]{IEEEtran}

\usepackage{cite}

\ifCLASSINFOpdf
   \usepackage[pdftex]{graphicx}
  \DeclareGraphicsExtensions{.pdf,.jpeg,.png}
\else

\fi

\usepackage{array,booktabs,amsmath,graphicx,fancyvrb,tabularx}
\usepackage{color,xcolor}
\usepackage{eucal}
\usepackage{hyperref}
\usepackage{amsmath}
\usepackage{amsfonts}

\usepackage{subfigure}

\newcommand{\code}[1]{{\color{black}\texttt{#1}}}
\usepackage{amssymb}
\usepackage{amsbsy,amsmath,amssymb,epsfig,bbm,mathrsfs,multirow,amsthm}
\newtheorem{theorem}{Theorem}
\newtheorem{remark}{\textbf{Remark}}

\usepackage{fancyhdr} 

\newcommand{\hhw}[1]{\textcolor{black}{{#1}}} 

\newcommand{\hw}[1]{\textcolor{black}{{#1}}} 

\newcommand{\zk}[1]{\textcolor{black}{{#1}}} 

\newcommand{\circled}[1]{{\textcircled{\scriptsize{#1}}}}

\begin{document}

\title{BrokerChain: A Blockchain Sharding Protocol by Exploiting Broker Accounts}

\author{
 \IEEEauthorblockN{
Huawei~Huang\IEEEauthorrefmark{1},
Zhaokang~Yin\IEEEauthorrefmark{1},
Qinde~Chen\IEEEauthorrefmark{1},
Guang~Ye\IEEEauthorrefmark{1},
Xiaowen~Peng\IEEEauthorrefmark{1},
Yue~Lin\IEEEauthorrefmark{1},
Zibin~Zheng\IEEEauthorrefmark{1},
Song~Guo\IEEEauthorrefmark{2}
    }
    
    \IEEEauthorblockN{
    \IEEEauthorrefmark{1}SSE, Sun Yat-sen University, China. Email:~ huanghw28@mail.sysu.edu.cn
    \\
    \IEEEauthorrefmark{2}The Hong Kong University of Science and Technology.
    }

    \IEEEcompsocitemizethanks{\IEEEcompsocthanksitem{The conference version [28] of this article was presented in INFOCOM 2022.}
    }
}

\maketitle

\begin{abstract}
\hw{
 State-of-the-art blockchain sharding solutions such as Monoxide, can cause severely imbalanced distribution of transaction (TX) workloads across all blockchain shards due to the deployment policy of their accounts. Imbalanced TX distributions then produce \textit{hot shards}, in which the cross-shard TXs may experience an unlimited confirmation latency. Thus, how to address the hot-shard issue and how to reduce cross-shard TXs become significant challenges of blockchain sharding. Through reviewing the related studies, we find that a cross-shard TX protocol that can achieve workload balance among all shards and simultaneously reduce the quantity of cross-shard TXs is still absent from the literature. To this end, we propose BrokerChain, which is a cross-shard blockchain protocol dedicated to account-based state sharding. Essentially, BrokerChain exploits fine-grained state partition and account segmentation. We also elaborate on how BrokerChain handles cross-shard TXs through broker accounts. The security issues and other properties of BrokerChain are analyzed rigorously. Finally, we conduct comprehensive evaluations using an open-source blockchain sharding prototype named \textit{BlockEmulator}. The evaluation results show that BrokerChain outperforms other baselines in terms of transaction throughput, transaction confirmation latency, the queue size of the transaction pool, and workload balance.
 }

\end{abstract}

\section{Introduction}

The sharding technique is viewed as a promising solution that can improve the scalability of blockchains  \cite{Luu2016Elastico, Kokoris2018OmniLedger, 2018Chainspace, 2018RapidChain, nguyen2019optchain, prism, TowardsScaling, Wang2019Monoxide, tao2020sharding, 2021MVCom}. 
The idea of sharding is to \textit{divide and conquer} when dealing with all the transactions (TXs) submitted to a blockchain system. Instead of processing all TXs by all blockchain nodes following the conventional manner, the sharding technique divides blockchain nodes into smaller committees, each only having to handle a subset of TXs \cite{Luu2016Elastico}. Thus, the transaction throughput can be improved exponentially.

There are mainly three types of sharding paradigms, i.e., network sharding, transaction sharding, and state sharding. Since \textit{network sharding} divides the entire blockchain network into smaller committees, it is viewed as the foundation of the other two sharding paradigms. 
In each round of the conventional sharding protocol, after the committee's formation, several disjoint subsets of TXs are assigned to those committees following the paradigm of \textit{transaction sharding}. Then, committees run a specific consensus protocol locally, such as Practical Byzantine Fault Tolerance (PBFT) \cite{castro1999practical}, to achieve a local consensus towards the set of assigned TXs.
Among those three sharding techniques, \textit{state sharding} is the most difficult one because it has to ensure that all states of a blockchain are amortized by all shards. Currently, the state sharding is mainly staying in the early-stage theoretical study. Several representative sharding solutions have been proposed, such as  Elastico \cite{Luu2016Elastico}, Omniledger \cite{Kokoris2018OmniLedger}, RapaidChain \cite{2018RapidChain},  and Monoxide \cite{Wang2019Monoxide}.
Those sharding solutions are based on either \textit{UTXO} (Unspent Transaction Output) or \textit{account}/\textit{balance} transaction model. For example, Elastico \cite{Luu2016Elastico} exploits the UTXO model, while Monoxide \cite{Wang2019Monoxide} adopts the account/balance transaction model. 

\begin{figure}[t]
\centering
\includegraphics[width=0.75\columnwidth]{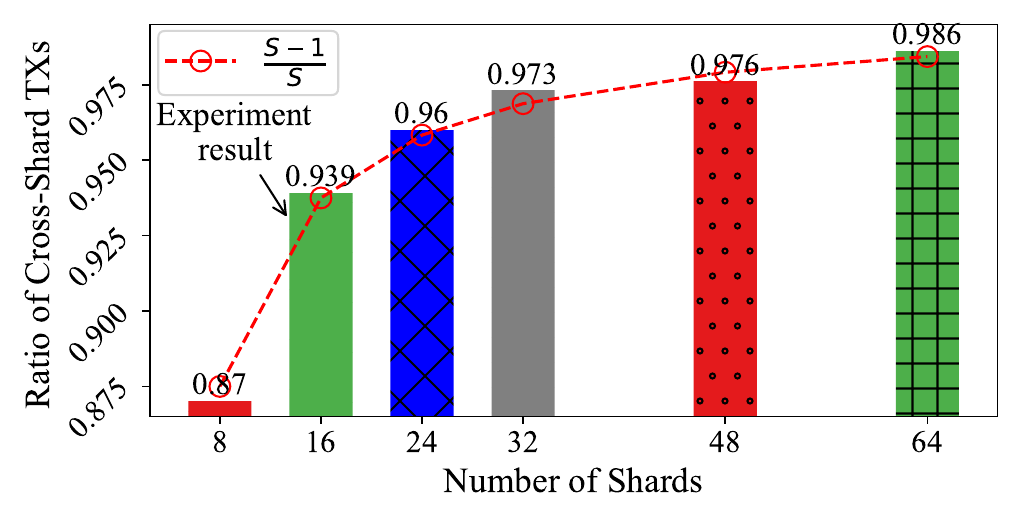}
\caption{\hw{Ratio of cross-shard transactions \textit{vs} the number of shards (i.e., $S$).}}
\label{fig:CrossShardRatio}
\end{figure}

\begin{figure}[t]
\centering
\includegraphics[width=0.8\columnwidth]{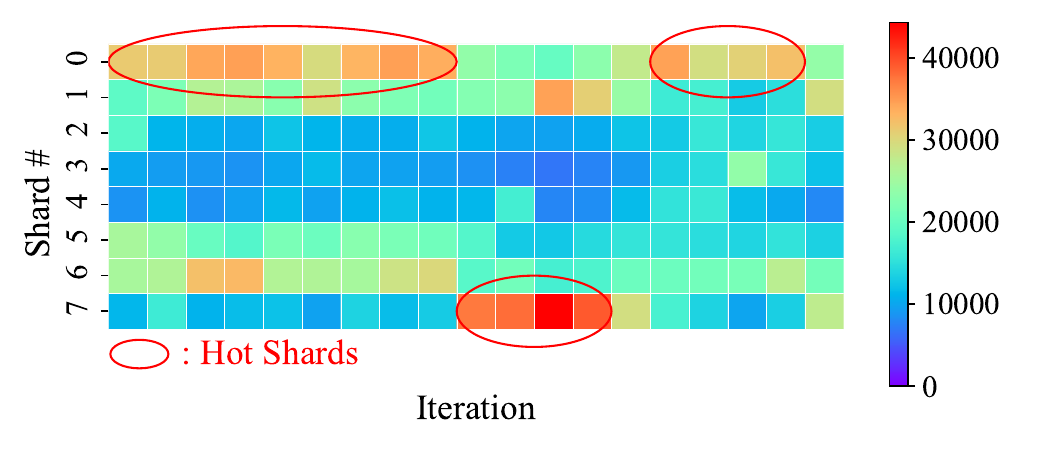}
    \caption{When feeding 80000 TXs for each iteration of blockchain consensus, Monoxide's sharding protocol \cite{Wang2019Monoxide} yields imbalanced transaction distributions among all shards. Consequently, some \textit{hot shards} are induced. Here, we call a shard congested by an overwhelming number of TXs a \textit{hot shard}.}\label{fig:Motivation}
\end{figure}

{\bf~Motivation. }\hw{The performance of existing sharding systems is mainly constrained by two factors. Firstly, the blockchain-sharding system has a high ratio of cross-shard transactions (abbr. as CTXs) when the number of shards grows. Secondly, the workloads across different shards show a severely unbalanced distribution.}

\hw{Let's explain how the ratio of CTXs associates with the number of shards. Given $S \in \mathbb{N^{+}}$ as the number of shards and assuming that accounts are randomly allocated to different shards, the probability that the payer account and the payee account of a transaction are not in the same shard is $\frac{S-1}{S}$. That is, the probability of a transaction becoming a CTX is $\frac{S-1}{S}$. To investigate the impact of the number of shards on the CTX ratio, we conduct motivating experiments using historical Ethereum transaction data.
In Fig. \ref{fig:CrossShardRatio}, the dashed line shows the theoretical analytical results between the CTX ratio and the number of shards according to random account allocation. Meanwhile, the histogram bars show the actual CTX ratios observed while invoking Monoxide \cite{Wang2019Monoxide}. Fig. \ref{fig:CrossShardRatio} shows that the actual CTX ratio closely matches the theoretical results.} When the number of shards is 32, the ratio of cross-shard transactions exceeds 97\%. In Monoxide, the atomic execution of each CTX is guaranteed through a \textit{transaction relay} mechanism \cite{Wang2019Monoxide}, in which a CTX is divided into a \textit{deduction} operation and a \textit{deposit} operation. The blockchain-sharding system processes them in the \hw{payer's shard and payee's shard of this CTX, respectively.} However, this method increases the total transaction processing workload. Moreover, as the system's CTX ratio increases, the total transaction workload also increases drastically.

In \hw{account-based state sharding}, user accounts are distributed to different blockchain shards according to the first few bits of their addresses, aiming to store all account states collaboratively.
Although this manner can improve the system throughput, CTXs are inevitably induced.
The handling of CTXs is a primary technique issue in the account/balance-based blockchain sharding system.
Monoxide \cite{Wang2019Monoxide} guarantees the atomicity of CTXs by introducing relay TXs. However, relay TXs may result in congested shards. We call such shards congested by an overwhelming number of TXs the \textit{hot shards}.
In each hot shard, TXs cannot be processed in time and thus may experience long confirmation latency.
To prove our concern about Monoxide's sharding protocol, we evaluate the TX workloads of all shards by feeding 80000 TXs at each iteration of Monoxide's consensus.
Fig. \ref{fig:Motivation} shows the TX distributions of Monoxide's shards. 
We observe that hot shards widely exist in the system across multiple iterations.
These congested hot shards are caused by active accounts, which launch a large number of TXs frequently in their associated shards.
Consequently, TXs are distributed over all shards in an imbalanced way. 
The insight behind those observations is that the account-deployment method adopted by Monoxide does not consider the frequency of launching TXs by accounts.
Moreover, a great number of CTXs could be also caused by Monoxide's sharding protocol.
As a result, when the number of shards becomes large, almost all TXs are CTXs.
When the payee account of a CTX is suspended in a congested hot shard, Monoxide may incur infinite confirmation latency for the transactions residing in local shards' transaction pools.

 Furthermore, through a thorough review of the state-of-the-art studies such as Elastico \cite{Luu2016Elastico}, Omniledger \cite{Kokoris2018OmniLedger}, and RapaidChain \cite{2018RapidChain}, we find that most of the cross-shard TX mechanisms mainly focus on the UTXO transaction model. 
 Although Monoxide \cite{Wang2019Monoxide} and SkyChain \cite{2020SkyChain} design their sharding protocols by adopting the account/balance model, the hot-shard issue has not been addressed.
 In summary, we have not yet found a cross-shard TX solution that can eliminate hot shards for the account/balance-based state sharding. 
 Thus, we are motivated to devise such a new sharding protocol.

 {\bf~Challenges.} 
 When assigning a large number of TXs in the account/balance-based state sharding, a natural problem is how to reach workload balance across all shards. Recall that the imbalanced workloads in hot shards can cause large TX confirmation latency, which threatens the eventual atomicity of CTXs.
 This problem becomes even more challenging when most of the TXs are cross-shard ones \cite{Wang2019Monoxide}.
 Therefore, guaranteeing the eventual atomicity of CTXs becomes a technical challenge that prevents the state-sharding mechanism from being largely adopted in practice.

 To this end, this paper proposes a cross-shard blockchain protocol, named \textit{BrokerChain}, for account/balance-based state sharding. BrokerChain aims to reduce the number of cross-shard TXs and ensure workload balance for all blockchain shards at the same time. 
 Our study in this paper leads to the following \textbf{contributions}.
 
\begin{itemize}
 
 \item \hhw{\textbf{Originality.}} BrokerChain partitions \hhw{the account-state graph} and performs account segmentation in the protocol layer. Thus, the well-partitioned account states can be amortized by multiple shards to achieve workload balance across all shards.

 \item \hhw{\textbf{Methodology.}} To alleviate the hot-shard issue, BrokerChain includes a cross-shard TX handling mechanism, which can guarantee the \textit{duration-limited} eventual atomicity of CTXs.

 \item \hhw{\textbf{Usefulness.}} We implemented BrokerChain and deployed a prototype in Alibaba Cloud. The transaction-driven experimental results proved that BrokerChain outperforms state-of-the-art baselines in terms of throughput, TX confirmation latency, the queue size of transaction pools, and workload balance.
 
 \end{itemize}

The rest of this paper is organized as follows.
Section \ref{sec:RelatedWork} reviews state-of-the-art studies. 
Section \ref{sec:Protocol} describes the protocol design. Section \ref{sec:ctx} depicts the handling of CTXs. Section \ref{sec:safety} analyzes the security issues and other properties of BrokerChain. 
Section \ref{sec:performance} demonstrates the performance evaluation results. Finally, Section \ref{sec:Conclusion} concludes this paper.

 
\section{Preliminaries and Related Work}\label{sec:RelatedWork}


\subsection{Transaction Models}

There are two mainstream transaction models in existing blockchain systems, i.e., the UTXO model \cite{nakamoto2008bitcoin} and the account/balance model \cite{2014Ethereum}. Under the UTXO model, a TX may involve multiple inputs and outputs. When the outputs originating from previous UTXOs are imported into a TX, those UTXOs will be marked as \textit{spent} and new UTXOs will be produced by this TX.
Under the account/balance model, users may generate a TX using the deposits in their accounts. Ethereum \cite{2014Ethereum} employs the account/balance model due to its simplicity in that each TX has only one payer account and one payee account. When confirming the legitimacy of a TX, it is necessary to check whether the deposits of the payer's account are sufficient or not.


\subsection{Representative Blockchain Sharding Solutions}

A great number of sharding solutions have been proposed to improve the scalability of blockchain systems \cite{huang2021survey}. Some of these representative solutions are reviewed as follows.
Elastico \cite{Luu2016Elastico} is viewed as the first sharding-based blockchain system, in which every shard processes TXs in parallel.
Kokoris \textit{et al.} \cite{Kokoris2018OmniLedger} then present a state-sharding protocol named OmniLedger, which adopts a scalable BFT-based consensus algorithm to improve TX throughput.
Facing the high overhead of sharding reconfiguration in Elastico and OmniLedger,  Zamani \textit{et al.} \cite{2018RapidChain} propose RapidChain to solve this overhead issue.
Chainspace \cite{2018Chainspace} exploits a particular distributed atomic commit protocol to support the sharding mechanism for smart contracts.
Then, Wang \textit{et al.} \cite{Wang2019Monoxide} propose Monoxide, which realizes an account/balance-based sharding blockchain. In Monoxide, a novel relay transaction mechanism is leveraged to process cross-shard TXs.
Next, Nguyen \textit{et al.} \cite{nguyen2019optchain} present a new shard placement scheme named OptChain. This method can minimize the number of cross-shard TXs for the UTXO-based sharding.
Prism \cite{prism} achieves optimal network throughput by deconstructing the blockchain structure into atomic functionalities.
Dang \textit{et al.} \cite{TowardsScaling} present a scaling sharded blockchain that can improve the performance of consensus and shard formation. 
Recently, Tao \textit{et al.} \cite{tao2020sharding} propose a dynamic sharding system to improve the system throughput based on smart contracts. 
Huang \textit{et al.} \cite{2021MVCom} propose an online stochastic-exploration algorithm to schedule the most valuable committees for the large-scale sharding blockchain.


\subsection{Cross-shard Transaction Processing}

A sharding blockchain must consider how to handle the cross-shard TXs.
In Omniledger \cite{Kokoris2018OmniLedger}, the authors adopt a client-driven two-phase commit (2PC) mechanism with lock/unlock operations to ensure the atomicity of cross-shard TXs.
Differently, Chainspace \cite{2018Chainspace} introduces a client-driven BFT-based mechanism. 
RapidChain \cite{2018RapidChain} transfers all involved UTXOs to the same shard by sub-transactions, such that a cross-shard TX can be transformed to intra-shard TXs.
In Monoxide \cite{Wang2019Monoxide}, the deduction operations and deposit operations are separated. Relay transactions are then exploited to achieve the eventual atomicity of cross-shard TXs.
Recently, Pyramid \cite{2021Pyramid} introduces a layered sharding consensus protocol. In this solution, cross-shard TXs are handled by a special type of nodes that serve two shards at the same time.


Compared with these existing studies, the proposed BrokerChain handles cross-shard TXs by taking the advantage of broker accounts. Furthermore, BrokerChain can achieve the workload balance while distributing TXs through the fine-grained state-partition and account-segmentation mechanisms.


\begin{figure*}[t]
\centering
\includegraphics[width=1.0\textwidth]{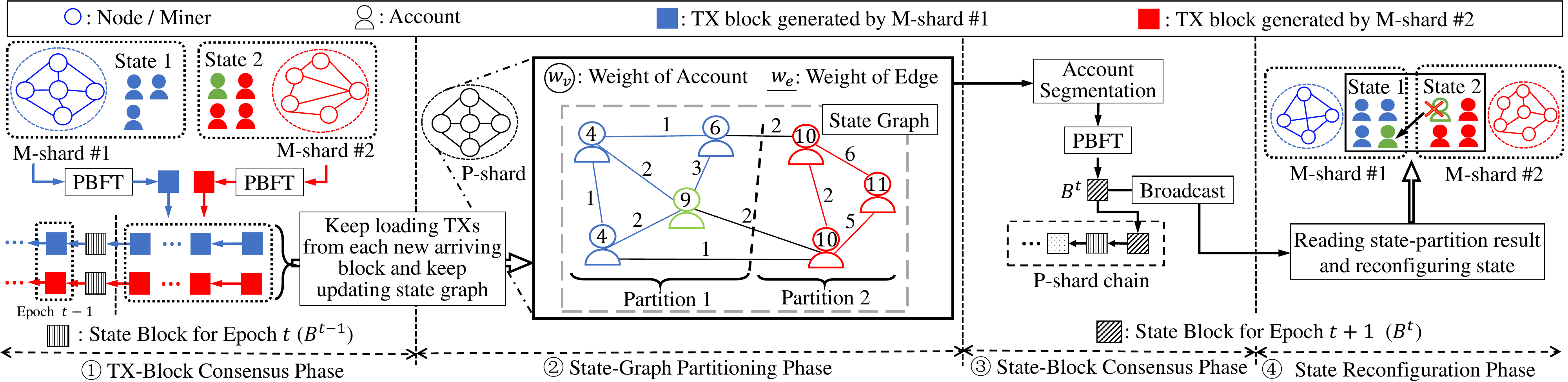}
    \caption{\hw{Four major phases of BrokerChain protocol in an epoch $t$. There are two types of shards in BrokerChain, i.e., the partition shard (P-shard) and mining shards (M-shard). Here, we only use two M-shards to illustrate how BrokerChain protocol works.}}
\label{fig:ProtocolOverview}
\vspace{-2mm}
\end{figure*}

\section{Network and Protocol Design of BrokerChain}\label{sec:Protocol}

In this section, we present BrokerChain, an adaptive state-sharding protocol based on the account/balance transaction model.

\begin{table}
    \caption{\zk{Symbols and Notations}}
    \centering
    \renewcommand{\arraystretch}{1.4}
    \begin{tabular}{|c|l|}
     \hline
   
     $T$ & \footnotesize{the set of all timeslots, $t\in T$} \\[3pt] 
      \hline
      
    $B^t$ & \footnotesize{the state block for epoch $t\in T$ at $t$} \\[3pt]
     \hline
     
     $K$ & \footnotesize{the \# of brokers, [$K$] represents the set of all brokers} \\[3pt] 
     \hline
     
     $S$ & \footnotesize{the \# of M-shards, [$S$] represents the set of all M-shards} \\[3pt] 
     \hline
     
     $\mu$ & \footnotesize{a user of the system} \\[3pt]  
      \hline
      
     $\mathbb{S}_\mu$ & \footnotesize{the account state of user $\mu$, $\mathbb{S}_\mu = \{X_\mu \mid \boldsymbol{\Psi}, \eta, \omega, \zeta \}$} \\[3pt] 
    \hline

     $X_{\mu}$ &  \footnotesize{the account address of user $\mu$} \\[3pt] 
     \hline
    
     $\boldsymbol{\Psi}$ & \footnotesize{the storage map of account, $\boldsymbol{\Psi} = [e_1, e_2, ..., e_S]$, $e_i$=0/1} \\[3pt] 
    \hline
    
      $\eta$ &  \footnotesize{the \textit{nonce} field of account state } \\[3pt] 
     \hline

     $\omega$ & \footnotesize{the \textit{value} field of account state} \\[3pt] 
     \hline
     
     $\zeta$ & \footnotesize{the \textit{code} field of account state} \\[3pt] 
     \hline

     $ \sigma_{\mu}$ & \footnotesize{the signature of user $\mu$} \\[3pt] 
     \hline

     $\Theta_{\text{raw}}$ & \footnotesize{the raw cross-shard transaction}\\[3pt] 
     \hline
     
     $\Theta_{1}$ &  \footnotesize{the first-half cross-shard TX of $\Theta_{\text{raw}}$} \\[3pt] 
     \hline
     
     $\Theta_{2}$ &  \footnotesize{the second-half cross-shard TX of $\Theta_{\text{raw}}$} \\[3pt] 
     \hline

     $H$ &  \footnotesize{ the height of block} \\[3pt] 
     \hline
     
     $N_{\text{TX}}$ &  \footnotesize{ the number of transactions played back per epoch } \\[3pt] 
     \hline
    \end{tabular}
\label{tab:TableofSymbols}
\end{table}

\subsection{Overview of BrokerChain}


Similar to Rapidchain \cite{2018RapidChain}, we first define an \textit{epoch} as a fixed length of the system running time in BrokerChain. 
To avoid Sybil attacks \cite{sybil1}, blockchain nodes should get their identities by solving a hash-based puzzle before joining in a network shard. In such puzzle, an unpredictable common randomness is created at the end of the previous epoch. Once a blockchain node solves the puzzle successfully, the last few bits of the solution indicate which shard the node should be designated to.
In our design, BrokerChain consists of two types of shards. 
\begin{itemize}

     \item {\bf M-shard.} A mining shard (shorten as \textit{M-shard}) generates TX blocks by packing TXs and achieves the intra-shard consensus at the beginning of each epoch.

     \item {\bf P-shard.} A partition shard (shorten as \textit{P-shard}) is devised to partition account states in an adaptive manner during each epoch.
    
\end{itemize}

\begin{figure}[t]
\centering
\includegraphics[width=0.95\columnwidth]{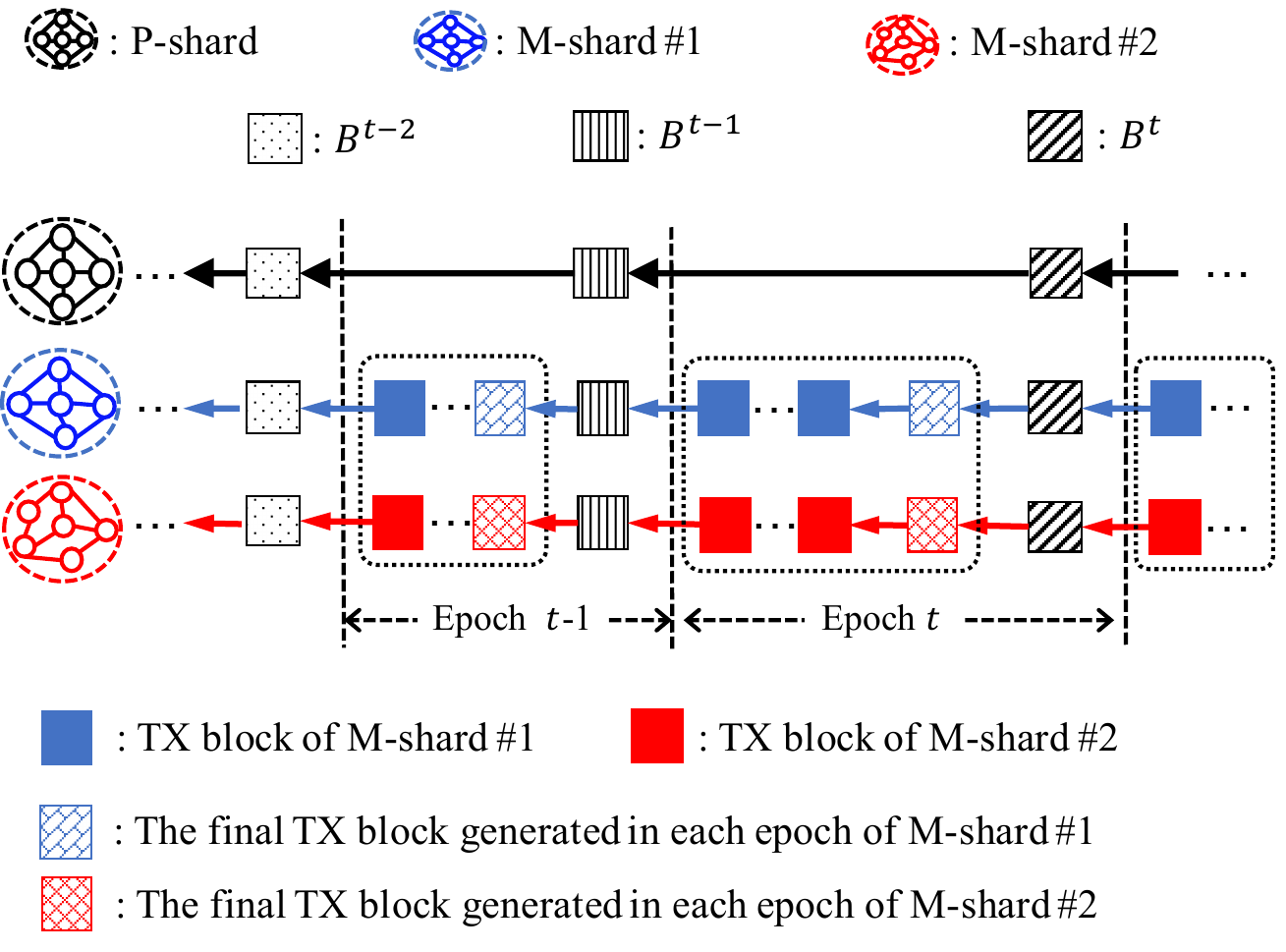}
\vspace{-1mm}
    \caption{\hw{Illustration of how blocks connect between epochs (e.g., $B^{t-1}$ is a state block generated in epoch $t$-1 but followed by the first transaction blocks of M-shards in epoch $t$).}}
\label{fig:SystemChainStructure}
\vspace{-2mm}
\end{figure}

    BrokerChain requires a number $S\in \mathbb{N^+}$ of M-shards and one P-shard existing in the sharding blockchain. For both M-shard and P-shard, we adopt PBFT protocol \cite{castro1999practical}  to achieve their respective intra-shard consensus and to avoid blockchain forks. 
    Using Fig. \ref{fig:ProtocolOverview}, we now describe the most important four sequential phases of BrokerChain as follows.
  \begin{itemize}
     \item {\bf Phase \circled{1}: TX-Block Consensus.} At the beginning of an epoch, each M-shard packages TXs from the TX pool, and generates several TX blocks by running PBFT protocol. The number of TX blocks generated by an M-shard in an epoch depends on the intra-shard network parameters such as network bandwidth.
     Note that, the first new TX block generated by an M-shard follows the \textit{state block} (i.e., $B^{t-1}$ shown in Fig. \ref{fig:ProtocolOverview}), which records the state-partition results of the previous epoch. \hw{Fig. \ref{fig:SystemChainStructure} illustrates how state blocks connect transaction blocks.}
     
     \item {\bf Phase \circled{2}: State-Graph Partitioning.} P-shard keeps loading the TXs from the new arriving blocks generated by M-shards, and keeps updating the state graph of all accounts. Once all TX blocks have been generated by M-shards within the current epoch, the state graph is fixed. 
     Then, the protocol begins to partition the state graph, aiming to achieve workload balance among all shards.
     
     \item {\bf Phase \circled{3}: State-Block Consensus}. With the state graph of all accounts partitioned in the previous phase, BrokerChain performs the \textit{account segmentation}, which is described in detail in Section \ref{sec:accountsegentation}.
     To reach a consensus towards the result of both state-graph partition and account segmentation, PBFT protocol is exploited again to generate a \textit{state block} (denoted by $B^t$), which is then added to the P-shard chain.
     
     \item {\bf Phase \circled{4}: State Reconfiguration}. 
     After reaching a consensus on the partition result, the P-shard broadcasts the state block $B^t$, which contains the $S$-partitioned new state graph, to all the associated M-shards.
     When receiving a state block $B^t$, an M-shard reads the state-partition result from the state block, and reconfigures its states accordingly such that TXs in the next epoch $t$+1 can be distributed to the designated shards according to the new account states.
     
 \end{itemize}

At the end of each epoch, BrokerChain also needs to update the formation of both P-shard and M-shards. To update those shards, the Cuckoo rule \cite{cuckoo2} is invoked such that the system can defend against the join-leave attacks \cite{join1, join2}.

 In the following, we elaborate on the most critical operations and data structures in BrokerChain, i.e., \textit{State-Graph Partition},  \textit{Account Segmentation}, and the \textit{modified Shard State Tree}.

\begin{figure}[t]
\centering
\includegraphics[width=0.9\columnwidth]{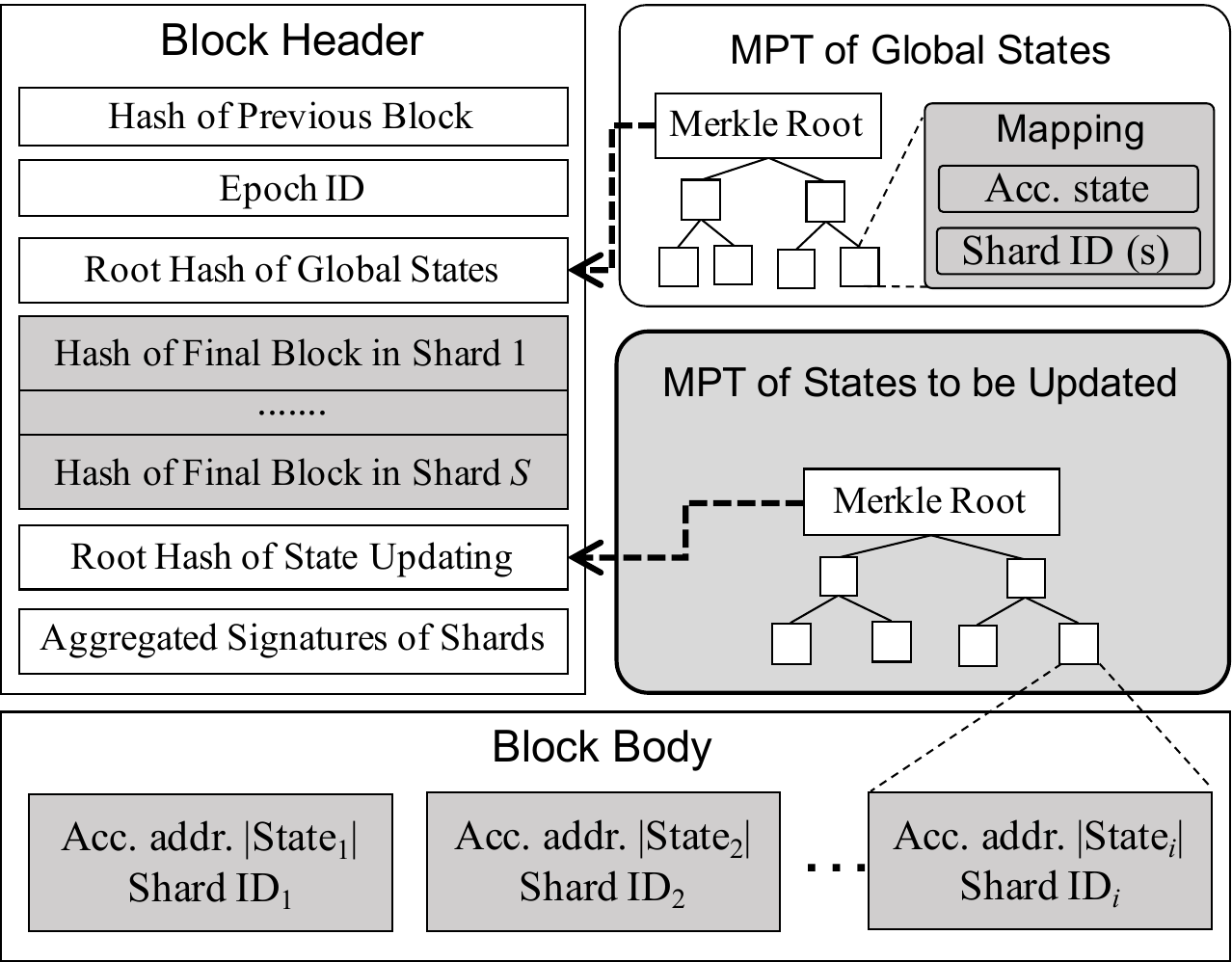}
    \caption{\hw{Data structure design of the \textit{state block}. The modules with grey background are designed particularly for state updating.}}
\label{fig:StateBlock}
\end{figure}

\subsection{State-Graph Partition}\label{sec:statereconfiguration}

\hw{
 \textbf{Updating State Graph.}} P-shard monitors the TX blocks generated by M-shards, and keeps loading the TXs in each arriving new block in real time to build/update the \textit{state graph}. When receiving a specified minimum number of TX blocks, the state graph is fixed.
 As shown in the State-Graph Partitioning Phase of Fig. \ref{fig:ProtocolOverview}, each vertex of such state graph represents an account. The edge weight (denoted by $w_e$) is defined as the number of TXs associated with the corresponding pair of accounts, while the vertex weight (denoted by $w_v$) is calculated by the sum of the associated edges' weight.
 \hw{Given the fixed state graph, P-shard begins to partition all accounts in the graph by invoking} Metis \cite{1995METIS}, which is a well-known heuristic graph-partitioning algorithm. Metis can partition the state graph into non-overlapping $S \in \mathbb{N^+}$ sectors, while reducing the number of cross-shard TXs and considering the workload balance across shards.

 \hw{\textbf{Designing the Data Structure for State Block.}} \hw{At epoch $t$, a \textit{state block} $B^{t}$ is generated by the P-shard in Phase \circled{3} \textit{State-Block Consensus}. After that, as shown in Phase \circled{4} \textit{State Reconfiguration}, $B^{t}$ shall be broadcast to all M-shards for updating account states in each M-shard. Fig. \ref{fig:StateBlock} depicts the data structure of the state block, say $B^{t}$. The block header forms a chain structure by containing the hash of the previous state block, ensuring the immutability of the P-shard chain. \hhw{The block header also records the hash of the \textit{final transaction block} produced by each M-shard in the current epoch $t$. This design enables state block $B^{t}$ to follow the final block generated by M-shard at each epoch $t$.} Global states and the state updating information of the accounts that need to be migrated are both stored using Merkle Patricia Trie (MPT). The hashes of their Merkle roots, i.e., \code{Root Hash of Global States} and \code{Root Hash of State Updating}, are also stored in the block header. Using these two root hashes, M-shards can verify and track account migration. The block body stores the mapping between account states and shard IDs. Such the design of the block body enables state updates for M-shards in the next epoch $t+1$.}

\subsection{Account Segmentation}
\label{sec:accountsegentation}

\begin{figure}[!t]
\centering
\includegraphics[width=1.0\columnwidth]{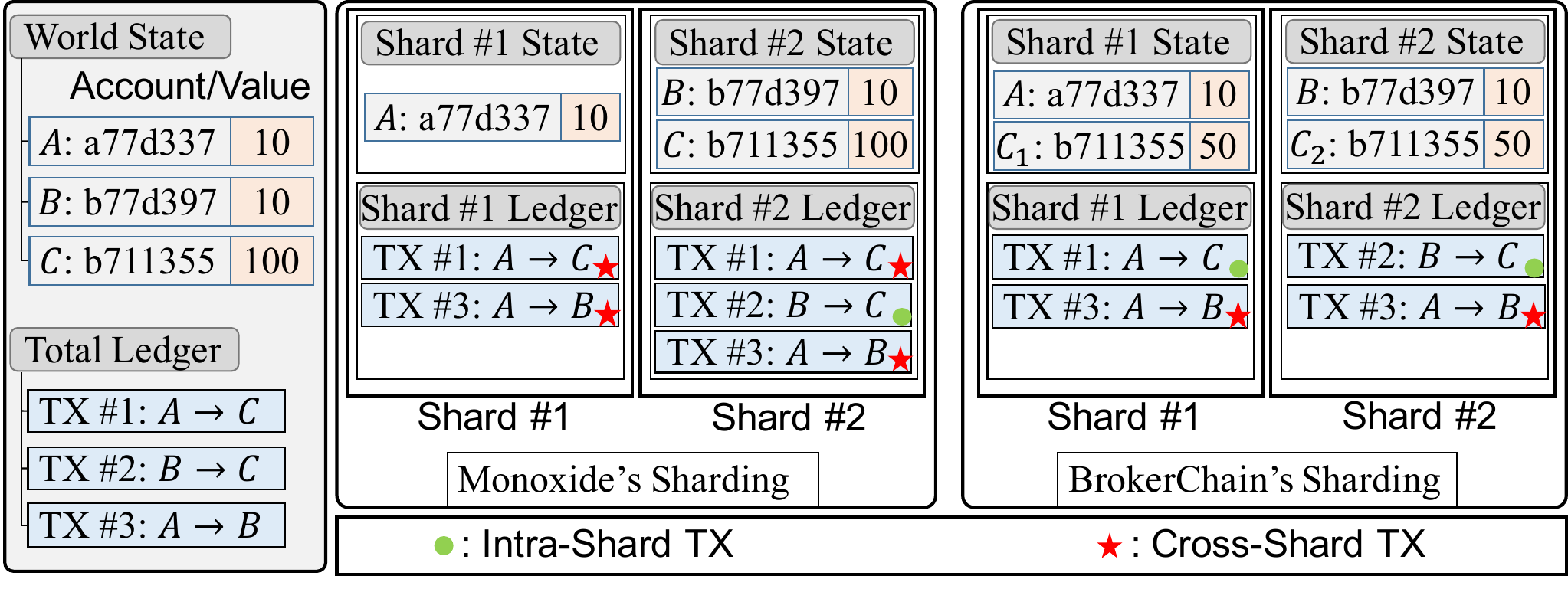}
    \caption{Comparison of account segmentation results between Monoxide and BrokerChain. We see that BrokerChain yields fewer cross-shard TXs and a more balanced TX workload across all shards.}
\label{fig:AccountSegmentation}
\end{figure}

In Monoxide \cite{Wang2019Monoxide}, an account must be stored in only one single shard. The fact is that an account, such as an exchange's active account, may participate in a great number of TXs. This kind of active account would result in hot shards inevitably.
 Facing the hot-shard issue, we give an example in the following to show that this issue could be solved at the user layer. 
 The only assumption required is that a user is allowed to hold multiple accounts, which are then permitted to store in multiple arbitrary state shards.
 Suppose that a specific user is holding multiple accounts.
 Following Monoxide's account-deployment rule, this user's accounts will be placed in different shards if the first $k\in \mathbb{N^{+}}$ bits of account addresses are different.
 Then, this user can launch a large number of TXs through his multiple accounts such that those TXs can be distributed into designated shards. Consequently, the workload balance of all shards is possible to achieve, and the number of cross-shard TXs can be reduced as well, simultaneously.

 {\bf Motivation to Account Segmentation.} Users, however, generally do not have the motivation to open multiple accounts and deposit their tokens there. Because it's inconvenient to manage those distributed multiple accounts.
 What's more, it does not make sense to force users to launch their TXs through their multiple accounts, aiming to balance the overall workload of a blockchain system.
 Therefore, we propose an \textit{account segmentation} mechanism, which works in the protocol layer of blockchain architecture.
 The proposed account segmentation mechanism has the following advantages. 
 i) Users are not forced to open multiple accounts.
 ii) A user account's state can be easily divided and stored in multiple shards. 
 iii) Through this user-transparent way, the workload balance of all shards is convenient to achieve such that hot shards can be addressed.

 {\bf Example of Account Segmentation.} To better understand the proposed account segmentation mechanism, we use Fig. \ref{fig:AccountSegmentation} to illustrate an example.
 Suppose that we have a sharding blockchain system with two shards (shard \#1 and \#2) and three accounts (denoted by $A$, $B$, and $C$). Now the ledger includes 3 original transactions \code{TX \#1:A$\rightarrow$C}, \code{TX \#2:B$\rightarrow$C}, and \code{TX \#3:A$\rightarrow$B}.
 If the sharding system has already deployed accounts $A$ and $B$ in shard \#1 and \#2, respectively. Then account $C$ can only be stored in shard \#2 (or shard \#1) under Monoxide's policy. As a result, the TXs ``\code{TX \#3:A$\rightarrow$B}'' and ``\code{TX \#1:A$\rightarrow$C}'' become 4 cross-shard TXs. 
In contrast, under BrokerChain, account $C$ (with 100 tokens) can be segmented into two smaller accounts $C_1$ and $C_2$ (50 tokens for each), which are then placed to both shards, respectively. This account segmentation result is equivalent to that account $C$ has 50 tokens stored in shard \#1, and another 50 tokens stored in shard \#2. Thus, the original TX ``\code{TX \#1:A$\rightarrow$C}'' turns to an intra-shard TXs. Furthermore, the workloads of both shards are balanced perfectly.
 From the example illustrated above, we see that the imbalanced TX workloads can be much alleviated and the number of cross-shard TXs can be reduced under BrokerChain.

{\color{black}

 %
 The difference between the proposed account segmentation and the case where users autonomously store their deposits in multiple smaller accounts is that the addresses of the account segmented by BrokerChain are all identical.
 That is, BrokerChain protocol stores the deposits of an account located at different shards with the same account address. 
 To identify an account's multiple states located at different shards, we adopt Ethereum's  \textit{counter} mechanism \cite{2014Ethereum}, which is called the \textit{nonce} of an account's state.
 
 }


\subsection{Modified Shard State Tree (mSST)} \label{sec:worldstate}

To enable the \textit{state-graph partitioning} and \textit{account segmentation} operations, BrokerChain has to enforce each shard to know the \textit{storage map} of all accounts.
Therefore, we devise a modified Shard State Tree (mSST) to store the account states based on Ethereum's \textit{world state tree} \cite{2014Ethereum}.

In contrast to the original world state tree, BrokerChain's mSST is built on top of the storage map of all accounts.
We denote such storage map as a vector $\boldsymbol{\Psi} = [e_1, e_2, ..., e_S]$ with each element $e_i$=0/1, $i \in \{1, 2, ..., S\}$, where $S\in \mathbb{N^{+}}$ is the number of M-shards.
Note that, BrokerChain needs to configure a storage map for every account.
Only when $e_i$ is equal to 1, can this specific account be viewed as being stored in shard $i$.
If there are multiple elements in  $\boldsymbol{\Psi}$ that are labeled to 1, then the associated account is segmented to the same number of accounts, which are stored in each specific shard, respectively.

\begin{figure}[t]
\centering
\includegraphics[width=1.0\columnwidth]{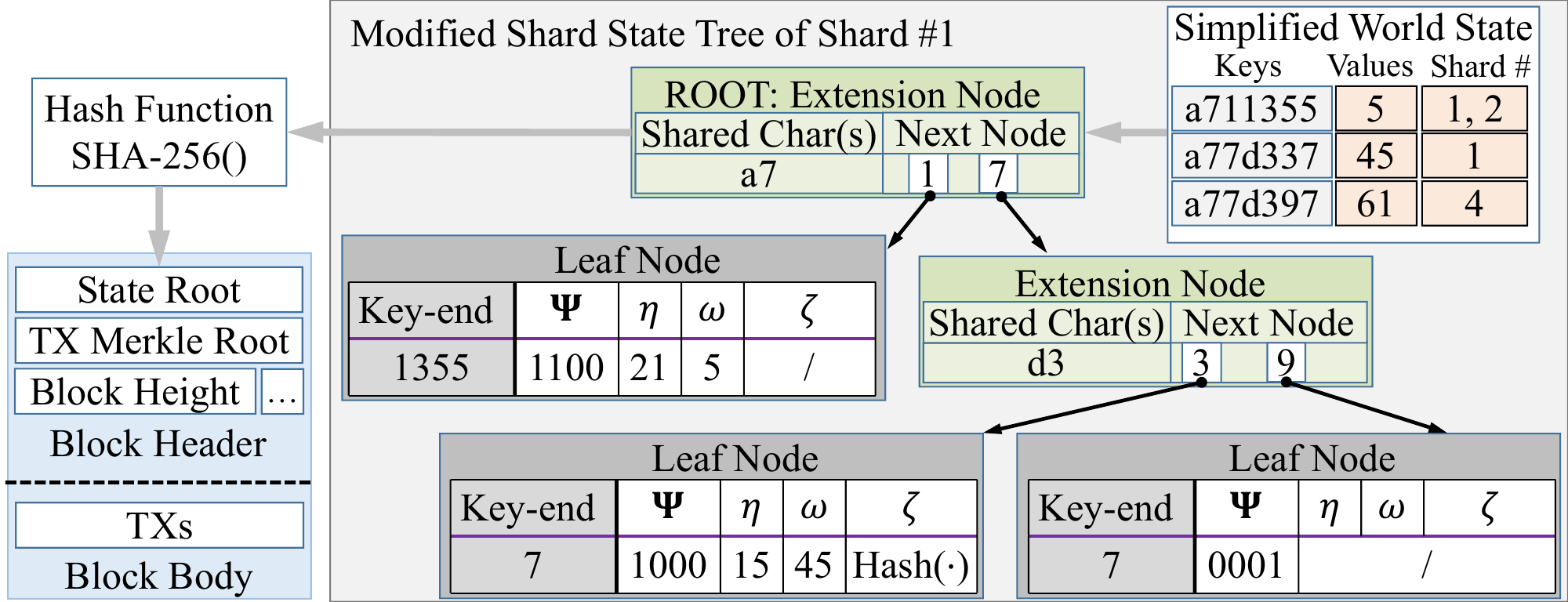}
    \caption{\hw{Data structure of the modified \textit{Shard State Tree} (mSST).}}
\label{fig:worldstate}
\end{figure}

\textbf{Data Structure of mSST.} Fig. \ref{fig:worldstate} illustrates the data structure design of mSST. 
The account state $\mathbb{S}_\mu$ of a specific user $\mu$ is represented by:
\begin{equation*}
\mathbb{S}_\mu = \{X_\mu \mid \boldsymbol{\Psi}, \eta, \omega, \zeta \},
\end{equation*}
where $X_\mu$ denotes the account address of user $\mu$, and $\eta$ denotes the \textit{nonce} field, which indicates the number of TXs sent from address $X_\mu$ or the contract creation operation generated by user $\mu$. Then, $\omega$ denotes the \textit{value} field, which shows the token deposits of the user. Finally, $\zeta$ denotes the \textit{code} field, which represents the account type.
Here, the account types include the user account and the smart contract account (the hash of the smart contract).
To a specific account, different shards maintain different mSSTs for this target account. The difference is determined by the \textit{value}, \textit{nonce} and \textit{code} fields of the local mSST. If a target account is not stored in a shard anymore, those fields corresponding to this account in the shard's mSST will be removed. And this shard needs to only update the states for the target account. Any change of the target account's state will cause the change of state root in mSST. Therefore, it is easy to maintain the consistency of an account's state in the associated shard.

\begin{figure}[t]
\centering
 \subfigure[\hw{Ethereum storage}]{
  \includegraphics[width=0.22\textwidth]{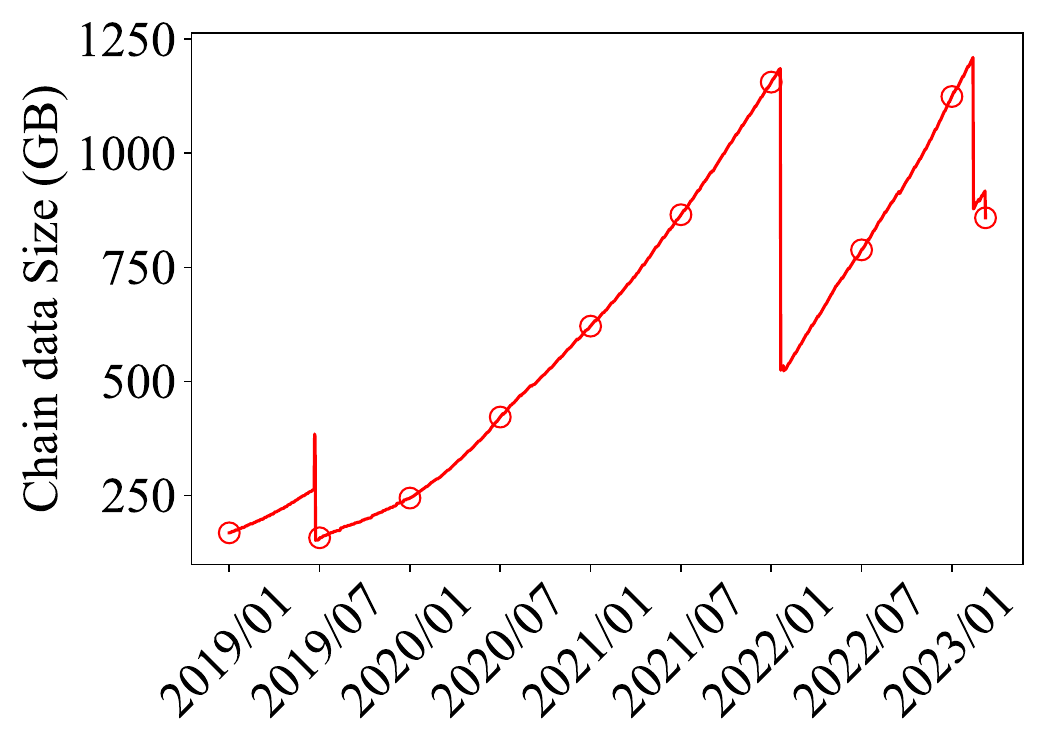}
   \label{fig:ethereumStorage}
  }%
  \hfill
\subfigure[\hw{Sharding system's storage}]{
  \includegraphics[width=0.22\textwidth]{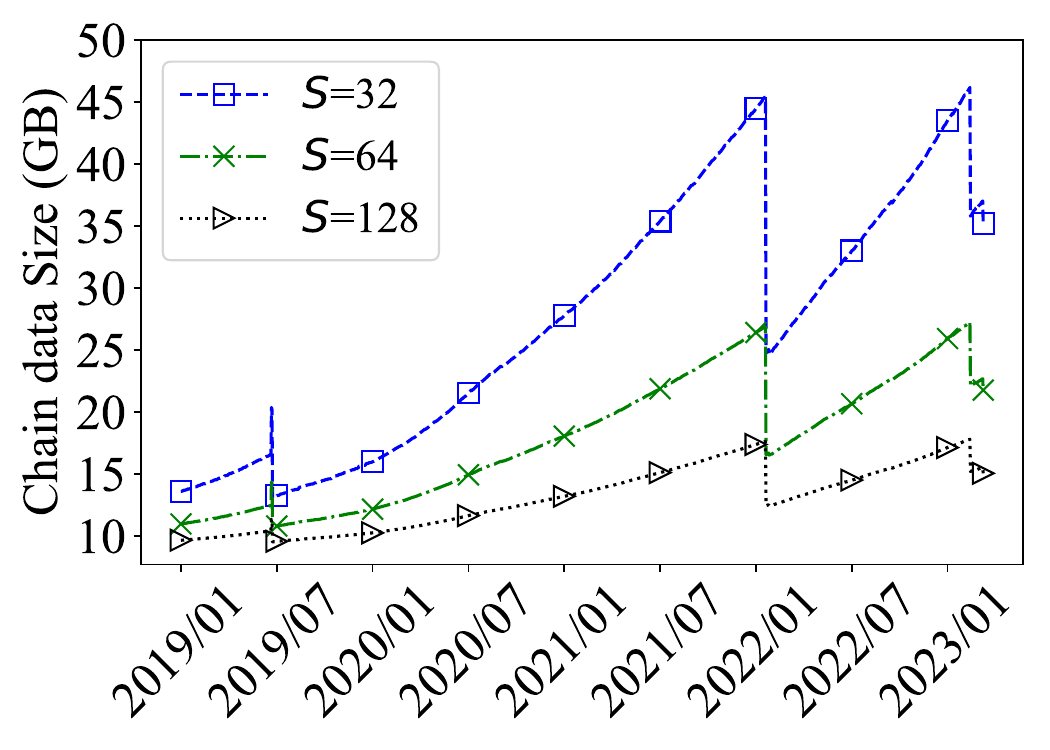}
   \label{fig:shardingStorage}
 }
\caption{\hw{Comparison of the storage overhead between Ethereum and the sharding system.}}
\label{fig:Storage}
\end{figure}

\textbf{Query Complexity.}
Through querying the \textit{storage map} in mSST, we can easily confirm whether a TX is an intra-shard TX or a cross-shard one. The querying time order is $O(1)$. Therefore, the protocol can easily get the deposits of an account by summing the \textit{value} field of all its segmented account states. The deposit-querying time order is $O(\xi)$, where $\xi$ is the number of shards where the account's segmented states are stored.

\hw{\textbf{Storage Complexity.}
By March 2023, the number of accounts in Ethereum has exceeded 224 million. The full node needs to store approximately 858 Gigabytes (GB) of on-chain data. As shown in Fig. \ref{fig:ethereumStorage}, even though Ethereum periodically prunes the world state tree, the full nodes still need to maintain high storage capacity. Thus, traditional single-chain structures cannot meet the needs of large-scale applications since ordinary nodes have to address the challenge of storage if they choose to serve as consensus nodes.
However, in a sharding system, full nodes in blockchain shards only need to store a small portion of state data compared with that of Ethereum full nodes. Fig. \ref{fig:shardingStorage} shows that the chain data size declines following the increasing number of shards. Thus, the storage complexity of the proposed protocol has the potential to serve large-scale blockchain applications.
}


\begin{figure}[t]
\centering
\includegraphics[width=0.8\columnwidth]{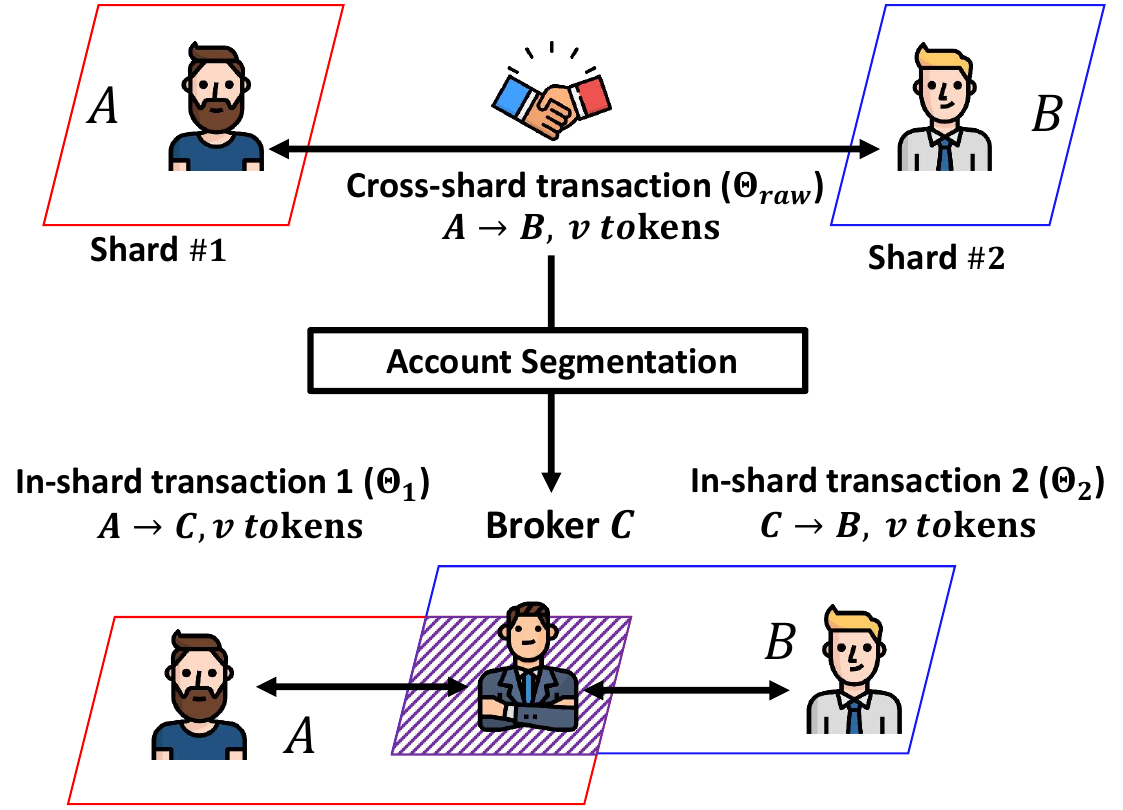}
    \caption{\hw{How a broker account turns a CTX into two intra-shard TXs.}}
\label{fig:accountsegmentation}
\end{figure}

\section{Handling the Cross-shard Transactions}\label{sec:ctx}

Although the proposed \textit{account segmentation} mechanism enables an account to be segmented into multiple smaller ones, the question is that the blockchain system needs to recruit some accounts that are willing to act as the intermediaries to help handle cross-shard TXs. We call such intermediaries the \textit{broker accounts}, which are shortened as \textit{brokers}.
In this section, we introduce how BrokerChain handles cross-shard TXs through brokers.
A fundamental requirement to become a broker is that its account has a sufficient amount of tokens.
For the system users who intend to become brokers, they could request to pledge their assets to a trusted third party or a special smart contract. Then, the BrokerChain protocol accommodates the broker eligibility for those interested system users. To stimulate a system user to play a broker’s role, an incentive mechanism is necessary. The design of such an incentive mechanism is left as our future work.

We use the accounts and the ledger shown in Fig. \ref{fig:accountsegmentation} as an example to demonstrate the handling of cross-shard TXs. 
Suppose that $C$ is a broker, whose account is segmented and stored in shards \#1 and \#2. 
Now, we have a raw CTX, i.e., payer $A$ transfers a number $v$ of tokens to the payee account $B$ via broker $C$. Here, we call shard \#1 the \textit{source shard} and shard \#2 the \textit{destination shard} of this CTX.
\hw{Through the broker $C$, the raw CTX is transformed into two intra-shard TXs, which shall be dealt with in the \textit{source shard} and \textit{destination shard}, respectively.}

\hw{When running the proposed protocol, as shown in Fig. \ref{fig:overallcrossTXmechanism}, BrokerChain processes each CTX in a either successful or failed case. Typically, a failed case is caused by large network latency while confirming the relay sub-transactions.} By exploiting Fig. \ref{fig:crossTXmechanism}, we introduce both the success and failure cases of handling cross-shard TXs.

\begin{figure}[t]
\centering
\subfigure[\hw{The overall workflow of processing a CTX.}]{
  \includegraphics[width=1\columnwidth]{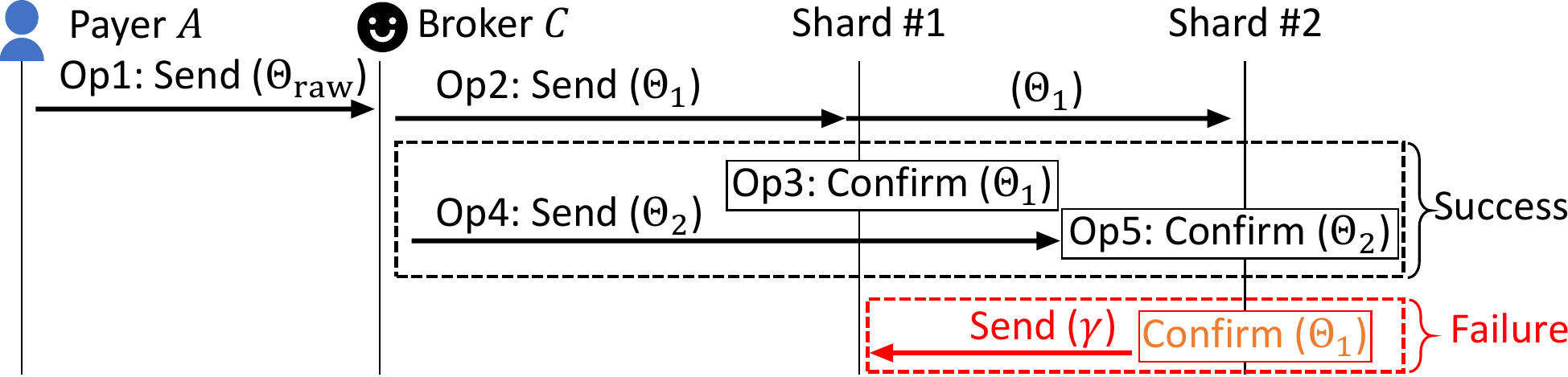}
   \label{fig:overallcrossTXmechanism}
 }
\subfigure[\hw{Successful handling of a CTX.}]{
  \includegraphics[width=1\columnwidth]{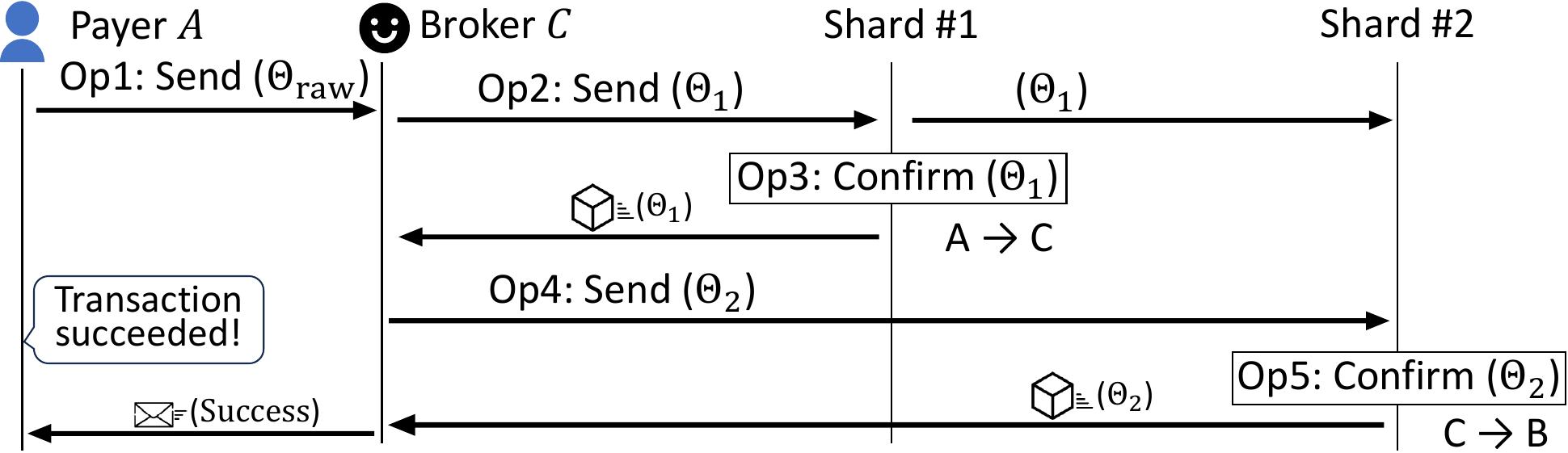}
   \label{fig:successcrossTXmechanism}
 }
 \subfigure[\hw{The handling of a CTX under the failed case.}]{
  \includegraphics[width=1\columnwidth]{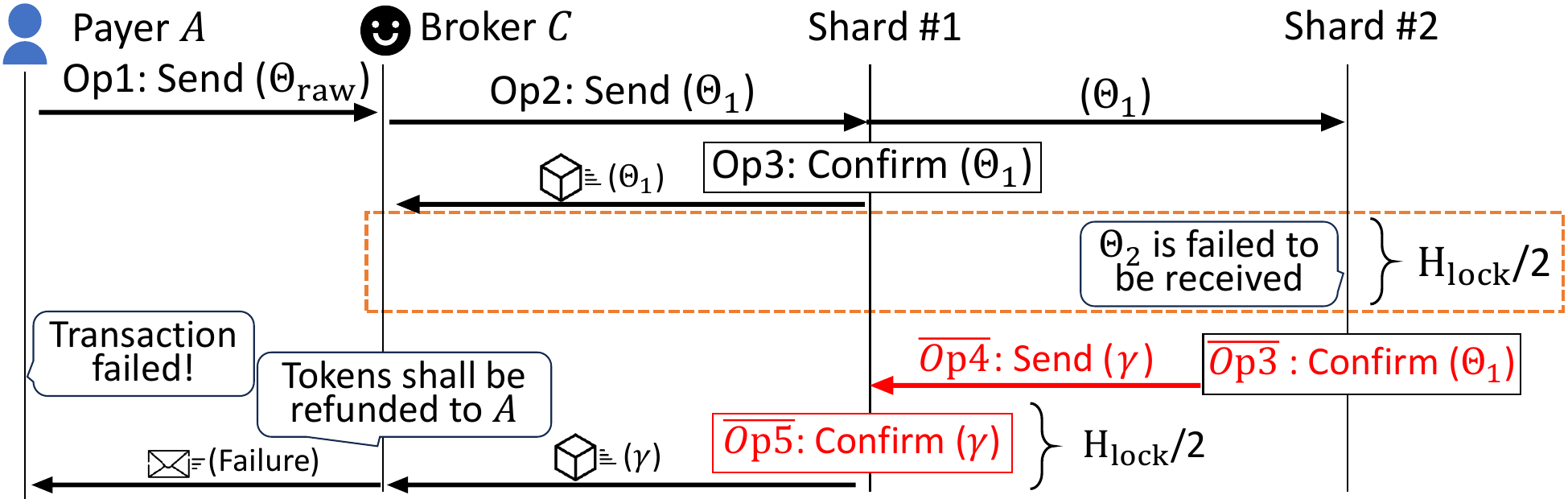}
   \label{fig:falurecrossTXmechanism}
 }
 
\caption{\hw{The success and failure cases during the handling of CTXs. A success case includes all five operations: Op1-Op5, while a failure case typically indicates that $\Theta_{2}$ is not received by shard \#2 within shard \#1's block-height interval [$H_{\text{current}}$, $H_{\text{current}}$ + $H_{\text{lock}}/2$].} }
\label{fig:crossTXmechanism}
\label{fig:32}
\end{figure}

 \subsection{Atomicity Guarantee in A Success Case}\label{sec:successCase}

 \hw{
  Fig. \ref{fig:successcrossTXmechanism} demonstrates the five operations of the successful handling of a CTX. The five operations are described as follows.}

 \begin{itemize}
 
 \item  \textbf{Op1: Create a raw transaction $\Theta_{\text{raw}}$}. Account $A$ first chooses an appropriate token-lock duration (denoted by $H_{\text{lock}}$), which is in fact the number of successive blocks counting from the block including $\Theta_{\text{raw}}$ to the one releasing locked tokens. $\Theta_{\text{raw}}$ is defined as follows.
    \begin{equation*}
    \Theta_{\text{raw}}:= \left\langle \left\langle B, v, C, H_{\text{lock}},\eta_{\text{payer}}, \eta_{\text{broker}} \right\rangle , {\sigma_{A}} \right\rangle,
     \end{equation*}
 where $\eta_{\text{broker}}$ and $\eta_{\text{payer}}$ denote the nonce of broker $C$ and the nonce of payer account $A$, respectively. Symbol $\sigma_{A}$ is the signature of payer $A$, calculated by invoking ECDSA algorithm \cite{signature}. The usage of the  token-lock duration $H_{\text{lock}}$ and nonces $\eta_{\text{payer}}, \eta_{\text{broker}}$ are explained subsequently.  
 Next, $A$ informs $\Theta_{\text{raw}}$ to the broker $C$.

 \item \textbf{Op2: Create the first-half cross-shard TX $\Theta_{1}$}. When receiving $\Theta_{\text{raw}}$, broker $C$ creates the first-half cross-shard TX, denoted by $\Theta_{1}$. Then, the current block height in shard \#1 when $\Theta_{1}$ is created is labeled as $H_{\text{current}}$.
 We then have
     \begin{equation*}
     \Theta_{1}:=\left\langle \left\langle \text{Type1}, \Theta_{\text{raw}}, H_{\text{current}} \right\rangle , \sigma_{C} \right\rangle,
     \end{equation*}
  where Type1 is an indicator of $\Theta_{1}$, and $\sigma_{C}$ is the signature of broker $C$. Next, $\Theta_{1}$ is broadcast to the blockchain network.
  When all the blockchain nodes in other shards receive $\Theta_{1}$, they execute the following steps.
 i) Validating $\Theta_{1}$ using signatures $\sigma_{A}$ and $\sigma_{C}$. 
 ii) Acquiring the payer's public key and calculating the payer's account address using signature $\sigma_{A}$. 
  iii) Through querying the mSST, shard nodes know the source and destination shards.
  iv) Using Kademlia \cite{2002Kademlia} routing protocol, $\Theta_{1}$ is forwarded to shards \#1 and \#2.
  After validating that $\eta_{\text{payer}}$ is correct and the deposits of payer $A$ are sufficient, $\Theta_{1}$ will be added to the TX pool of shard \#1, and waits to be packaged in a block.

    \item  \textbf{Op3: Confirm $\Theta_{1}$}. After broadcasting,  the blockchain nodes in shard \#1 shall include $\Theta_{1}$ in the block whose height is labeled as $H_{\text{source}}$. Apparently, $H_{\text{current}} \leq H_{\text{source}}$. Then, payer $A$ transfers $v$ tokens to broker $C$, and these tokens will be locked within the block-height interval $[H_{\text{source}}, H_{\text{source}}+H_{\text{lock}}]$. The nonce of payer $A$ will be increased by 1 to prevent the replay attack \cite{replay}.

  \item  \textbf{Op4: Create the second-half cross-shard TX $\Theta_{2}$}. When broker $C$ is acknowledged that $\Theta_{1}$ has been confirmed, $C$ then creates the second-half cross-shard TX $\Theta_{2}$: 
       \begin{equation*}
        \Theta_{2}:= \left\langle \left\langle \text{Type2}, \Theta_{\text{raw}} \right\rangle , \sigma_{C} \right\rangle,
       \end{equation*}
       where Type2 is a label of $\Theta_{2}$.
 Next, $\Theta_{2}$ will be broadcast to the blockchain network.  $\Theta_{2}$ is finally routed to the destination shard, i.e., shard \#2.  When $\eta_{\text{broker}}$ is verified qualified and the deposits of broker $C$ in the shard \#2 are sufficient, $\Theta_{2}$ will be added to the TX pool of shard \#2.

 \item  \textbf{Op5: Confirm $\Theta_{2}$}. When $\Theta_{2}$ is received by the blockchain nodes in shard \#2,  it will be packaged in a new block if the current block height of the shard \#1 is smaller than $H_{\text{current}}$ + $H_{\text{lock}}/2$. The account state in shard \#2 is then updated when the following token transfer is executed: broker $C$ transfers a number $v$ of tokens to payee account $B$. Broker $C$'s nonce is then increased by 1 to prevent the replay attack \cite{replay}.

 \end{itemize}

 In the \textit{success} case of CTX handling, the payee account $B$ can get a number $v$ of tokens when $\Theta_{2}$ is successfully included in the destination shard's block. On the other hand,  only when the block height of shard \#1 is greater than $H_{\text{source}}$+$H_{\text{lock}}$, can broker $C$ receive a number of $v$ refunded tokens. 

 \subsection{Atomicity Guarantee in A Failure Case}\label{sec:failureCase}

 According to our design shown in Fig. \ref{fig:worldstate}, all shards mutually share their block headers where the essential block height information is included. 
 Thus, if the blockchain nodes in shard \#2 fail to receive $\Theta_{2}$ when the current block height of shard \#1 exceeds $H_{\text{current}}$ + $H_{\text{lock}}/2$, we say that $\Theta_{2}$ is failed to be acknowledged by  shard \#2. This failure result is very possibly induced by a malicious broker who intends to embezzle the payer account's locked tokens.

 As shown in Fig. \ref{fig:falurecrossTXmechanism}, in order to process such a failure case, BrokerChain executes the following steps to guarantee the atomicity of CTXs.

 \begin{itemize}
     \item \hw{\textbf{$\overline{\textbf{Op3}}$: Confirm $\Theta_{1}$}.} When $\Theta_{2}$ is failed, $\Theta_{1}$ will be included in a block located at shard \#2. Broker $C$'s nonce will be attempted to update in shard \#2.

     \item \hw{\textbf{$\overline{\textbf{Op4}}$: Send $\gamma$}.} When $\Theta_{1}$ is confirmed, blockchain nodes in shard \#2 send the following confirmation of $\Theta_{1}$ to  shard \#1 as the \textit{failure proof} (denoted by $\gamma$):
        \begin{equation*}
        \gamma := \left\langle\Theta_{1}, \textit{dest}, H_{\text{dest}}, \{P_{\text{dest}}\}\right\rangle,
       \end{equation*}
    where \textit{dest} represents the index of the destination shard, $H_\text{dest}$ denotes the height of destination shard's block where $\Theta_{1}$ is included, and $\{P_{\text{dest}}\}$ refers to the Merkle tree path \cite{nakamoto2008bitcoin}. $\{P_{\text{dest}}\}$ consists of the hash values of all the associated Merkle tree nodes along the path originating from Merkle tree root and terminating at the entry hash node corresponding to $\Theta_{1}$. The Merkle tree path $\{P_{\text{dest}}\}$ is used to verify that $\Theta_1$ is packaged in a block located at shard \#2.

    \item \hw{\textbf{$\overline{\textbf{Op5}}$: Confirm $\gamma$}.} Once the nodes in shard \#1 receive the failure proof $\gamma$ and verify the correctness of path $\{P_{\text{dest}}\}$, shard \#1 is acknowledged that $\Theta_{1}$ has been included in a block located at  shard \#2 but $\Theta_{2}$ was not. Then, $\gamma$ will be included in the source shard's block whose block height is smaller than $H_{\text{source}}$ + $H_{\text{lock}}$. Finally, the locked tokens in Op3  will be refunded to the payer account $A$.
 \end{itemize}

 %

No matter whether a CTX is handled in a success or failure case, only one of the two created TXs $\Theta_{1}$ or $\Theta_{2}$ is allowed to be included in a block at the destination shard. The failure-proof $\gamma$ could be possibly lost by accident during broadcasting. Once $\gamma$ is found lost, it can be reconstructed and launched again by the destination shard.

\section{Property Analysis of BrokerChain}\label{sec:safety}

\subsection{Tackling Double-Spending Attacks}

BrokerChain treats the payee account of a cross-shard TX benign. Thus, in the following, we analyze how BrokerChain tackles typical security threats brought by the payer and broker.

 \textbf{Threat-1: Sender $A$ is malicious.} A malicious payer account may launch double-spending attacks.
 In such attacks, payer $A$ may create a double-spending TX in the source shard along with the raw TX $\Theta_{\text{raw}}$ denoted by $\Theta_{A}:= \left\langle A \rightarrow A', \eta_A\right\rangle$, where $A'$ is payer $A$'s another account and $\eta_A$ is the nonce of $\Theta_{A}$. 
 Note that, to achieve double spending, $\eta_A$ must be exactly same with $\eta_\text{payer}$ signed in the firsts-half cross-shard TX $\Theta_{1}$. 
 Although the attack strategy described above seems feasible, it is impossible to implement under BrokerChain.
 Taking the advantage of the \textit{counter} mechanism aforementioned, if $\Theta_{A}$ is confirmed by the source shard and $\eta_A$ is updated, then $\Theta_{1}$ will be failed to be confirmed by the source shard because its $\Theta_{1}$'s nonce $\eta_{\text{payer}}$ is conflict to $\eta_A$. Only when broker $C$ finds that $\Theta_{1}$ has been confirmed by the source shard, $C$ begins to create $\Theta_{2}$. Through this way, broker $C$ can defend against double-spending attacks launched by the malicious payer account.

\textbf{Threat-2: Broker $C$ is malicious.} A malicious broker $C$ may create the following double-spending TX in the destination shard exactly when $\Theta_{1}$ is confirmed in the source shard: $\Theta_{C}:= \left\langle C \rightarrow C', \eta_C \right\rangle$, where $C'$ is broker's another account and $\eta_C$ is exactly same with $\eta_{\text{broker}}$ signed in TX $\Theta_{2}$. 
Similarly, if $\Theta_{C}$ is confirmed by the destination shard and the nonce of broker $C$ is updated, then $\Theta_{2}$ will not be confirmed by the destination shard because $\eta_{\text{broker}}$ is conflict. 
By monitoring the update of nonce $\eta_{\text{broker}}$, BrokerChain can identify whether a broker is malicious or not.
To prevent broker $C$ from embezzling payer $A$'s pledged tokens, BrokerChain enforces the handling routine of such cross-shard TX to fall into the failure case. 


In conclusion, with the proposed asset pledge and token-lock mechanisms, BrokerChain can ensure the atomicity of cross-shard TXs under the threats of double-spending attacks.

\subsection{Recommended Setting for Token-Lock Duration}

Recall that in Op1 of cross-shard TX's handling described in Section \ref{sec:successCase}, a payer needs to first choose an appropriate token-lock duration towards creating a raw transaction.
The question is how to set such token-lock duration for the raw transaction.
A feasible approach is to set the token-lock duration according to the system throughput observed currently. An empirical recommended setting is to ensure at least 20 times of the average latency of handling the cross-shard TXs.

\subsection{Analysis of Cross-shard Confirmation}

As shown in Fig. \ref{fig:crossTXmechanism}, the cross-shard verifications occurred during the handling of cross-shard TXs include the following two cases: i) under the success case, broker $C$ has to verify that whether $\Theta_1$ is included in the source shard's chain, and ii) under the failure case, the source shard's nodes need to verify the \textit{failure proof} $\gamma$ sent by the destination shard. 
Thus, in the success case, the computing complexity of cross-shard verification is $O(1)$. In the failure case, the computing complexity of cross-shard verification is $O(n)$, where $n$ is the number of blockchain nodes in the source shard. 
Therefore, BrokerChain can lower such cross-shard verification overhead by shifting the on-chain verification to the off-chain manner taking advantages of the broker's role.

Regarding the cross-shard TX's confirmation latency, it is theoretically longer than twice intra-shard TX's confirmation latency. This is because a cross-shard TX must be confirmed at the source and destination shards in a sequential order.
In contrast, BrokerChain can lower the confirmation latency based on broker's reputation.
If a broker is completely trusted by the sharding blockchain, both $\Theta_1$ and $\Theta_2$ can be allowed to execute simultaneously. Thus, their confirmation latency can approach to that of an intra-shard TX technically.

\subsection{Duration-Limited Eventual Atomicity}

Monoxide \cite{Wang2019Monoxide} guarantees the \textit{eventual atomicity} of cross-shard TXs. However, Monoxide does not specify the confirmation time of relay TXs. Thus, the cross-shard TX's handling in Monoxide may take a long time. In contrast, BrokerChain can ensure each cross-shard TX is done within the predefined token-lock duration $H_{\text{lock}}$. When the cross-shard TX is successfully processed, the payer account's pledged tokens will be received by the payee account. Otherwise, payer's pledged tokens will be refunded by the broker. Therefore, we claim that BrokerChain ensures the \textit{duration-limited} eventual atomicity for cross-shard TXs.

\begin{figure}[t]
\centering
 \subfigure[\hw{Failure Probability ($\hat{P}$) \textit{vs} P-shard Size ($m$)}]{
  \includegraphics[width=0.22\textwidth]{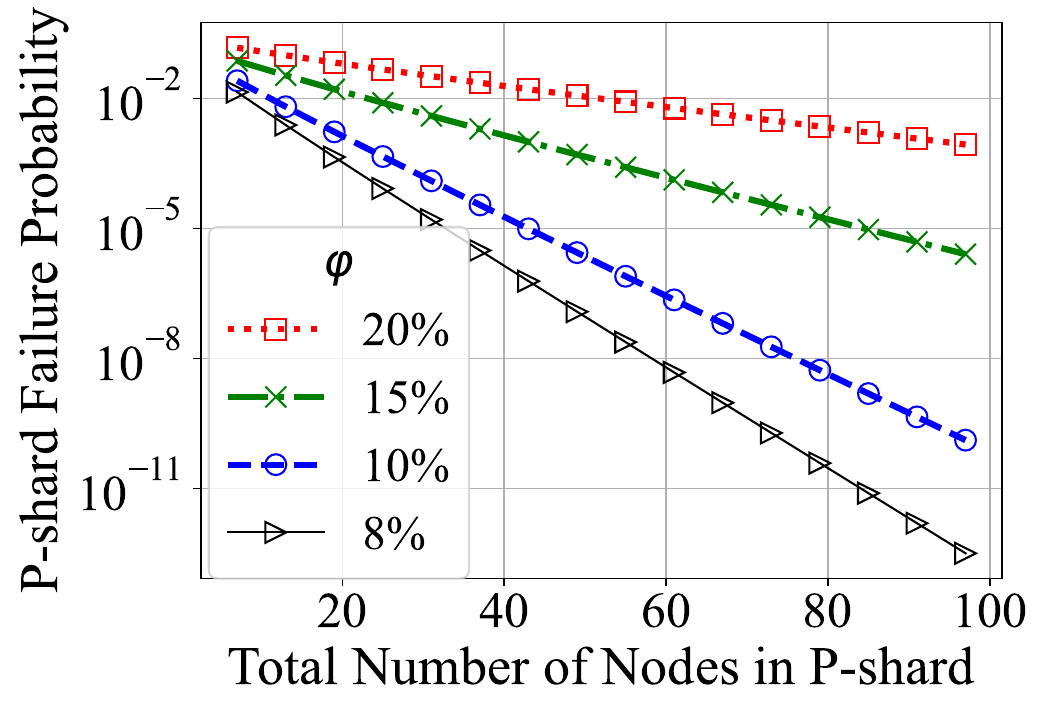}
   \label{fig:shardFailurePorb}
  }%
  \hfill
\subfigure[\hw{Failure Probability ($\hat{P}$) \textit{vs} Malicious Computational Power ($\varphi$)}]{
  \includegraphics[width=0.22\textwidth]{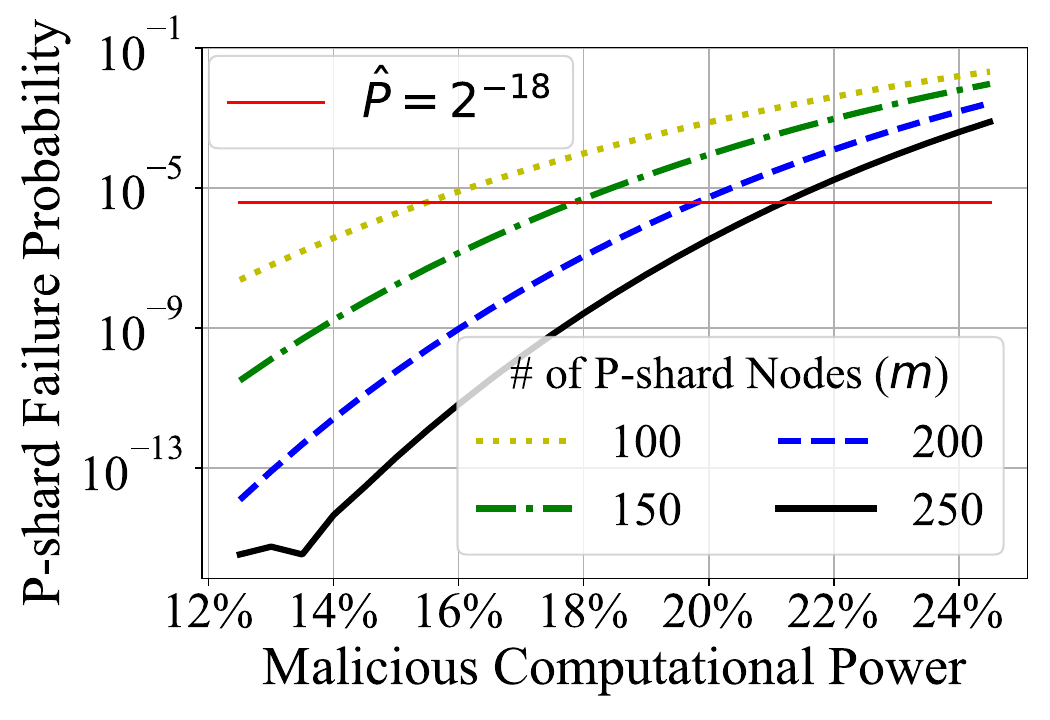}
   \label{fig:malicousNodeWithProb}
 }
\caption{\hw{Theoretical failure probability of the \textit{P-shard}.}}
\label{fig:ProbabilityOfPshard}
\end{figure}

\section{Security Analysis of Shards}\label{sec:safetyOfShard}

\subsection{\hw{Security Analysis for P-shard}}

\hw{
 Let \hhw{real number-valued symbols} $\alpha \in (0, 1)$ denote the proportion of honest nodes' computational power in the P-shard, $\varphi \in (0, 1)$ denote the proportion of malicious nodes' computational power in the overall system, and $\upsilon \in (0, 1)$ denote the probability that a node is malicious in the P-shard.
}

\begin{theorem}\label{theorem1}
\hw{Supposed that blockchain nodes can freely choose which shard to join, we have $\upsilon=\frac{\varphi}{\varphi+\alpha\cdot (1-\varphi )}$,} \hhw{where $\varphi+\alpha\cdot (1-\varphi )>0$.}
\end{theorem}

\begin{proof}

\hw{Nodes join the sharding system by solving hash puzzles, where computational power dominates the security of a shard. 
\hhw{To help better understand our calculation, we define the \textit{relative} proportion of computational power of either malicious or honest nodes as their \textit{\textbf{credits}} while constructing shards.}
Because nodes can choose the type of shards to join, malicious nodes could concentrate their computational power to attack the P-shard.
We already know that $\varphi$ is the \textit{credits} of malicious nodes while constructing the P-shard. Thus, the \textit{credits} of honest nodes in the P-shard is $\alpha\cdot\left(1-\varphi\right)$. Thereby, the equation of $\upsilon=\frac{\varphi}{\varphi+\alpha\cdot\left(1-\varphi\right)}$ holds. 
}
\zk{The proof of Theorem \ref{theorem1} concludes.}
\end{proof}

\begin{remark}
\hhw{Theorem \ref{theorem1} tells us that when the malicious nodes' computational power $\varphi$ is fixed, the risk of suffering from computational-power attacks in the P-shard increases if $\alpha$ declines.
When $\alpha$ becomes smaller, the effect of malicious nodes' relative computational power (i.e., the \textit{credits} of malicious nodes) will be amplified in the P-shard.}
\hw{For example, by fixing $\varphi$=20\%, when $\alpha$=0.9, we have $\upsilon$=21.7\%; but when $\alpha$=0.5, we have $\upsilon$=33\%.  Therefore, the system should not allow nodes to join the P-shard freely.}
\end{remark}

\hhw{We then let $m \in \mathbb{N^{+}}$ denote the number of nodes in P-shard, $ \vartheta \in \mathbb{N^{+}}$ denote the number of nodes in each M-shard, $\kappa \in  (0, 1)$ denote the probability that a node is assigned to P-shard, and  $\hat{P} \in (0, 1)$ denote the probability of the P-shard failure. Here, we call a \textit{\textbf{P-shard failure}} occurs if the P-shard is cracked by malicious nodes. Via the following Theorem \ref{theorem2}, we calculate the probability of the P-shard failure.}

\begin{theorem}\label{theorem2}
\hhw{If the system randomly assigns nodes to all shards, we have $\upsilon=\varphi$. This indicates that the computational power of malicious nodes will not be amplified. The probability of the P-shard failure is then calculated as follows:}  \begin{equation}\label{equ:pshardfail}
    \hat{P}=\sum\limits_{i=\lfloor(m-1)/3\rfloor+1}^{m}\begin{pmatrix}m\\i\end{pmatrix}\begin{pmatrix}\upsilon\end{pmatrix}^{i}\cdot\begin{pmatrix}1-\upsilon\end{pmatrix}^{m-i}. 
    \end{equation}
\end{theorem}

\begin{proof}
\hw{Since shard nodes are randomly assigned, the probability that malicious nodes \hhw{join} P-shard depends solely on the proportion of malicious nodes out of the total number of nodes. The total number of nodes is $m+\vartheta \cdot S$, thus $\kappa=\frac{m}{m+ \vartheta \cdot S}$, where $m$ and $\vartheta$ are the numbers of nodes in the P-shard and each M-shard, respectively. The \hhw{\textit{credits}} of malicious nodes in P-shard is $\varphi \cdot \kappa$ and the \hhw{\textit{credits}} of honest nodes in the P-shard is $(1-\varphi)\cdot \kappa$. Therefore, 
the equation of  $\upsilon=\varphi$ holds, i.e., the impact of malicious nodes' computational power is not amplified. 
The intra-shard consensus protocol PBFT can tolerate at most 1/3 Byzantine nodes in each shard. When nodes are randomly \hhw{assigned} to different shards, the probability of the P-shard failure $\hat{P}$ is written as Eq. \eqref{equ:pshardfail}.
\zk{The proof of Theorem \ref{theorem2} concludes.}
}

\end{proof}

\begin{remark}

\hhw{Theorems \ref{theorem1} and Theorems \ref{theorem2} imply that the system should enforce nodes randomly join different shards. According to $\hat{P}$, the system should keep an appropriate number of nodes in the P-shard to make the probability of the P-shard failure lower than a certain predefined threshold, e.g., $\hat{P} \leq 2 ^ {- \lambda}$, where $\lambda \in \mathbb{N^{+}}$ is the parameter that describes the security level of the sharding system. When $\lambda=18$, the probability of the P-shard failure is not greater than $2^{-18}$$\approx$$3.8\times10^{-6}$, which means that a P-shard failure is expected to occur once every $2.6\times 10^{5}$ times of shard reconfiguration. When the number of the P-shard consists of 250 nodes (i.e., $m=250$), the P-shard can tolerate malicious nodes whose total computational power reaches 21\% at the most.}

\hhw{Fig. \ref{fig:shardFailurePorb} shows that the probability of the P-shard failure drops following the number of nodes in the P-shard while changing  $\varphi$ from 20\% to 8\%. A smaller $\varphi$ results in a lower probability of the P-shard failure. 
Fig. \ref{fig:malicousNodeWithProb} shows that the probability of P-shard failure increases when the malicious nodes' computational power increases. Note that, the red line indicates the security threshold when $\lambda=18$. While changing the size of the P-shard from 100 to 250, we observe that a large number of nodes can alleviate the P-shard failure given a specific $\varphi$.}
\end{remark}

\begin{figure}[t]
\centering
 \subfigure[\hw{Failure Probability ($\overline{P}$) \textit{vs} M-shard Size ($\vartheta \cdot S$)}]{
  \includegraphics[width=0.22\textwidth]{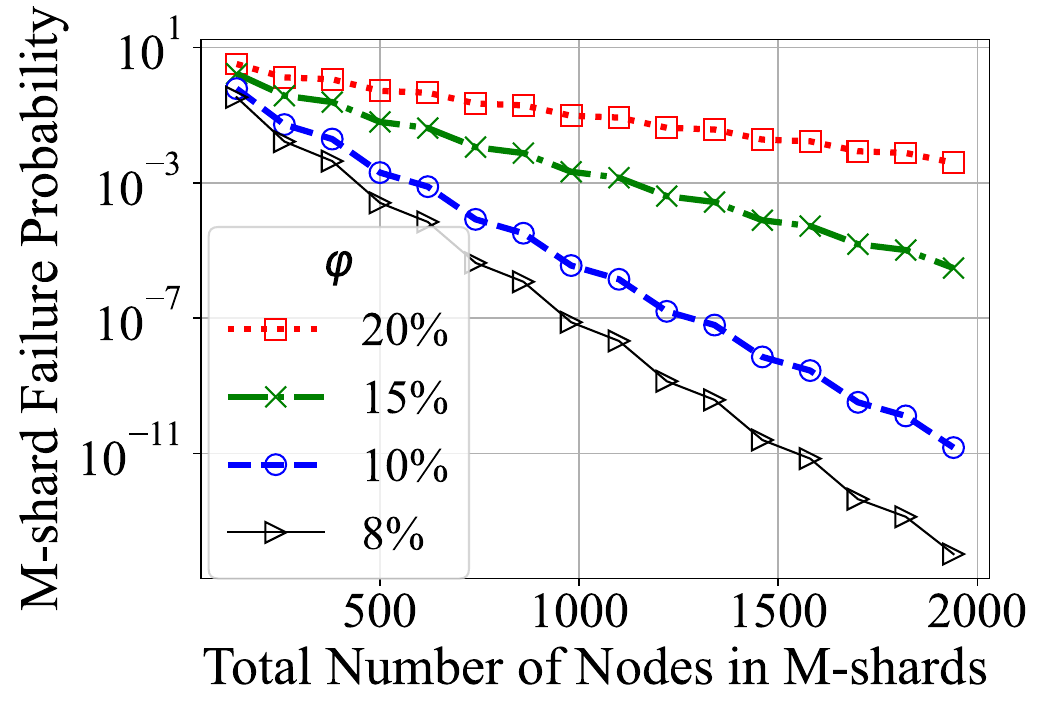}
   \label{fig:workshardSizwWithPorb}
  }%
  \hfill
\subfigure[\hw{Failure Probability ($\overline{P}$) \textit{vs} Malicious Computational Power ($\varphi$)}]{
  \includegraphics[width=0.22\textwidth]{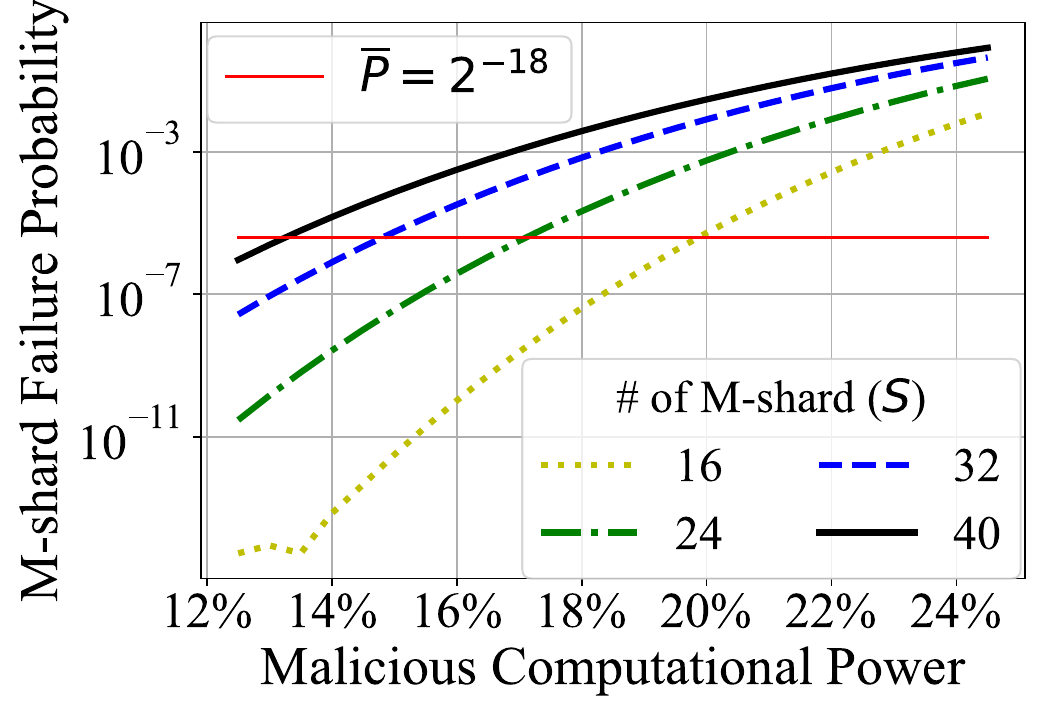}
   \label{fig:workShardMalicousNodeWithProb}
 }
\caption{\hw{Theoretical failure probability of the \textit{M-shard}.}}
\label{fig:ProbabilityOfMshard}
\end{figure}

\subsection{\hw{Security Analysis for M-shard}}

\hhw{Let \hhw{the real number-valued symbol} $\delta  \in (0,1)$ denote the probability \hhw{that an M-shard node is malicious}; \hhw{the integer variable} $N \in \mathbb{N^{+}}$ denote the total number of malicious nodes in all M-shards; \hhw{the symbol} \code{shard i} ($i \in \mathbb{N}^+$) as the M-shard whose ID is \hhw{$i \in [S]$; \hhw{the integer} $n_{i} \in \mathbb{N^{+}} (i \in [S])$} denote the number of malicious nodes in \code{shard i}; \hhw{and the integer} $\beta \in \mathbb{N^{+}}$ denote the threshold for the number of malicious nodes in each M-shard. \hw{If $n_{i} $ exceeds threshold $\beta$, the M-shard \code{i} cannot work properly.} We then let \hhw{the real number-valued symbol} $p(n_{i})$ denote the probability \hw{that \code{shard i} has a number} $n_{i}$ of malicious nodes, and \hhw{the real number-valued symbol} $\overline{P} \in (0,1) $ denote the probability \hw{that an} M-shard is failed \hw{because of malicious nodes}. \hw{Finally, the integer} $\vartheta \in \mathbb{N}^+$ denotes the number of nodes in \hhw{any equal-sized} M-shard.}

\begin{theorem}\label{theorem3}

\hhw{Assuming nodes are randomly assigned to each shard and the system satisfies Byzantine fault tolerance, we have $\beta=\left\lfloor\frac{\vartheta-1}{3}\right\rfloor$, \hw{and the probability of M-shard failure (denoted by $\overline{P}$) is calculated} as follows:}
\begin{equation}
    \begin{aligned}
    \overline{P} = 
    1 & - \sum_{N=0}^{\vartheta \cdot S} \Big \{ \left(\varphi \right)^{N}\left(1-\varphi \right)^{\vartheta \cdot S-N} \\
    & \times  \sum_{n_{1}=0}^{\min\{N, \beta\}} \cdots\sum_{n_{S}=0}^{\min\{N- \sum_{j=1}^{S-1}n_{j}, \beta \} }\binom{\vartheta}{n_{1}}\cdots \binom{\vartheta}{n_{S}} \Big \}. 
    \label{equ:mshardfail}
    \end{aligned}
\end{equation}
\end{theorem}

\begin{proof}

 \hhw{\hw{We first introduce a random integer variable} $X \in \mathbb{N^{+}}$ that represents the number of malicious nodes in \hhw{all} M-shards. Because the nodes are randomly assigned to shards, $X$ follows the binomial distribution. Thus, we have}
    \begin{equation}
    P\left(X=N\right)=\begin{pmatrix}\vartheta \cdot S\\N\end{pmatrix}\cdot\begin{pmatrix}\delta \end{pmatrix}^{N}\begin{pmatrix}1-\delta \end{pmatrix}^{\vartheta \cdot S-N}.
    \label{equ:binomial}
    \end{equation}
\hhw{The probability that a number $n_{1} \in \mathbb{N}^+$ of malicious nodes are assigned to \code{shard 1} is calculated as follows:}
    \begin{equation}\label{eq:prxn}
    p\left(n_1\right)=\frac{\binom{N}{n_1}\binom{\vartheta \cdot S-N}{\vartheta-n_1}}{\binom{\vartheta \cdot S}\vartheta}.
    \end{equation}
\hhw{\hhw{Given Eq. \eqref{eq:prxn}, the conditional }probability that a number $n_{2} \in \mathbb{N}^+$ of malicious nodes are assigned to \code{shard 2} is written as follows:}
    \begin{equation}
    p\left(n_2\mid n_1\right)=\frac{\binom{N-n_1}{n_2}\binom{\vartheta \cdot S-\vartheta-(N-n_1)}{\vartheta-n_2}}{\binom{\vartheta \cdot S-\vartheta}\vartheta}.
    \end{equation}
\hhw{Similarly, the \hhw{conditional} probability that a number $n_{S} \in \mathbb{N}^+$ of malicious nodes are assigned to \code{shard S} is:}
    \begin{equation}
    \begin{aligned}
    &p\left(n_S\mid n_1,n_2,\cdots,n_{S-1}\right)\\ 
    &=\frac{\binom{N-\sum_{i=1}^{S-1}n_i}{n_S}\binom{\vartheta \cdot S-(S-1)\cdot \vartheta-\left(N-\sum_{i=1}^{S-1}n_i\right)}{\vartheta-n_S}}{\binom{\vartheta \cdot S-(S-1)\cdot \vartheta}\vartheta}.
     \end{aligned}
    \end{equation}
\hhw{Therefore, the joint distribution of \{$n_1, n_2 , \cdots, n_S$\} is}
    \begin{equation}
    \begin{aligned}
    &p\left(n_{1},n_{2},\cdots,n_{S}\right)=\\ &p\left(n_1\right)\times p\left(n_2\mid n_1\right)\times\cdots\times p\left(n_S\mid n_1,n_2,\cdots n_{S-1}\right).
    \end{aligned}
    \label{equ:joint2simplify}
    \end{equation}
\hhw{We can then simplify Eq. \eqref{equ:joint2simplify} by referring to \cite{hafid2020novel}:}
  \vspace{-1mm}
    \begin{equation}
    \begin{aligned}
p\left(n_{1},n_{2},\cdots,n_{S}\right)=\frac{\prod_{i=1}^{S}\binom{\vartheta}{n_{i}}}{\binom{\vartheta \cdot S}{N}}.
    \end{aligned}
    \label{equ:joint}
    \end{equation}

\hhw{In a PBFT-based blockchain sharding system, as long as the proportion of malicious nodes in an M-shard is not smaller than $1/3$, the system cannot guarantee the atomicity of all \hhw{intra-shard} transactions. Thus, we have $\beta=\left\lfloor\frac{\vartheta-1}{3}\right\rfloor$, and the probability that any M-shard works improperly can be represented by the \hhw{real number-valued variable} $\tau \in (0,1)$:}
    \begin{equation}
    \begin{aligned}
    \tau=1-p\left(n_{1}\leq\beta,\cdots,n_{S}\leq\beta\mid X=N\right).
    \end{aligned}
    \label{equ:mshardnodefailed}
    \end{equation}

\hhw{According to Theorem \ref{theorem2}, we get $\delta =\varphi$. By combining Eq. \eqref{equ:binomial}, Eq. \eqref{equ:joint} and Eq. \eqref{equ:mshardnodefailed}, we \hhw{have} the probability of M-shards failure as follows:}
  \vspace{-1mm}
    \begin{equation}
    \begin{aligned}
    & \overline{P} =\sum_{N=0}^{n \cdot S}P\left(X=N\right)\cdot \tau.
    \end{aligned}
    \label{equ:10}
    \end{equation}
\hhw{Finally, via simplifying Eq. \eqref{equ:10}, we get Eq. \eqref{equ:mshardfail}. \hhw{The proof of Theorem \ref{theorem3} concludes.}}

\end{proof}

\begin{remark}

\hhw{ \hw{Eq. \eqref{equ:mshardfail} helps us} estimate the theoretical upper bound of the probability of an M-shard's failure \hhw{by referring to} \cite{2018RapidChain}:}
    \begin{equation*}
    \begin{aligned}
    \overline{P}&\leq\sum_{i=1}^{S}p\left(n_{i}>\beta\right)\\ &\leq S\cdot\sum_{i=\beta+1}^{\vartheta}\binom{\vartheta}{i}\left(\varphi\right)^{i}\cdot\left(1-\varphi\right)^{\vartheta-i}.
    \end{aligned}
    \end{equation*}

\hhw{Fig.\ref{fig:workshardSizwWithPorb} shows that the failure probability of M-shards \hw{(i.e., $\overline{P}$)} changes with the total number of nodes in \hhw{all} M-shards \hw{(i.e., $\vartheta \cdot S$)} when the number of shards is fixed at $S$=16. When the computing power of malicious nodes is high, increasing the number of nodes \hw{helps improve} the security \hhw{level} of the system. When $\vartheta \cdot S$= 4000, Fig. \ref{fig:workShardMalicousNodeWithProb} shows that the failure probability of M-shards \hw{$\overline{P}$ increases following the enlarged $\varphi$.}
\hw{Here, the 4 curves indicated by different configurations of $S$ = \{16, 24, 32, 40\} demonstrate how $\overline{P}$ varies while the number of M-shards (i.e., $S$) grows.
}
\hhw{When given a fixed total} number of nodes, \hhw{a smaller number of} shards \hhw{leads to that} more nodes \hhw{will be assigned to each M-shard. This lowers the probability of M-shard's failure $\overline{P}$.}
\hhw{Similar to the case} with the P-shard, if the security \hhw{level parameter} of the M-shard is set to $\lambda$=18, i.e., $\overline{P} \leq 2^{-18}$, the system can tolerate \hhw{that} malicious nodes occupy \hhw{approximately} 20\% of the \hhw{overall} computing power when $S$=16.}
\end{remark}


\begin{figure*}[t]
\centering
 \subfigure[\zk{TX arrival rate=3200 TXs/Sec}]{
  \includegraphics[width=0.23\textwidth]{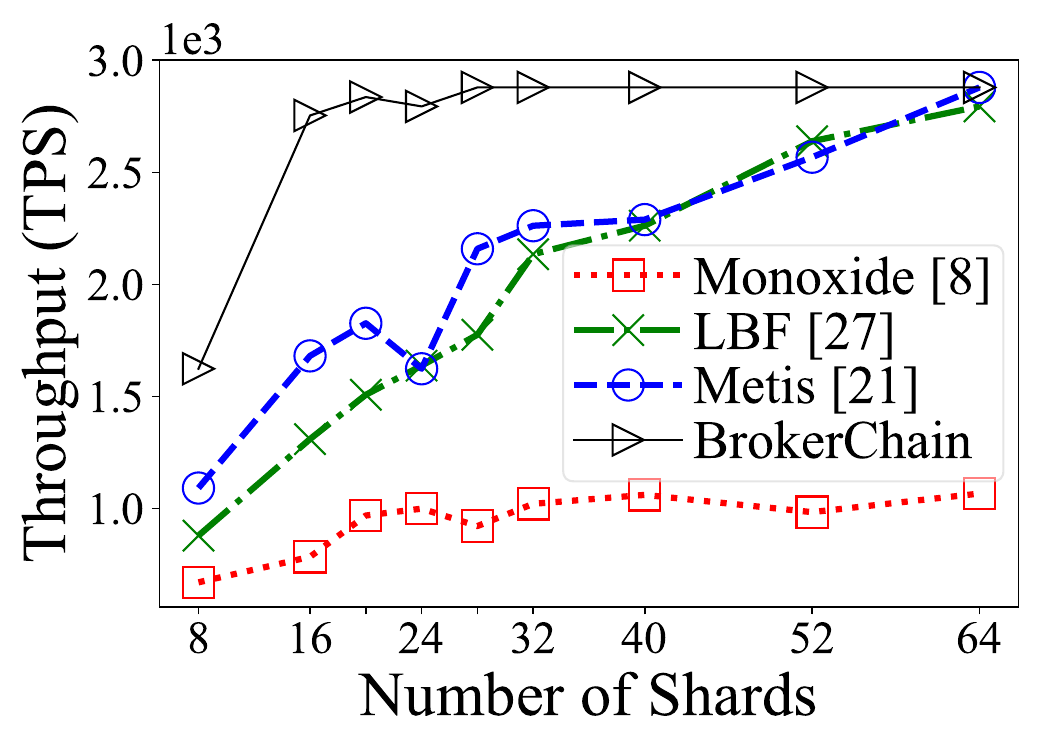}
  \label{fig:LineThrputvsSWithrate3200}
  }%
  \hfill
\subfigure[\zk{TX arrival rate=5000 TXs/Sec}]{
  \includegraphics[width=0.23\textwidth]{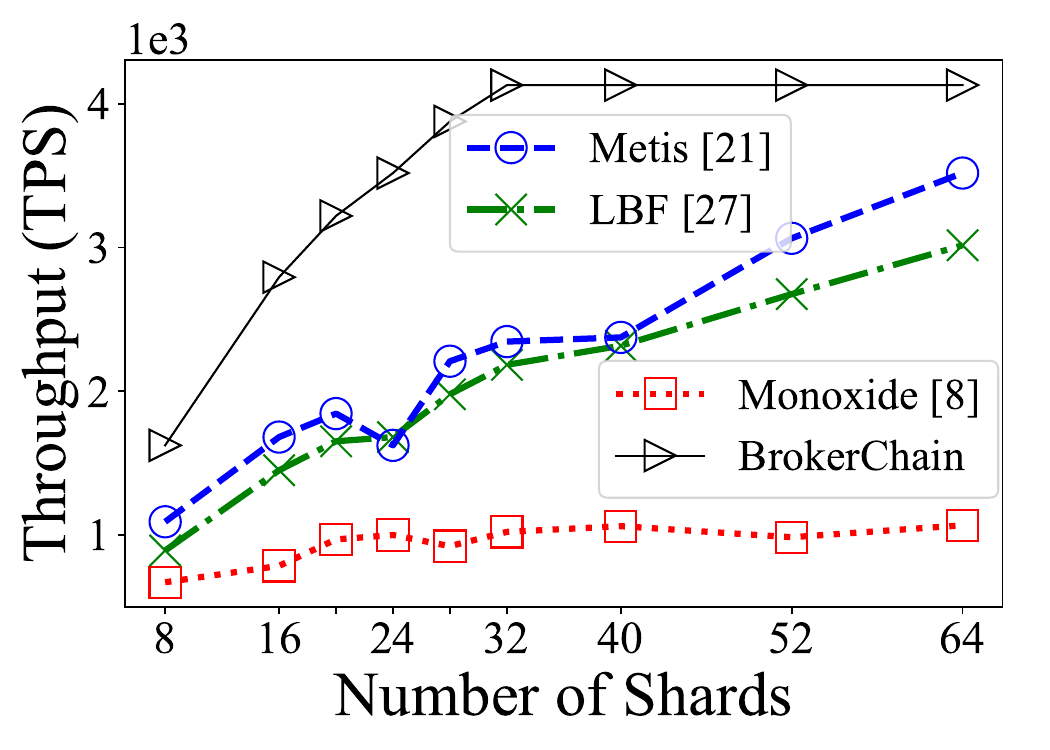}
  \label{fig:LineThrputvsSWithrate5000}
 }%
  \hfill
\subfigure[\zk{TX arrival rate=8000 TXs/Sec}]{
  \includegraphics[width=0.23\textwidth]{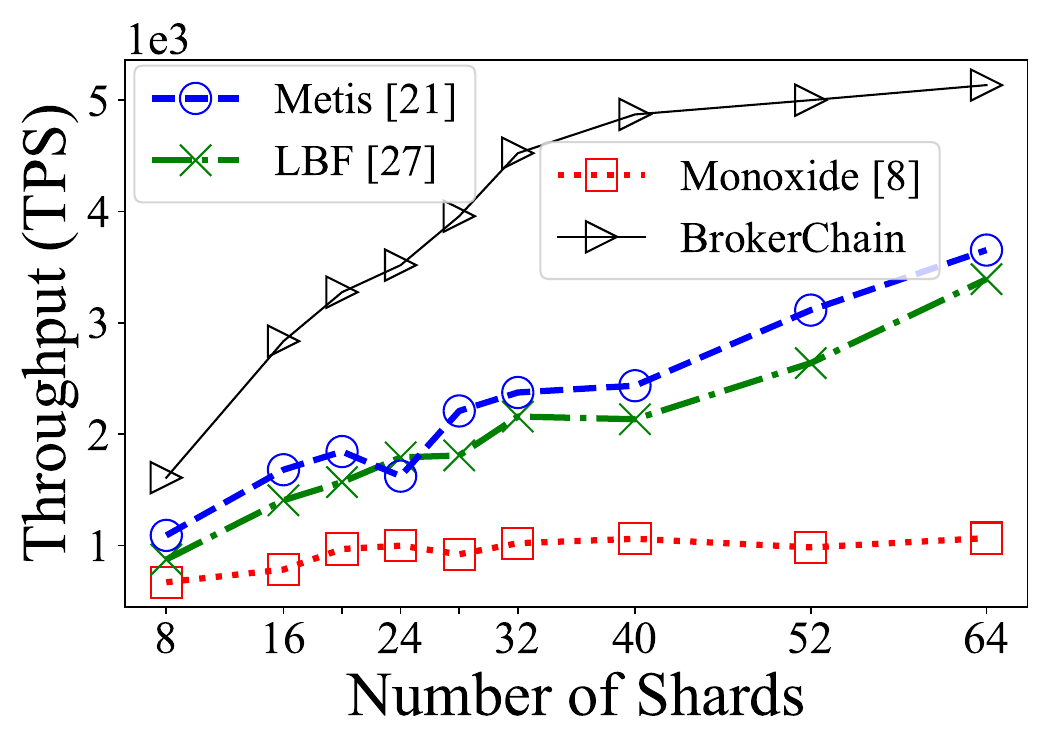}
  \label{fig:LineThrputvsSWithrate8000}
 }
  \hfill
\subfigure[\zk{TX arrival rate=16000 TXs/Sec}]{
\includegraphics[width=0.23\textwidth]{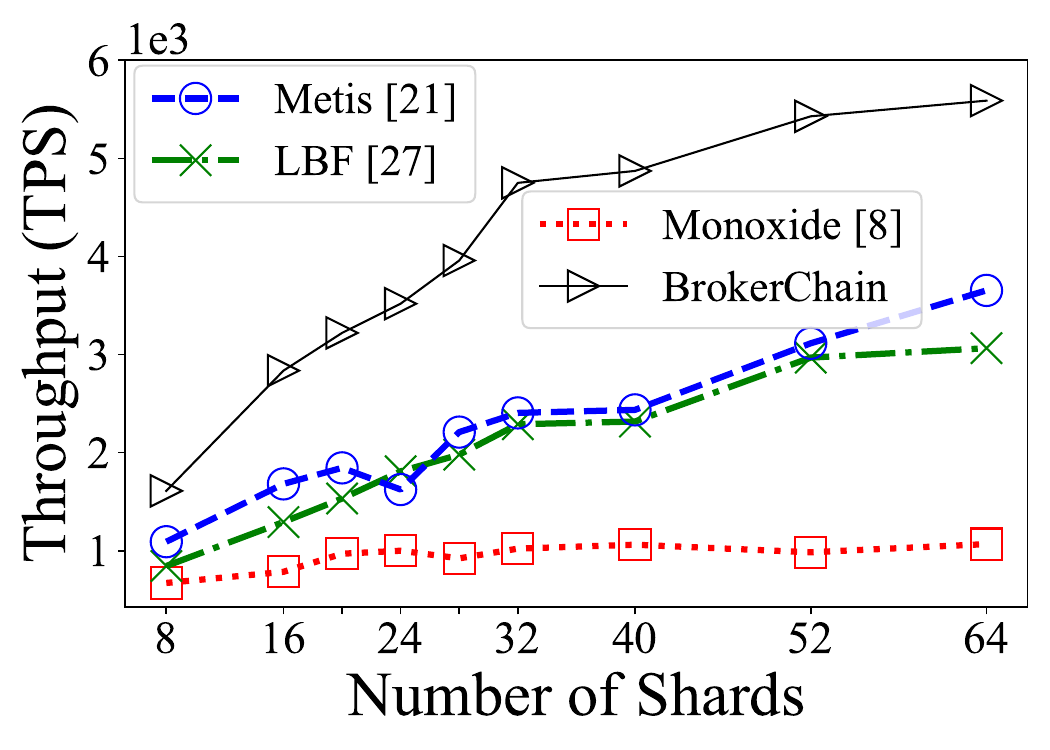}
\label{fig:LineThrputvsSWithrate16000}
 }
\caption{\zk{The throughput under different methods, while varying Tx arrival rate within $\{3200, 5000, 8000, 16000\}$ TXs/Sec and increasing $\#$ of shards. }}
\label{fig:LineThrputvsS}
\end{figure*}

\begin{figure*}[t]
\centering
 \subfigure[\zk{TX arrival rate=3200 TXs/Sec}]{
  \includegraphics[width=0.23\textwidth]{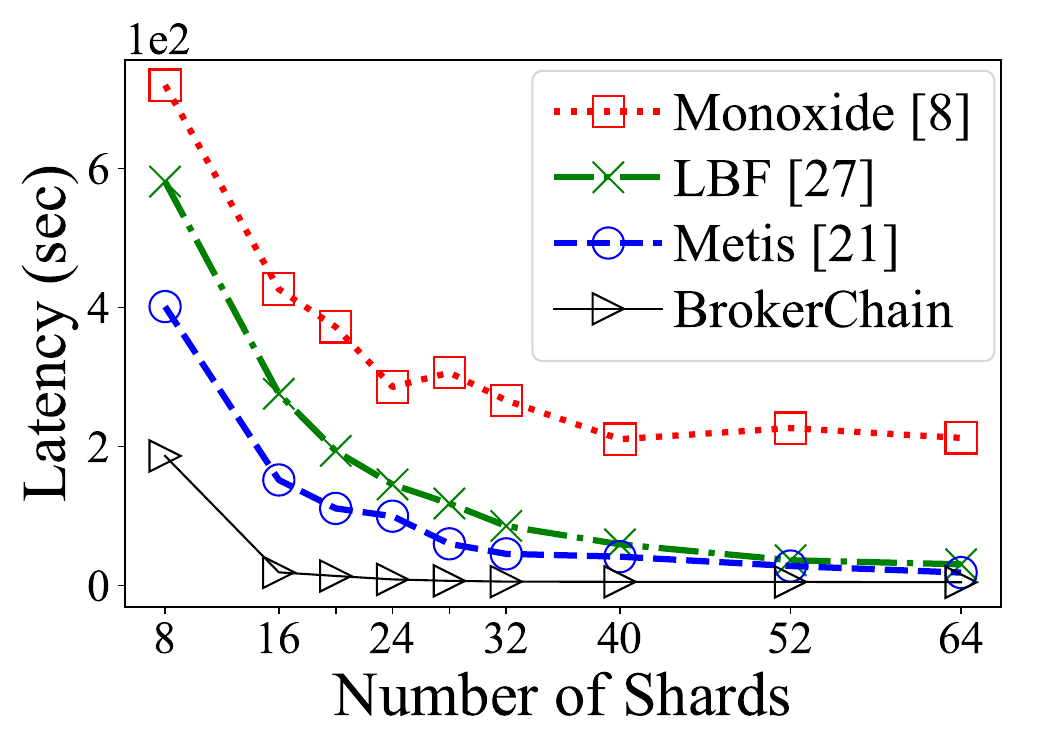}
  \label{fig:LineLatencyvsSWithrate3200}
  }%
  \hfill
\subfigure[\zk{TX arrival rate=5000 TXs/Sec}]{
  \includegraphics[width=0.23\textwidth]{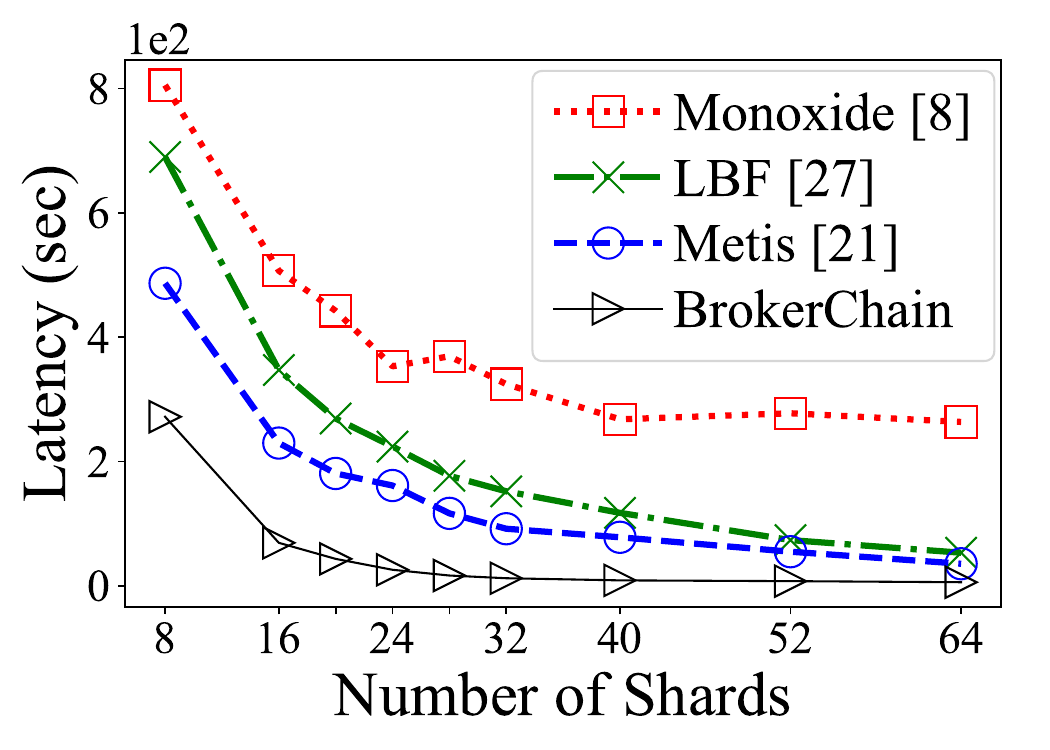}
  \label{fig:LineLatencyvsSWithrate5000}
 }%
  \hfill
\subfigure[\zk{TX arrival rate=8000 TXs/Sec}]{
  \includegraphics[width=0.23\textwidth]{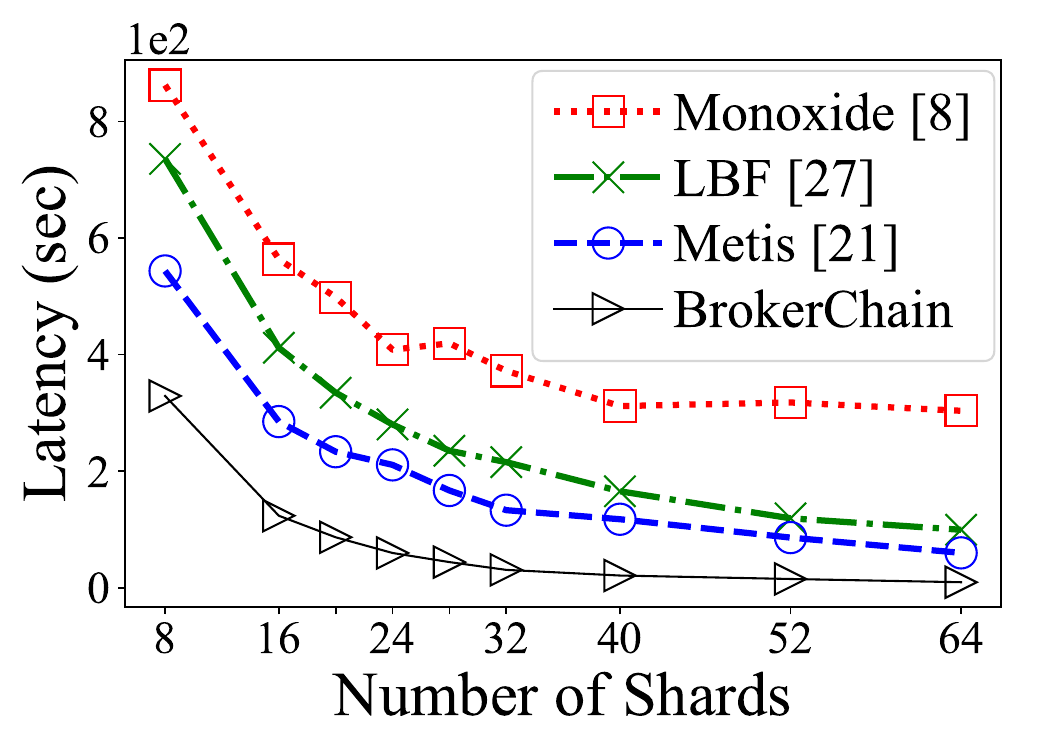}
  \label{fig:LineLatencyvsSWithrate8000}
 }
  \hfill
\subfigure[\zk{TX arrival rate=16000 TXs/Sec}]{
\includegraphics[width=0.23\textwidth]{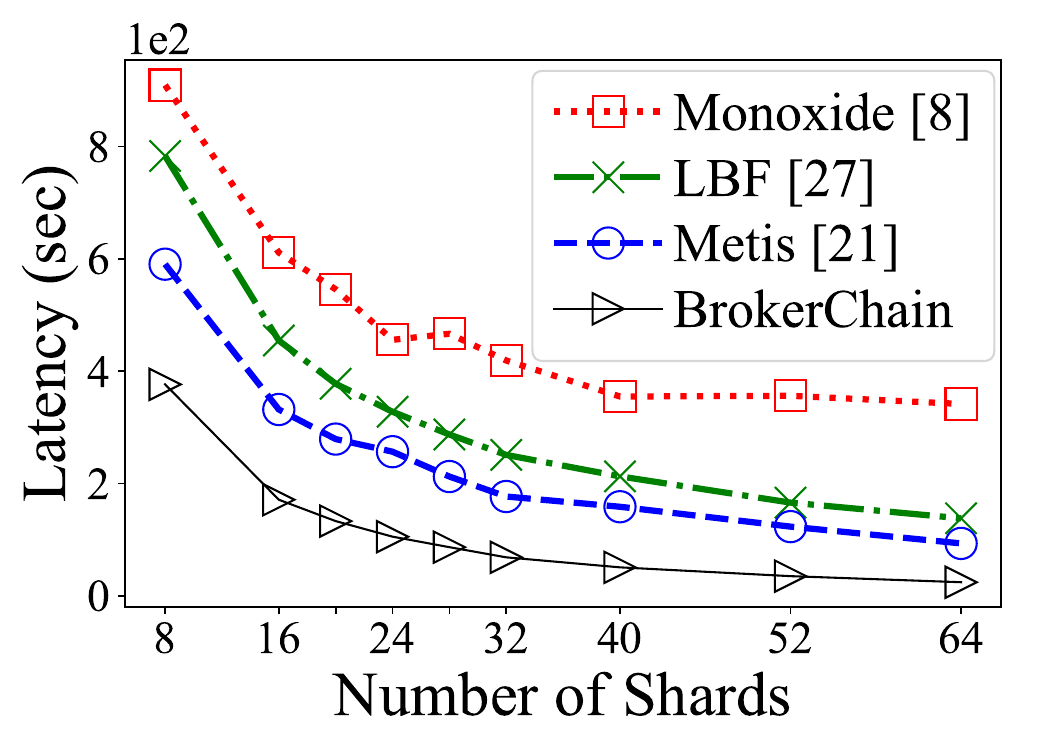}
\label{fig:LineLatencyvsSWithrate16000}
 }
\caption{\zk{The Latency under different methods, while varying Tx arrival rate within $\{3200, 5000, 8000, 16000\}$ TXs/Sec and increasing $\#$ of shards. }}
\label{fig:LineLatencyvsS}
\end{figure*}

\section{Performance Evaluation}\label{sec:performance}

\subsection{Settings}

\textbf{TX-driven Testbed BlockEmulator.} To evaluate the proposed BrokerChain, we first developed \hw{an open-sourced experimental testbed named \underline{BlockEmulator}\cite{huang2023blockemulator}, which enables the blockchain sharding mechanism. On top of BlockEmulator, we then implemented a TX-driven prototype of BrokerChain, which was implemented using Java and Python. BlockEmulator can be deployed in either local physical machines or cloud-based virtual machines such as Alibaba Cloud. Next, we ran our prototype system in BlockEmulator by replaying real-world historical TXs.}

\textbf{Dataset} and \textbf{its Usage.} We adopt real Ethereum TXs as the dataset, which contains 1.67 million historical TXs recorded from Aug. 7th, 2015 to Feb. 13, 2016. At the beginning of each epoch, a number  $N_{\text{TX}}$ of TXs are prepared in chronological order and replayed to the blockchain sharding system with predefined arrival rates. Those TXs are then assigned to different M-shards according to their account's state.

\textbf{Baselines.} We consider the following three baselines. 
Monoxide \cite{Wang2019Monoxide} distributes accounts according to the first few bits of their addresses. 
LBF \cite{greedy} updates the distribution of accounts periodically to achieve load-balanced TX distributions.
Metis  \cite{1995METIS} purely emphasizes the balanced partition on the account's state graph.

\textbf{Metrics.} 
We first study the effect of several critical system parameters considering various system metrics such as TX throughput, TX's confirmation latency, and the queue size of the TX pool. Regarding workload performance, we then measure the total and the variance of shard workloads.

{\color{black}

\begin{figure*}[t]
\centering
 \subfigure[\zk{TX arrival rate=16000 TXs/Sec}]{
  \includegraphics[width=0.22\textwidth]{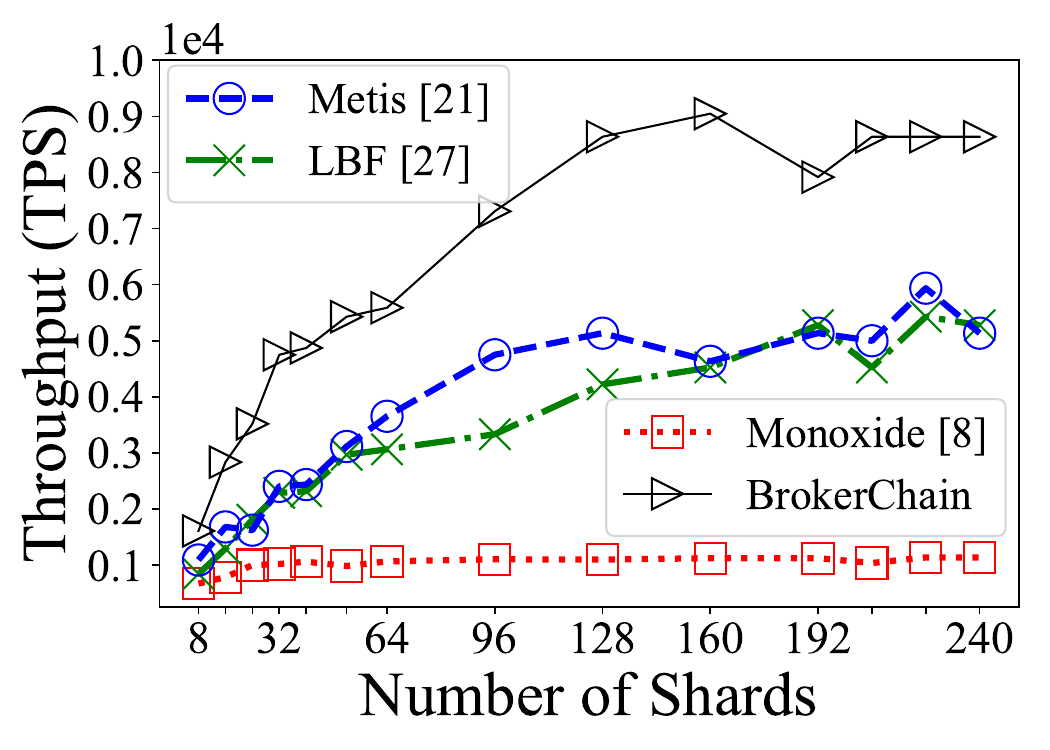}
   \label{fig:LineThrputvsStoMax}
  }%
  \hfill
\subfigure[\zk{TX arrival rate=16000 TXs/Sec}]{
  \includegraphics[width=0.22\textwidth]{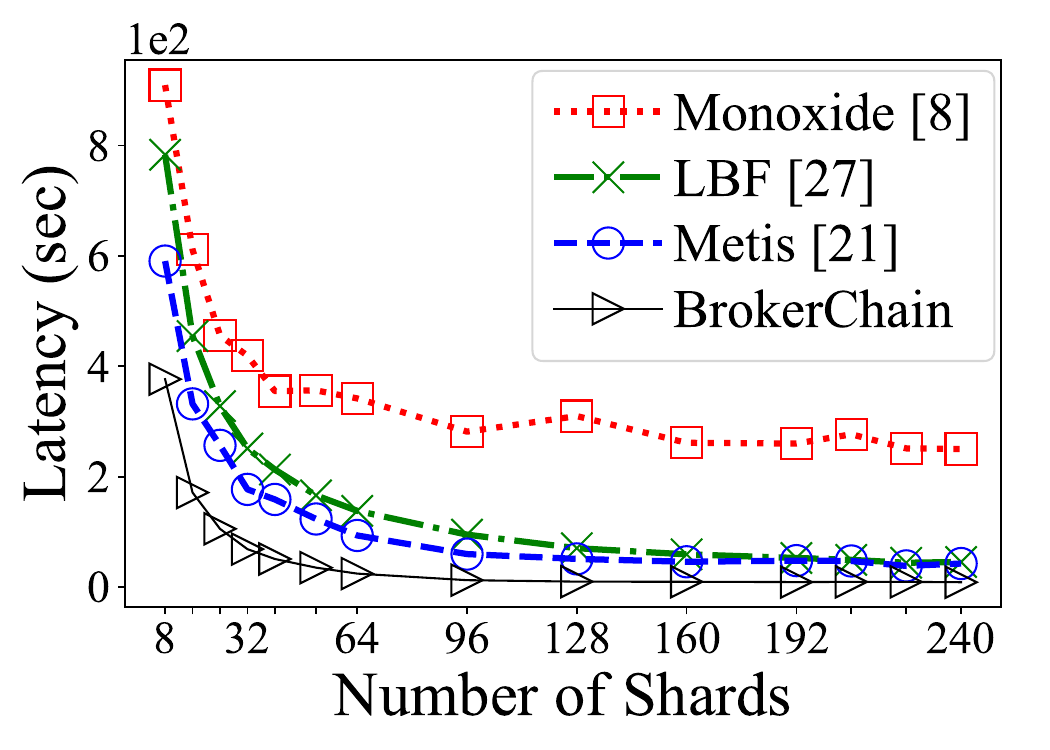}
   \label{fig:LineLatencyvsStoMax}
 }
   \hfill
\subfigure[Throughput, \zk{$S$=32}]{
  \includegraphics[width=0.22\textwidth]{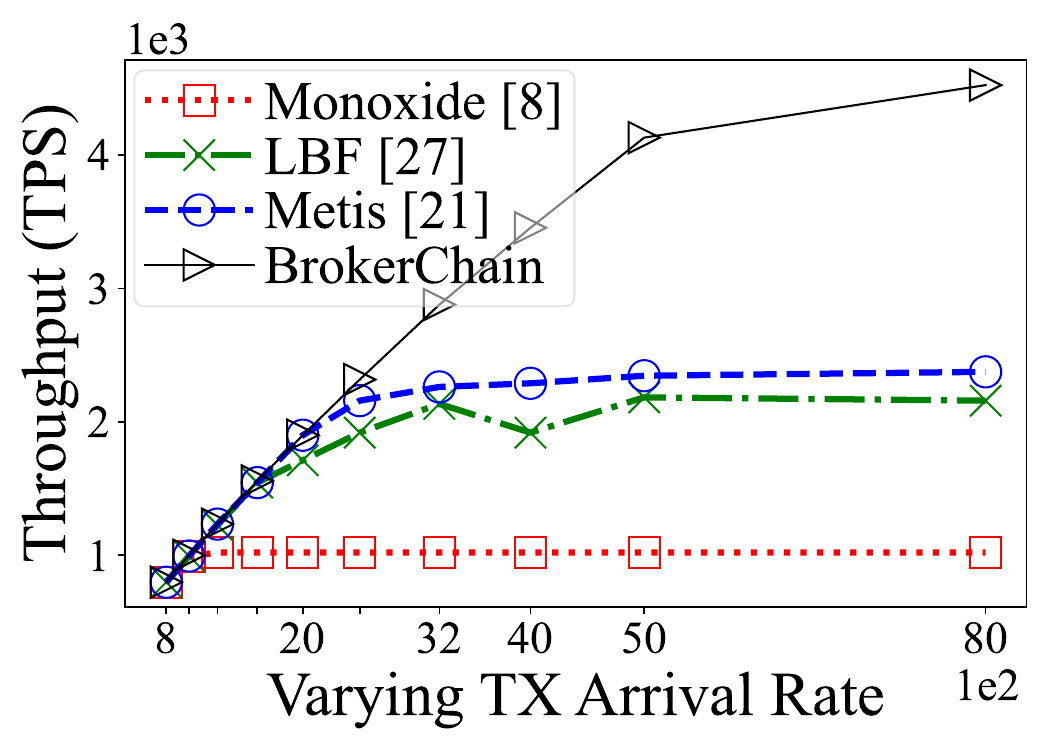}
   \label{fig:LineThrputvsTXRate}
 }
    \hfill
\subfigure[TX confirmation latency, \zk{$S$=32}]{
  \includegraphics[width=0.22\textwidth]{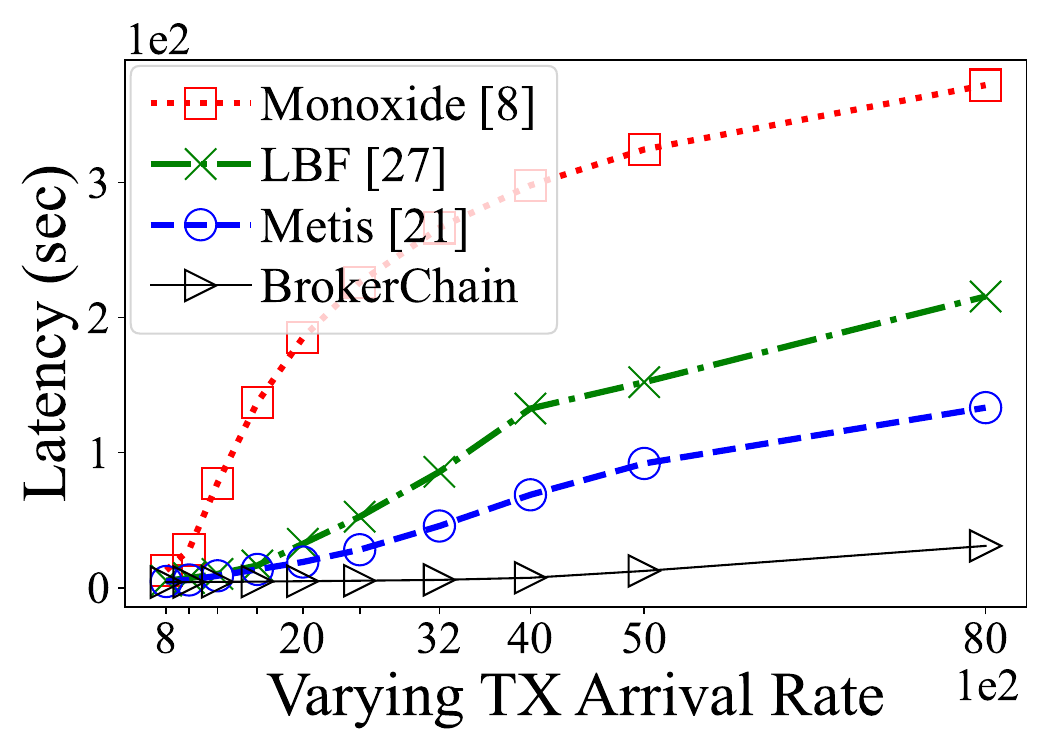}
   \label{fig:LineLatencyvsTXRate}
 }%
\caption{\zk{Throughput and latency \textit{vs} the \# of shards and TX arrival rates, while maximizing system performance to its full potential.}}
\label{fig:LineThrputLatency}
\end{figure*}


\subsection{System Throughput and TX Confirmation Latency}

\begin{table}[t]
\caption{Experimental Results Yielded by Cloud-based Prototype}
\centering
\renewcommand{\arraystretch}{1.4}
\begin{tabular}{llll}
\hline
 Methods/Algorithms   & Monoxide  & LBF  & BrokerChain \\
Avg. System throughput (TPS)         & 156.98  & 236.52  & 352.15\\
Avg. TX confirmation latency (sec)    & 1792.01   & 999.99  & 275.94\\

\hline
\end{tabular}
\label{tab:experimental}
\end{table}

First, we evaluate throughput and transaction confirmation latency using \hw{BlockEmulator deployed in the cloud}. \hw{Through renting 112 virtual machines from Alibaba Cloud, we deploy our prototype system in those virtual machines, which form 16 blockchain shards in total.} Each virtual machine is equipped with 1 CPU core (Intel Xeon, 2.5/3.2GHz) and 2GB memory. The bandwidth of all network connections between nodes is set to 5 Mbps. For each M-shard, the block interval and block capacity are set to 8 seconds and 500 TXs, respectively. The TX arrival rate is fixed to 500 TXs/Sec. \hw{The number of broker accounts $K$ is set to 40.} Since Metis does not correlate with TX's throughput and latency, we only compare the performance of BrokerChain with Monoxide and LBF.
The experimental results are shown in Table \ref{tab:experimental}. We see that the average throughput of BrokerChain achieves 2.24$\times$ and  1.49$\times$ of Monoxide's and LBF's TPS, respectively. Furthermore, BrokerChain also maintains an average TX confirmation latency of 275.94 seconds, which is much lower than that of the other two baselines.

Next, we perform TX-driven simulations \hw{in the BlockEmulator deployed in local physical machines,} by tuning more sophisticated system parameters.
The shard block capacity and block interval are updated to 2000 TXs and 8 seconds, respectively. \hw{$K$ is still set to 40.}
Through varying the TX arrival rate and the number of shards respectively, we still study the throughput and TX's confirmation latency under different sharding \hw{mechanisms}.
%
The simulation results are shown in Fig. \ref{fig:LineThrputvsS} and Fig. \ref{fig:LineLatencyvsS}, \hw{in which the Tx arrival rate is changed within $\{3200, 5000, 8000, 16000\}$ TXs/Sec to simulate increasing workloads, and the number of shards (i.e., $S$) is increased from 8 to 64.}

\hw{
 Overall, those results show that BrokerChain outperforms other baselines in terms of throughput and TX's confirmation latency.
 Fig. \ref{fig:LineThrputvsS} demonstrates that BrokerChain exhibits a quick growth of throughput when the number of shards increases, and it converges earlier than the other three baselines when fixing a TX arrival rate. For example, when $S=8$, all methods show a low throughput. When $S$ grows, the throughput of all methods improves drastically. This is because a larger number of shards brings higher transaction processing capability. 
 BrokerChain peaks its TPS at approximately 2800, 4100, and 5100 when $S$ is set to 16, 32, and 64, and the TX arrival rate is configured with 3200, 5000, and 8000 TXs/Sec, respectively.
 On the other hand, the performance gap between BrokerChain and other baselines widens following the growing TX arrival rates.
 The reasons behind the excellent throughput of BrokerChain are two-fold: i) BrokerChain has a better workload balance across shards, and ii) it can lead to a low ratio of CTX. These two advantages help BrokerChain have a better performance that is close to the linearly increasing scalability than other baselines.
 We disclose more insights in the subsequent evaluations.
 }
%
%

\hw{Fig. \ref{fig:LineLatencyvsS} shows the average TX confirmation latency under the same parameters as that of Fig. \ref{fig:LineThrputvsS}. The latency performance of all methods under each parameter setting illustrates an opposite observation compared with Fig. \ref{fig:LineThrputvsS}. The reasons behind those observations are also the same. Thus, we omit the description of Fig. \ref{fig:LineLatencyvsS}.
}

%
%
%

\hw{
 In the previous group of evaluation, we found that the throughput of BrokerChain and the other two baselines Metis and LBF did not converge when the number of shards reached 64 and the TX arrival rate exceeded 8000 TXs/Sec. This motivates us to explore the upper-boundary performance of those methods. Thereby, we increase the number of shards while fixing the TX arrival rate at 16000 TXs/Sec. 
 Fig. \ref{fig:LineThrputvsStoMax} and Fig. \ref{fig:LineLatencyvsStoMax} show the pressure test of throughput and latency, respectively.
 Note that, almost all transactions are the CTXs over all shards when $S$ reaches 128. This fact imposes a huge challenge of handling CTXs to all methods. Due to the proposed broker-based mechanism, BrokerChain exhibits excellent performance in both throughput and TX confirmation latency compared with the other three baselines. For example, BrokerChain can peak at 9048 TPS and show a low latency while increasing the number of shards.
 Compared with other baselines, the TPS of BrokerChain is almost 10$\times$ and 2$\times$ of that of Monoxide and the other two methods, respectively.
}
%
%
%
%
%
%

Fig. \ref{fig:LineThrputvsTXRate} and Fig. \ref{fig:LineLatencyvsTXRate} show the maximum performance of throughput and latency \textit{vs} varying TX arrival rates, while fixing $S$=32. 
BrokerChain shows an overwhelming low latency compared with the other three baselines.
\hw{When the TX arrival rate reaches 8000 TXs/sec, BrokerChain hits more than twice the throughput of Metis. This observation proves again that BrokerChain can guarantee low TX confirmation latency and high TPS, particularly under a high TX arrival rate.
}

\begin{figure*}[t]
\centering
 \subfigure[\zk{TX arrival rate=3200 TXs/Sec, S=32}]{
  \includegraphics[width=0.22\textwidth]{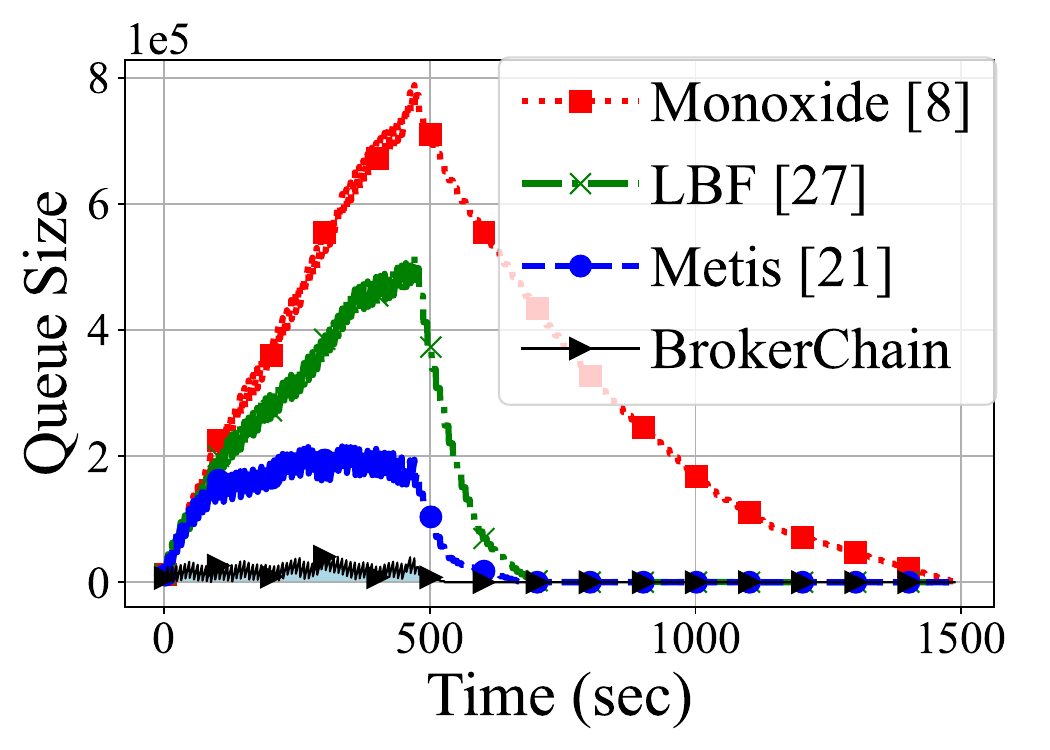}
   \label{fig:queuesizeCase1}
  }%
  \hfill
\subfigure[\zk{TX arrival rate=5000 TXs/Sec, S=32}]{
  \includegraphics[width=0.22\textwidth]{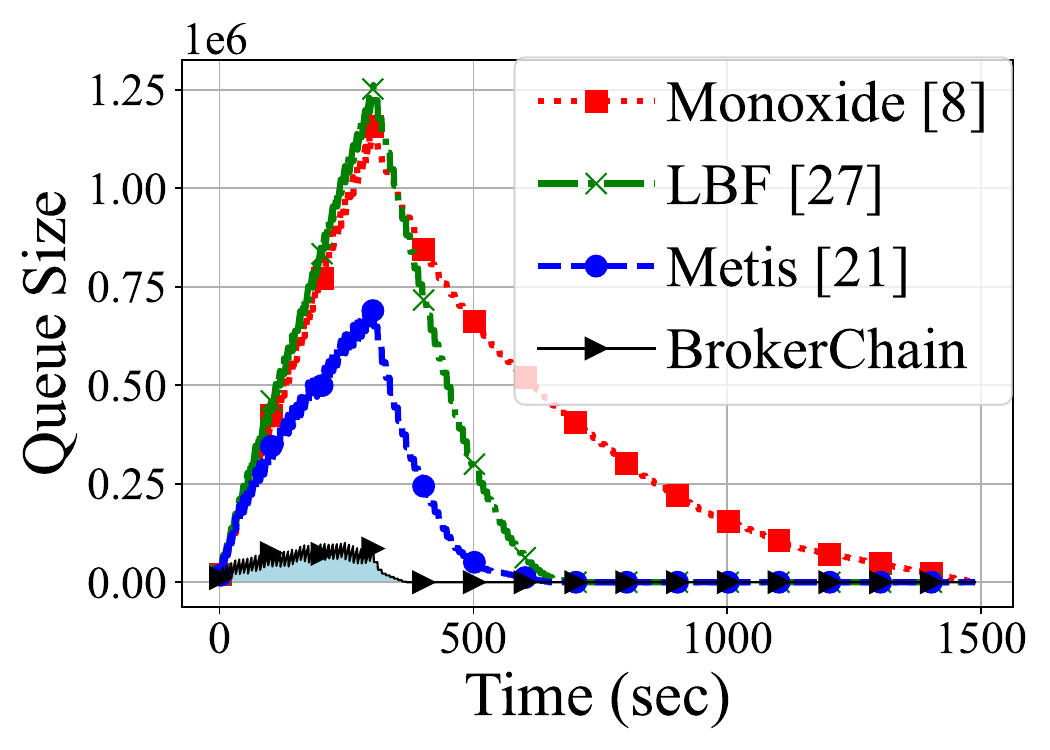}
   \label{fig:queuesizeCase2}
 }%
  \hfill
\subfigure[\zk{TX arrival rate=8000 TXs/Sec, S=32}]{
  \includegraphics[width=0.22\textwidth]{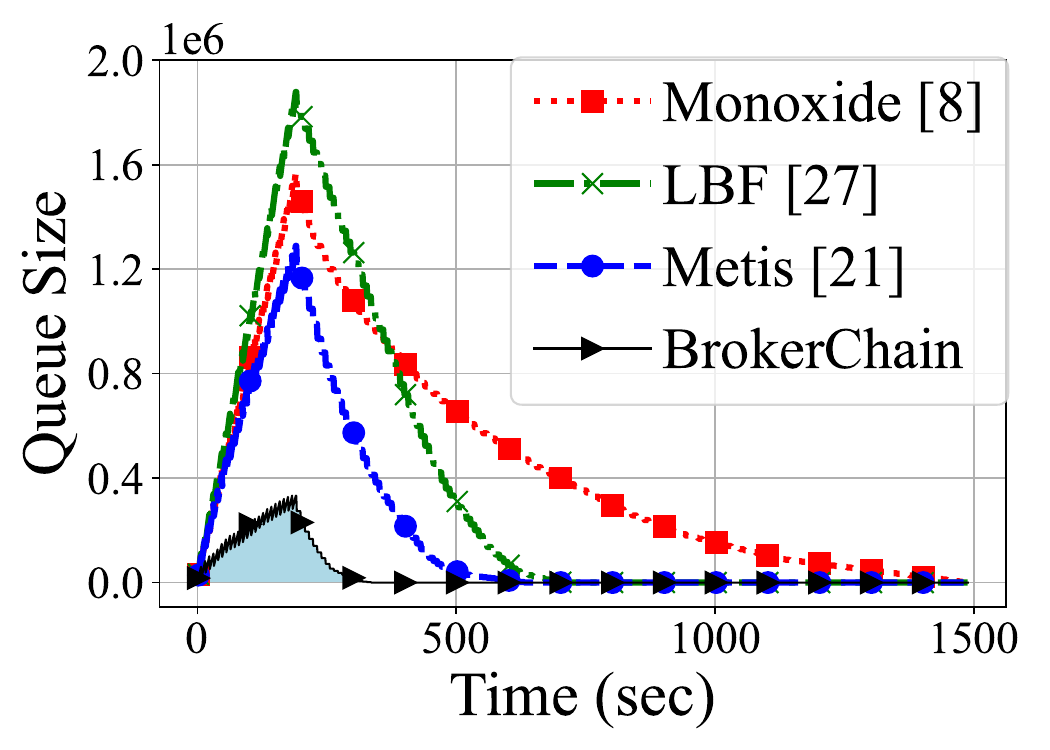}
   \label{fig:queuesizeCase3}
 }
   \hfill
\subfigure[\zk{TX arrival rate=16000 TXs/Sec, S=32}]{
  \includegraphics[width=0.22\textwidth]{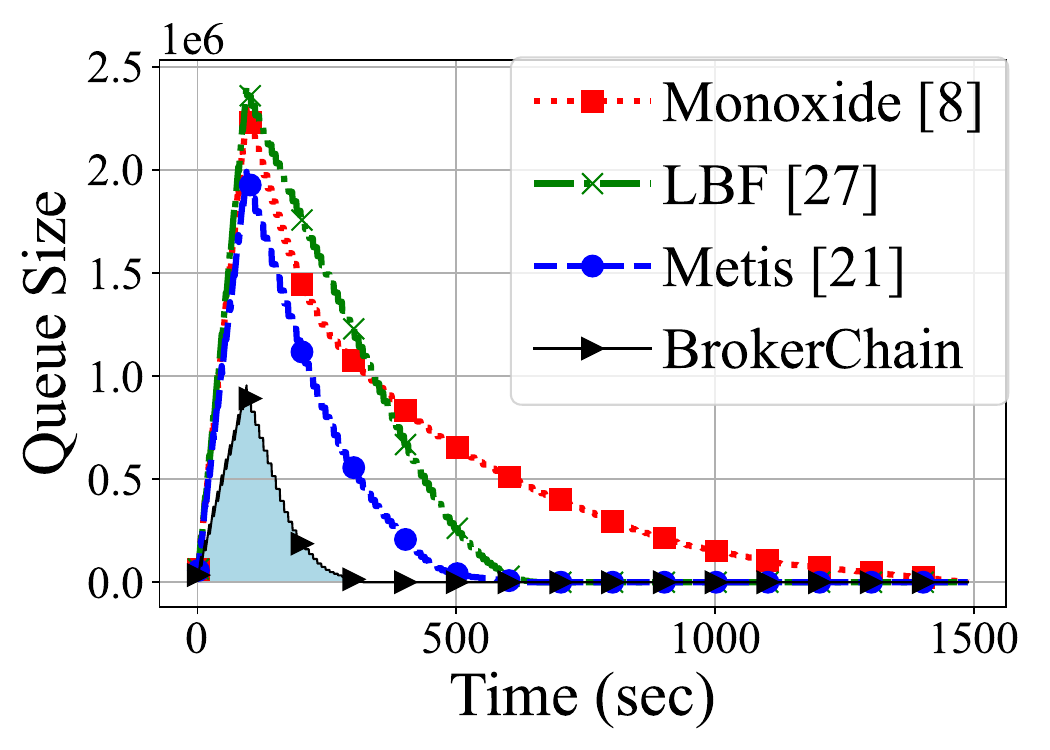}
   \label{fig:queuesizeCase4}
 }
 \subfigure[\zk{TX arrival rate=3200 TXs/Sec, S=64}]{
  \includegraphics[width=0.22\textwidth]{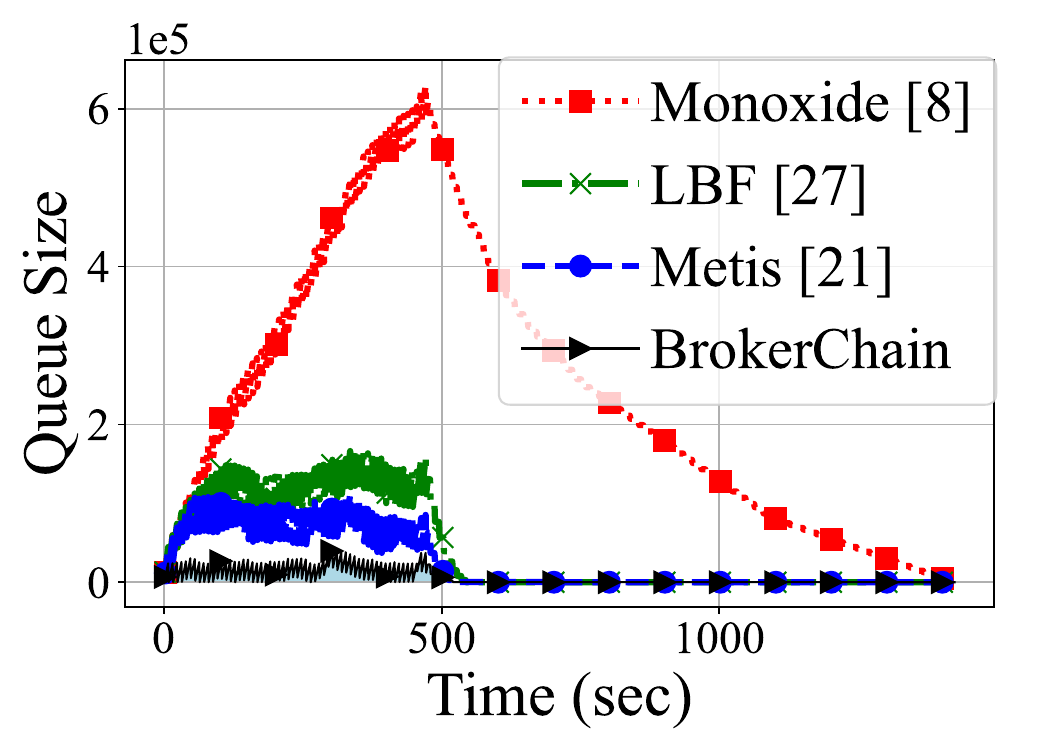}
   \label{fig:queuesizeCase5}
  }%
  \hfill
\subfigure[\zk{TX arrival rate=5000 TXs/Sec, S=64}]{
  \includegraphics[width=0.22\textwidth]{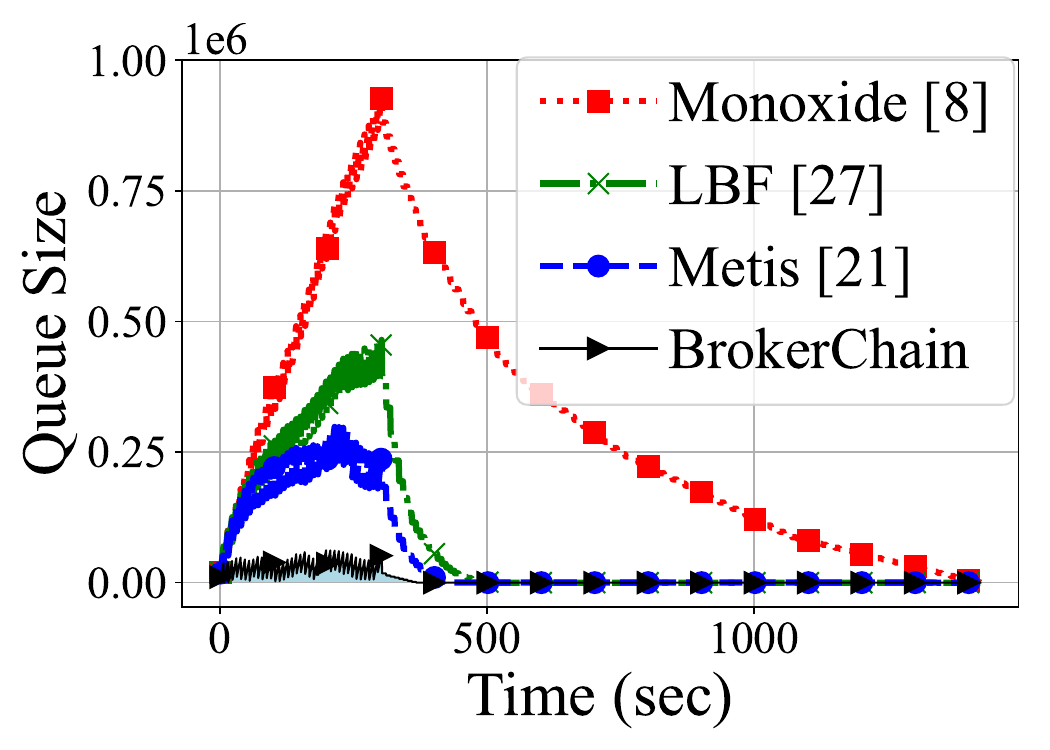}
   \label{fig:queuesizeCase6}
 }%
  \hfill
\subfigure[\zk{TX arrival rate=8000 TXs/Sec, S=64}]{
  \includegraphics[width=0.22\textwidth]{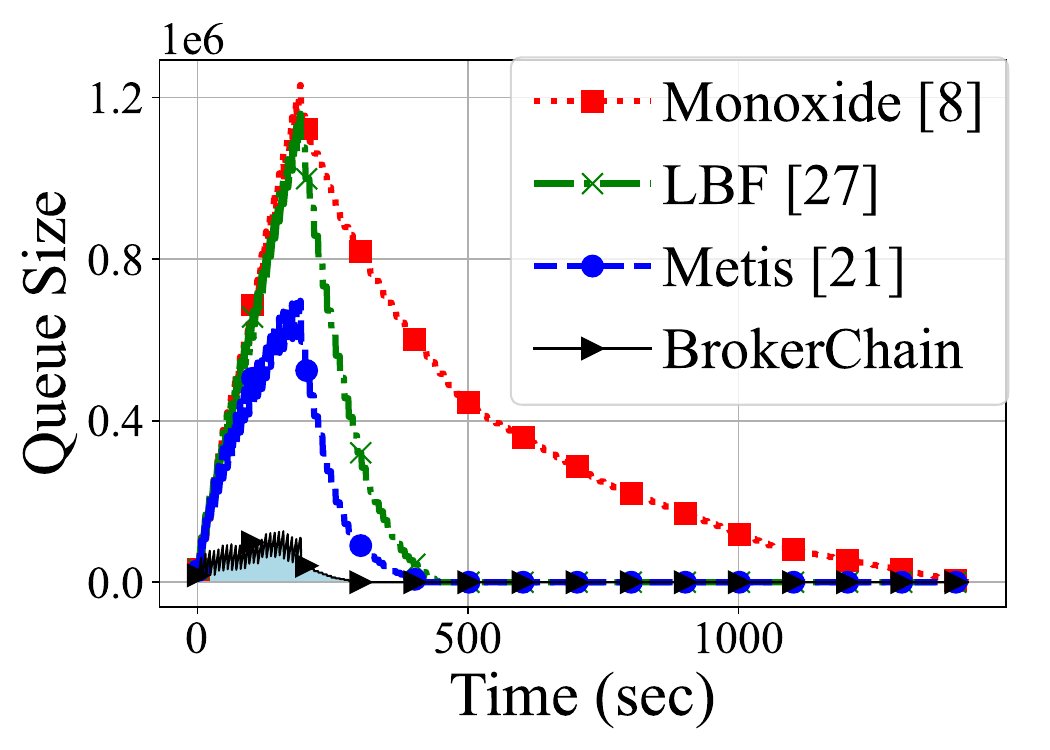}
   \label{fig:queuesizeCase7}
 }
   \hfill
\subfigure[\zk{TX arrival rate=16000 TXs/Sec, S=64}]{
  \includegraphics[width=0.22\textwidth]{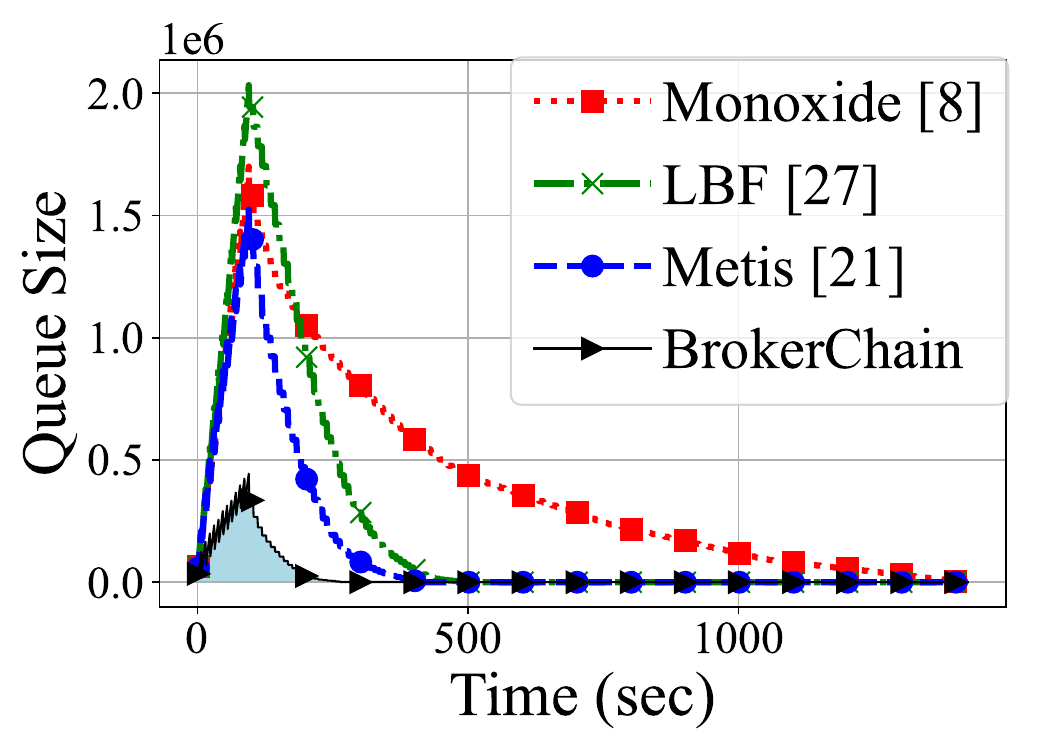}
   \label{fig:queuesizeCase8}
 }
\caption{\hw{Queue size of the TX pool while changing the arrival rate of transactions}.}
\label{fig:queuesize}
\end{figure*}

\begin{figure}[t]
\centering
\includegraphics[width=0.8\columnwidth]{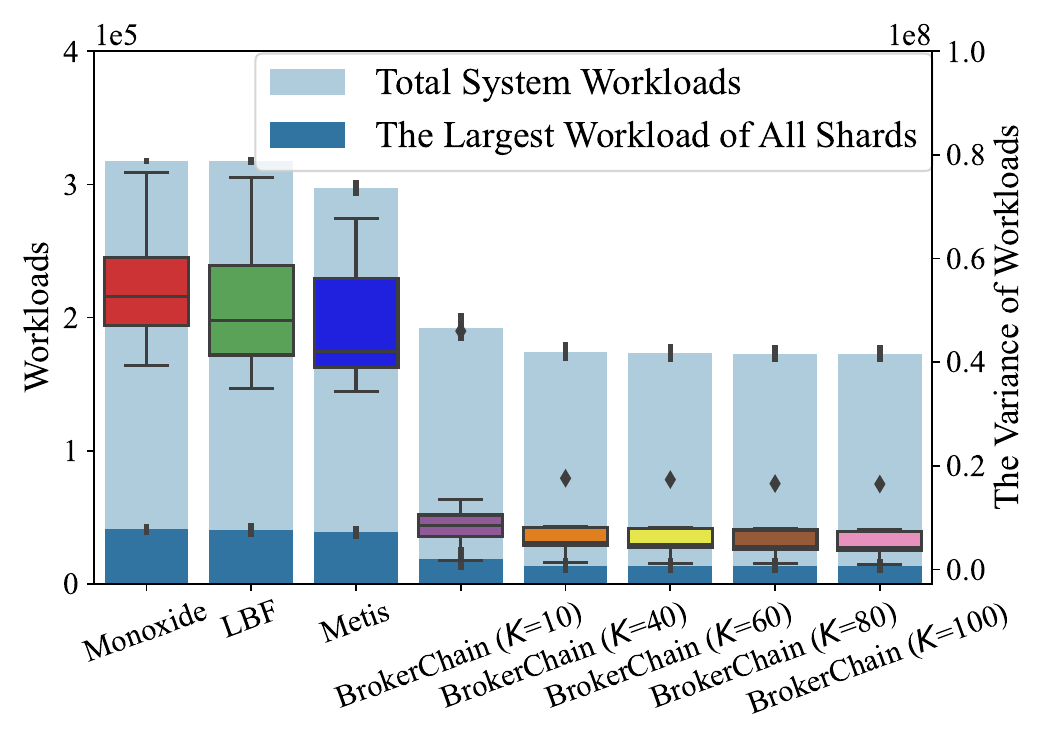}
    \caption{Effect of the number of segmented accounts (i.e., \hw{the number of broker accounts} $K$).}
\label{fig:twinS16DiffK}
\end{figure}

\begin{figure*}[t]
\centering
 \subfigure[\zk{CDF of total system workloads}]{
  \includegraphics[width=0.23\textwidth]{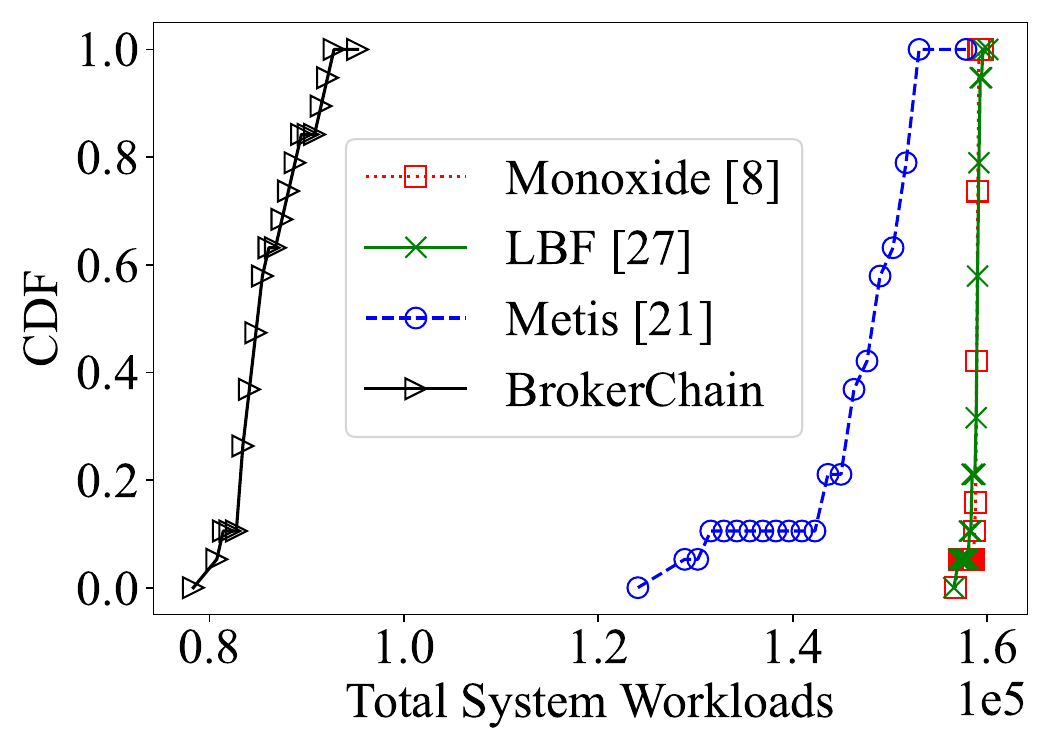}
  \label{fig:cdftoal}
  }%
  \hfill
\subfigure[\zk{CDF of the variance of workloads}]{
  \includegraphics[width=0.23\textwidth]{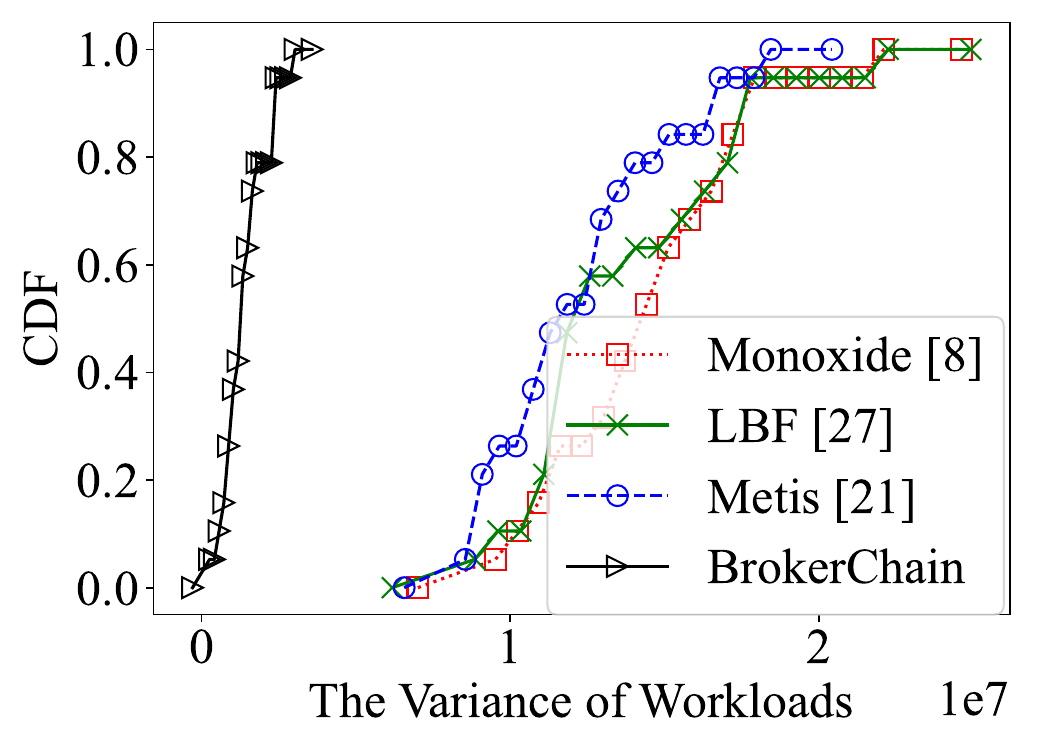}
  \label{fig:cdfvar}
 }%
  \hfill
\subfigure[\zk{CDF of the largest workload}]{
  \includegraphics[width=0.23\textwidth]{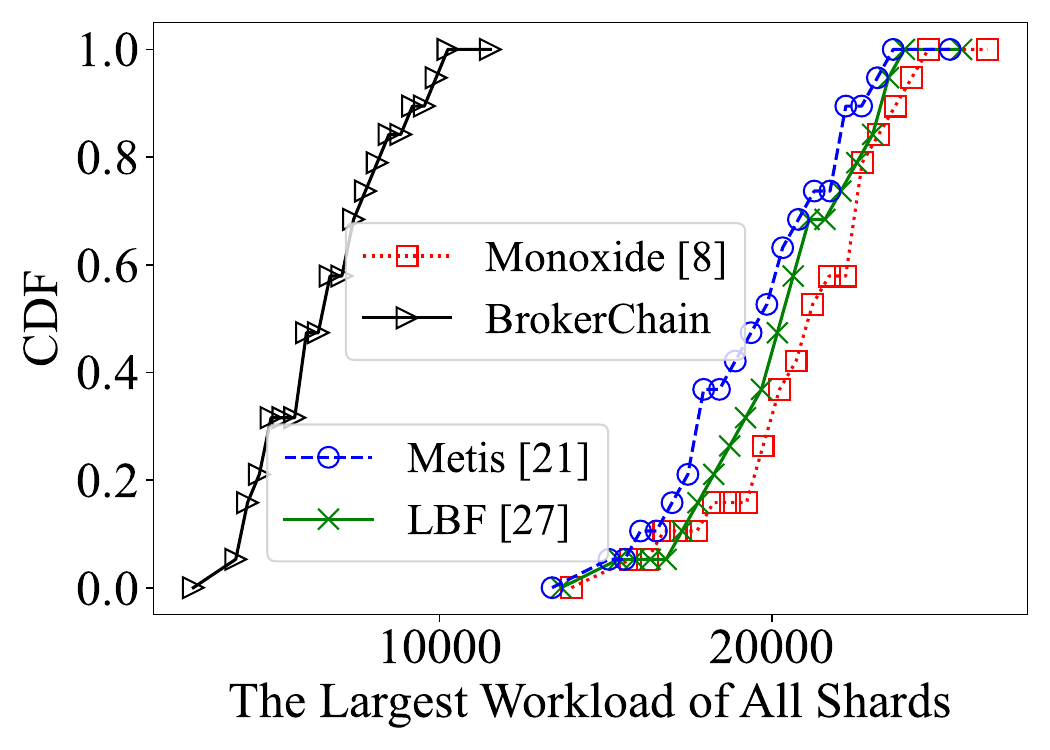}
  \label{fig:cdfmax}
 }
  \hfill
\subfigure[\zk{Cross-shard TX ratio}]{
\includegraphics[width=0.23\textwidth]{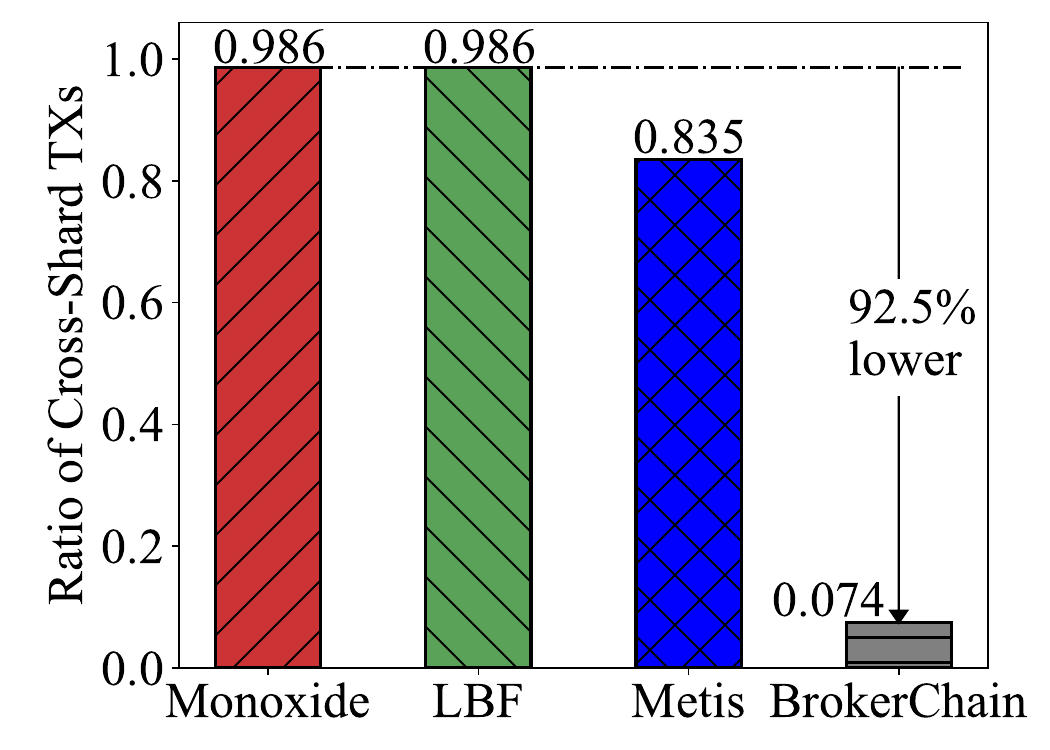}
\label{fig:barS64}
 }
\caption{\zk{The workload performance under different methods, while fixing $S$=64, $N_\text{TX}$=8e4 and $K$=40. CDF stands for cumulative distribution function.}}
\label{fig:FixCompare}
\end{figure*}

\begin{figure*}[t]
\centering
 \subfigure[\zk{Total system workloads \textit{vs} $S$}]{
  \includegraphics[width=0.23\textwidth]{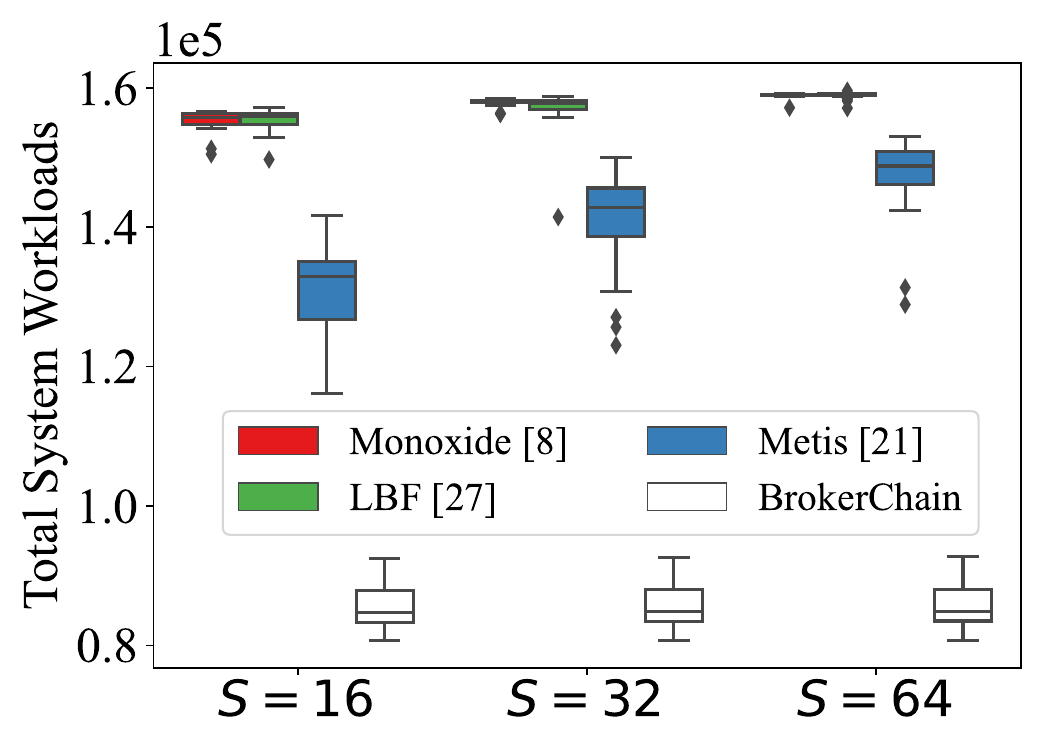}
   \label{fig:TotalBOXDiffS}
  }%
  \hfill
\subfigure[\zk{The variance of workloads \textit{vs} $S$}]{
  \includegraphics[width=0.23\textwidth]{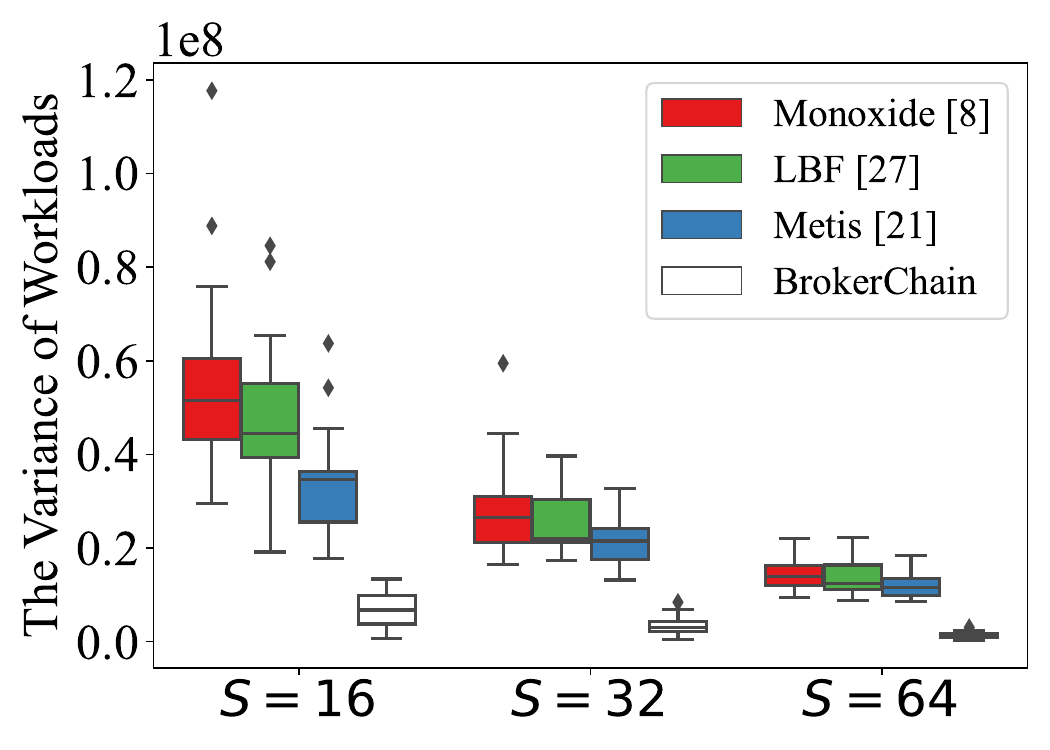}
   \label{fig:VARBOXDiffS}
 }%
  \hfill
\subfigure[\zk{The largest workload \textit{vs} $S$}]{
  \includegraphics[width=0.23\textwidth]{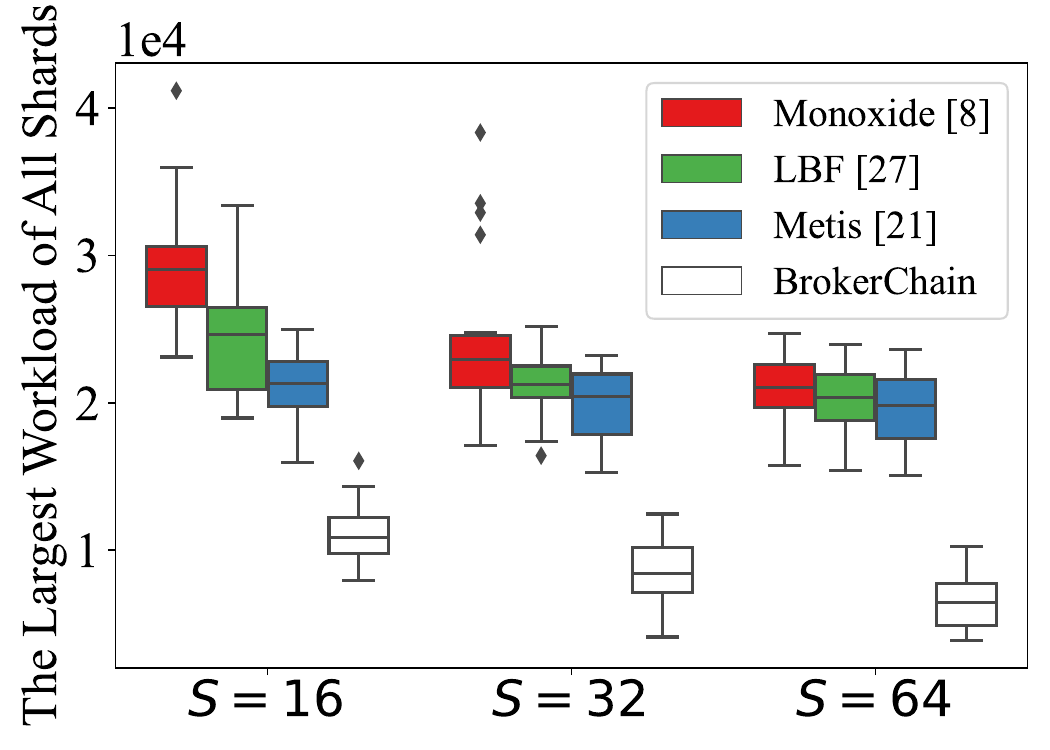}
   \label{fig:MaxBOXDiffS}
 }
  \hfill
\subfigure[\zk{Cross-shard TX ratio, $S$=32}]{
\includegraphics[width=0.23\textwidth]{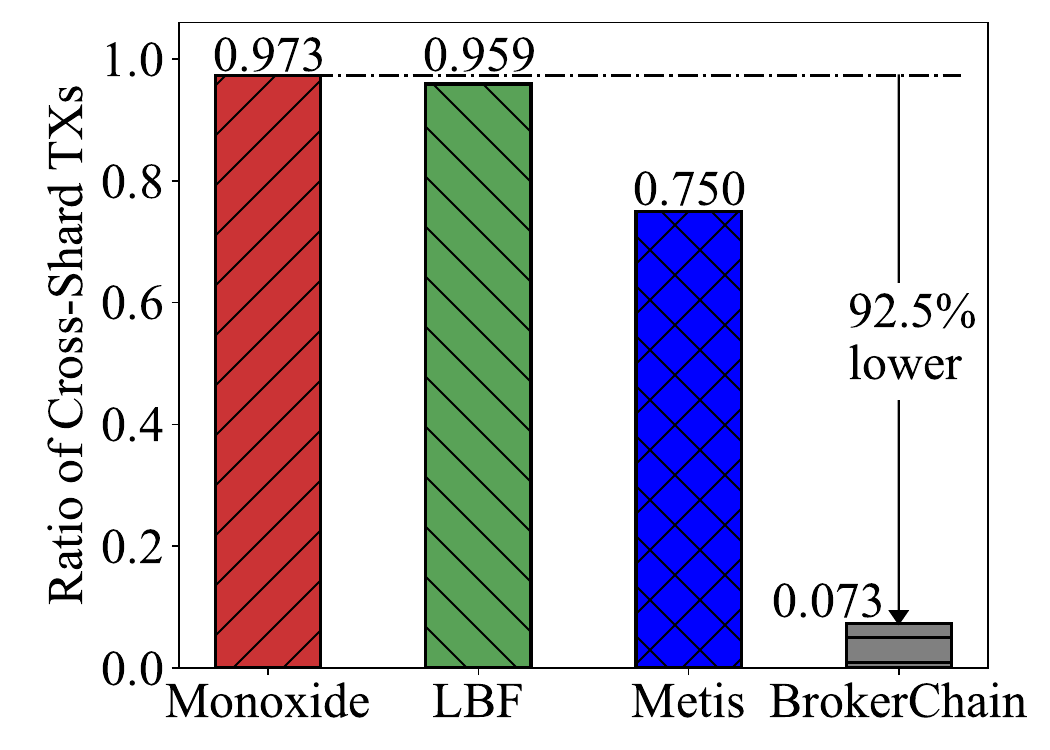}
\label{fig:barS32}
 }
\vspace{-2mm}
\caption{\zk{The effect of shard \# on the total, the largest and the variance of workloads, while varying $S$ within $\{16, 32, 64\}$, and fixing $N_\text{TX}$=8e4, $K$=40.}}
       
\label{fig:diffShard}
\end{figure*}

\begin{figure*}[t]
\centering
 \subfigure[\zk{Total system workloads \textit{vs} $N_\text{TX}$}]{
  \includegraphics[width=0.23\textwidth]{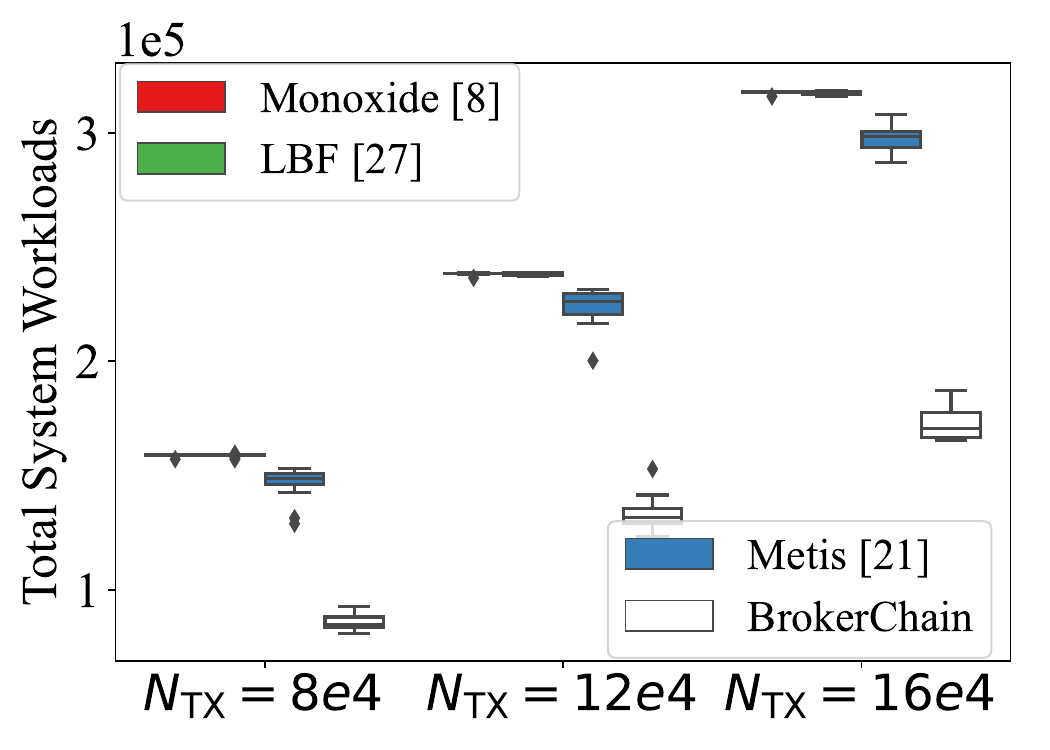}
   \label{fig:TotalBOXDiffNTX}
  }%
  \hfill
\subfigure[\zk{The variance of workloads \textit{vs} $N_\text{TX}$}]{
  \includegraphics[width=0.23\textwidth]{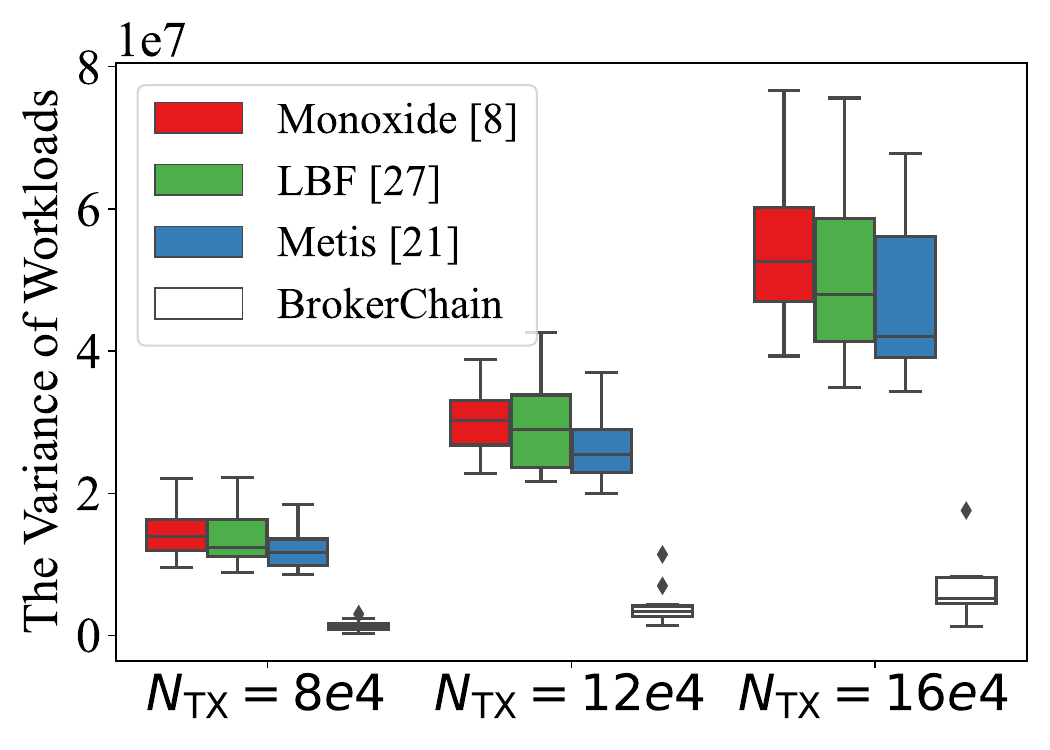}
   \label{fig:VARBOXDiffNTX}
 }%
  \hfill
\subfigure[\zk{The largest workload  \textit{vs} $N_\text{TX}$}]{
  \includegraphics[width=0.23\textwidth]{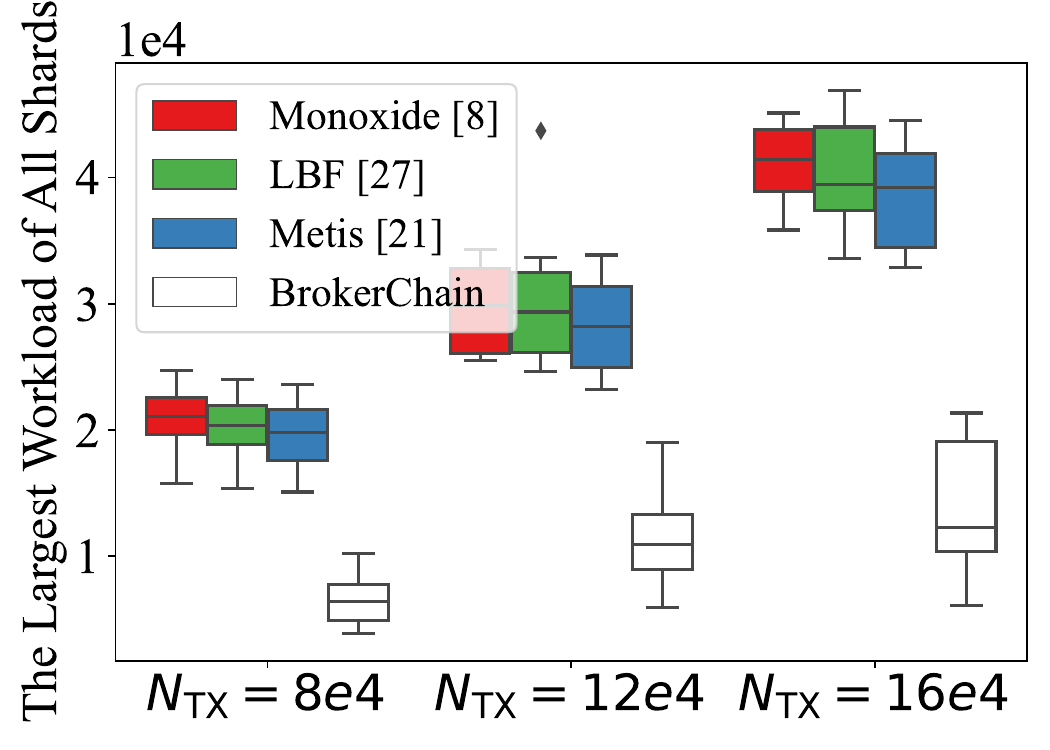}
   \label{fig:MaxBOXDiffNTX}
 }
  \hfill
\subfigure[\zk{Cross-shard TX ratio, $N_\text{TX}$=12e4}]{
\includegraphics[width=0.23\textwidth]{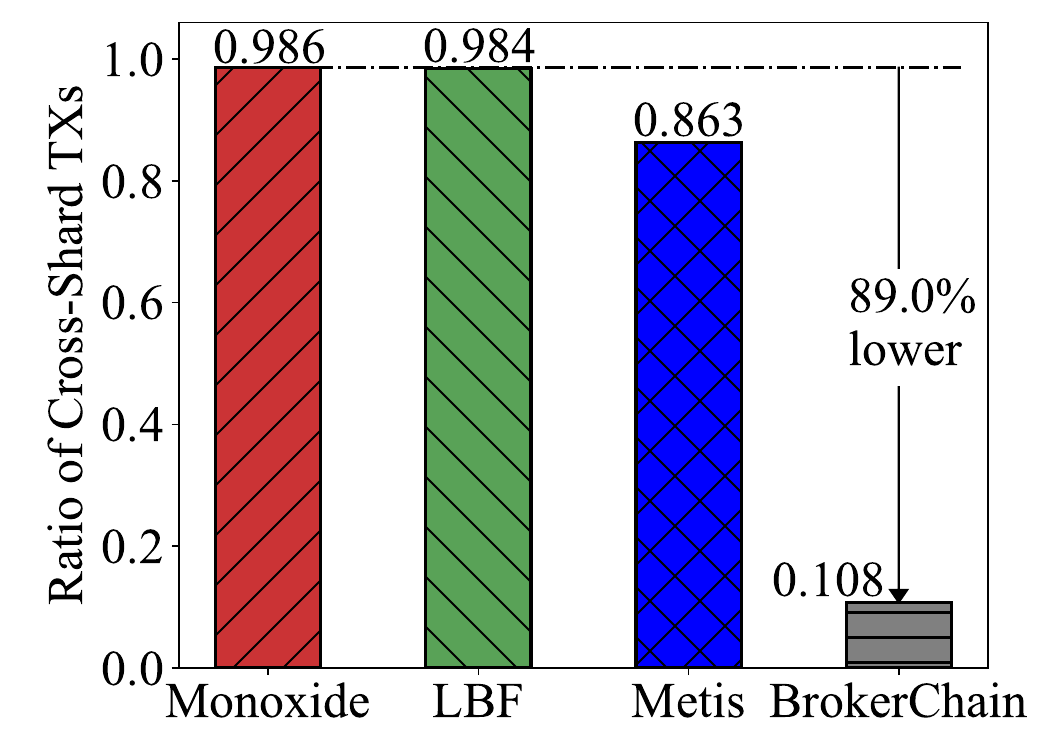}
\label{fig:barS64NTX120000}
 }
\caption{\zk{The effect of TX \# on the total, the largest and the variance of workloads, while varying $N_\text{TX}$ within $\{$8e4, 12e4, 16e4$\}$, and fixing $S$=64, $K$=40.}}
\label{fig:diffNTX}
\end{figure*}

\begin{figure*}[t]
\centering
\subfigure[\zk{Heatmap of transaction workload \hhw{yielded by} Monoxide.}]{
  \includegraphics[width=0.8\columnwidth]{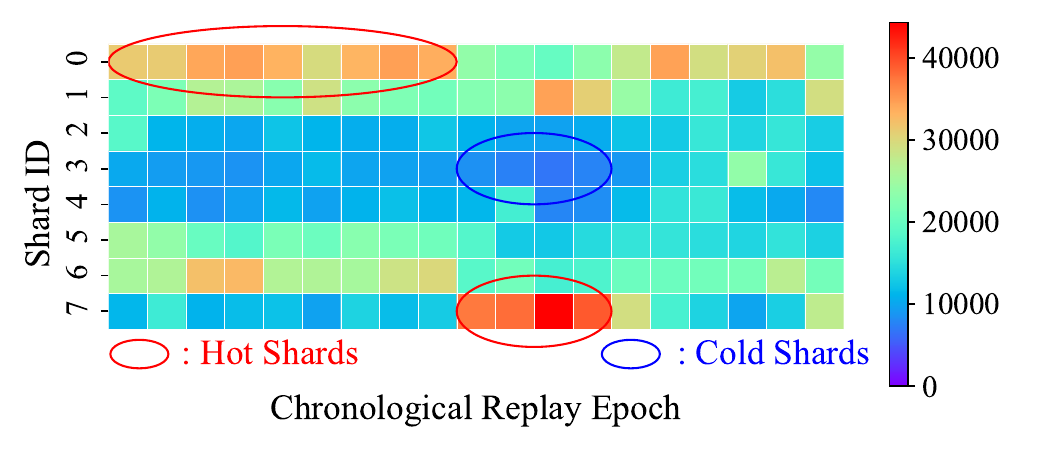}
   \label{fig:MonoxideHeatMapVmin0}
 }
\subfigure[\zk{Heatmap of transaction workload \hhw{yielded by} BrokerChain.}]{
  \includegraphics[width=0.8\columnwidth]{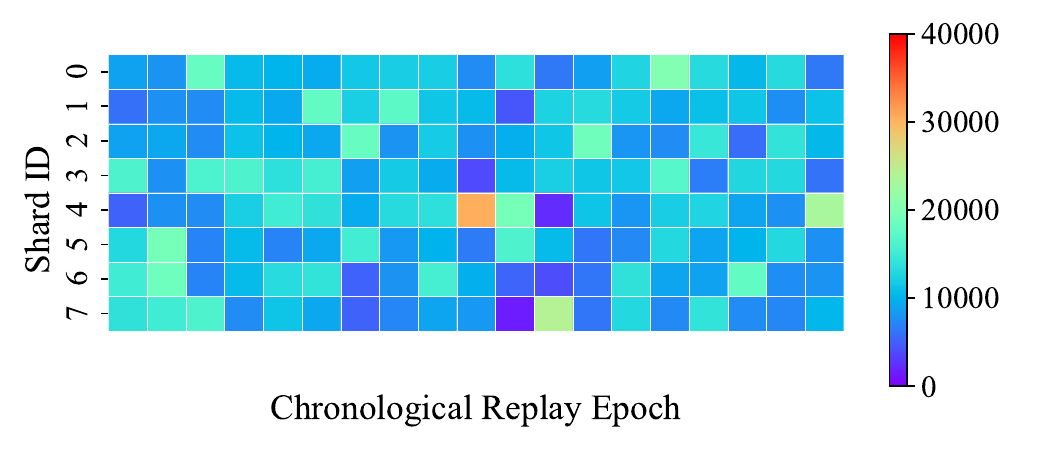}
   \label{fig:ProposedHeatMapVmin0}
 }
\caption{\hhw{The heatmap comparison of transaction workloads yielded by Monoxide and BrokerChain. Parameters are set to $S$=8, $N_\text{TX}$=8e4 for each epoch. } }
\label{fig:heatMapVmin0}
\end{figure*}

\subsection{Queue Size of TX Pool}

 We then investigate the change of the TX pool by monitoring its queue size. 
 %
 %
 We adopt a total of 1.67 million TXs to keep feeding the TX pool of the blockchain system with certain stable rates until all TXs are used up.
 Referring to the VISA's throughput, i.e., 4000 TPS approximately, we change the TX arrival rate \hw{from 3200 to 16000} TXs/Sec and fix $K$=40 for BrokerChain. The number of shards is fixed to 32 \hw{or 64}.
 %
 %
 Fig. \ref{fig:queuesize}  shows that the queue size keeps growing when injecting TXs continually at the first few hundred seconds. When all 1.67 million TXs are consumed, the queue size shrinks. Among all methods, Monoxide maintains a linearly increasing queue and shows the largest.
 When the TX arrival rate is \hw{lower than 5000 TXs/Sec, BrokerChain has little accumulation in the TX pool, either with 32 shards or 64 shards. Once the TX arrival rate is higher than 5000 TXs/Sec, the accumulation effect of TX pool gradually becomes obvious. Upon comparing Fig. \ref{fig:queuesizeCase1}-\ref{fig:queuesizeCase4} with Fig. \ref{fig:queuesizeCase5}-\ref{fig:queuesizeCase8}, a greater number of M-shards can help alleviate the transaction accumulation effect of TX pool.
 This result suggests that BrokerChain is more capable to maintain a small-sized TX pool than the other three baselines.}
 %
 %
 We attribute this result to the following insights. 
 Firstly, Monoxide tends to put TXs into a small number of hot shards. This policy induces a large number of cross-shard TXs, which cause large processing latency. Thus, the TX pool always has a large size under Monoxide. 
 Secondly, LBF tries to distribute TXs evenly to all shards, thus the queue size of TX pool is maintained in a low level under a small TX arrival rate. However, LBF is not capable to handle the cross-shard TXs timely under a large arrival rate.
\hw{Thirdly, Metis tries to adjust the distribution of TXs in an adaptive manner. This can minimize the count of cross-shard TXs and optimize the workload balance across shards. However, the effect of Metis is not significant. In contrast, BrokChain can keep the ratio of cross-shard TXs at a low level. Thus, BrokChain can process the accumulated transactions timely.
 }

\subsection{The Effect of the Number of \hw{Broker} Accounts}

 Recall that broker's accounts can be segmented and deployed to different shards. Now we evaluate the effect of the number of \hw{broker} accounts by varying $K$ while fixing $S$=64, $N_\text{TX}$=16e4. Through running 10 epochs \hw{of blockchain consensus,} the workload performance is shown in Fig. \ref{fig:twinS16DiffK}, in which the total system workload is \hw{calculated as} the sum of all shard's TX workloads.
 We see that as $K$ increases, the total system workloads, the largest and the variance of shard workloads all decrease.
 The reason is that with more broker accounts, they can help reduce the number of cross-shard TXs and make the TX workloads more balanced than other baselines. 
 \hw{However, when $K$ exceeds 40, the performance-improving effect becomes saturated. This is because when $K$ is greater than 40, the TXs involving brokers cover the majority of all TXs in the dataset, thus further increasing $K$ does not lead to significant improvement.
 }

\subsection{Shard Workloads and Cross-shard TX Ratios}

Through Fig. \ref{fig:FixCompare}, we study the total, the largest and the variance of shard workloads under all methods.
By assigning $N_\text{TX}$=8e4 to $S$=\zk{64} shards at each epoch and seting $K$=40 for BrokerChain, Fig. \ref{fig:cdftoal} shows that BrokerChain yields the lowest total system workloads. This is because BrokerChain can reduce the number of cross-shard TXs to a certain low degree. To have a clearer insight, we compare the ratios of cross-shard TXs in Fig. \ref{fig:barS64}. The average cross-shard TX ratios of Monoxide, LBF, Metis and BrokerChain are 98.6\%, 98.6\%, 83.5\% and \zk{7.4\%}, respectively. 
Furthermore, we are also curious about the breakdown of total system workloads. Fig. \ref{fig:cdfvar} and Fig. \ref{fig:cdfmax} show that BrokerChain has the smaller variances of workloads and smaller the largest workload than that of other 3 baselines. These observations imply that BrokerChain can maintain the most balanced shard workloads. 

In the next group of simulation, to figure out the effect of the number of shards on workload performance, we vary $S$ while fixing $N_\text{TX}$=8e4, $K$=40. 
As shown in Fig. \ref{fig:TotalBOXDiffS}, we see that the increasing $S$ leads to growing total system workloads. This is because the number of cross-shard TXs increases following the increasing number of shards.
Fig. \ref{fig:VARBOXDiffS} and Fig. \ref{fig:MaxBOXDiffS} show that BrokerChain outperforms other baselines. Furthermore, BrokerChain also has the fewest \textit{outliers} in figures. This observation implies that BrokerChain can make shard workloads more balanced and more stable than baselines.
Again, Fig. \ref{fig:barS32} shows that BrokerChain still has the lowest cross-shard TX ratio.

Finally, to evaluate the effect of the number of TXs feeding per epoch on the workload performance, we vary $N_\text{TX}$ within $\{$8e4, 12e4, 16e4$\}$ and fix $S$=64 and $K$=40.  Fig. \ref{fig:TotalBOXDiffNTX} shows that BrokerChain maintains around \zk{40\%-45\%} lower the total system workloads than baselines.
 Although the increasing $N_\text{TX}$ leads to growing system workloads, Fig. \ref{fig:VARBOXDiffNTX}, Fig. \ref{fig:MaxBOXDiffNTX} and Fig. \ref{fig:barS64NTX120000} demonstrate that BrokerChain yields the lower variance, the lower largest shard workload, and the lower cross-shard TX ratio than other baselines.

\subsection{\hw{Visual Comparison of M-Shard Workloads}}

\hw{
 To offer a visual comparison in terms of shard workloads between BorkerChain and Monoxide, we show the heatmaps of transaction workloads under these two methods in Fig. \ref{fig:heatMapVmin0}.
 Using the same parameter settings with that of Fig. \ref{fig:Motivation}, \ref{fig:ProposedHeatMapVmin0} illustrates that BrokerChain can effectively eliminate hot shards and yield a much more balanced workload distribution than Monoxide.
}





\section{Conclusion}\label{sec:Conclusion}

BrokerChain \cite{huang2022brokerchain} has been proposed to serve as a cross-shard protocol for the account-based sharding blockchain. In BrokerChain, the TX workload balance among all shards is achieved by the fine-grained state-graph partition and account segmentation mechanisms. BrokerChain handles the cross-shard TXs by exploiting \hw{the sophisticated broker mechanism. Compared with our conference version \cite{huang2022brokerchain}, this full-version article shows more design details and more theoretical analysis, and presents new experimental results.} The evaluation results obtained from both the cloud-based prototype and the TX-driven simulations demonstrate that BrokerChain outperforms the state-of-the-art sharding methods in terms of TX throughput, transaction confirmation latency, queue size of TX pool, and workload balance.
In our future work, we plan to study the incentive mechanism that can incentivize accounts to become brokers.

\section*{Acknowledgement}

This Work is partially supported by National Key R\&D Program of China (No. 2022YFB2702304), and the National Natural Science Foundation of China (62272496).

\ifCLASSOPTIONcaptionsoff
  \newpage
\fi

\bibliographystyle{IEEEtran}
\bibliography{Reference}



\begin{IEEEbiography}[{\includegraphics[width=1in,height=1.25in,clip,keepaspectratio]{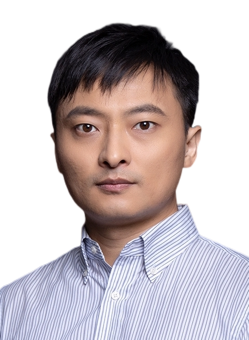}}]{Huawei~Huang} (SM'22) received his Ph.D. in Computer Science and Engineering from the University of Aizu, Japan. He is an associate professor at the School of Software Engineering, Sun Yat-Sun University, China. His research interests mainly include blockchain systems. He has served as a guest editor for multiple special issues on blockchain at IEEE JSAC, OJ-CS, and IET Blockchain. He also served as a TPC chair for a number of blockchain conferences, workshops, and symposiums.
\end{IEEEbiography}

\begin{IEEEbiography}[{\includegraphics[width=1in,height=1.25in,clip,keepaspectratio]{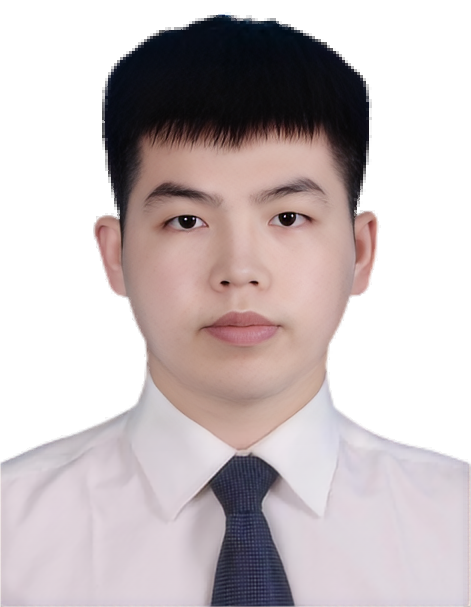}}]
{Zhaokang~Yin} is currently a master-program student at the School of Software Engineering, Sun Yat-Sen University. His research interests mainly include Blockchain.
\end{IEEEbiography}

\begin{IEEEbiography}[{\includegraphics[width=1in,height=1.25in,clip,keepaspectratio]{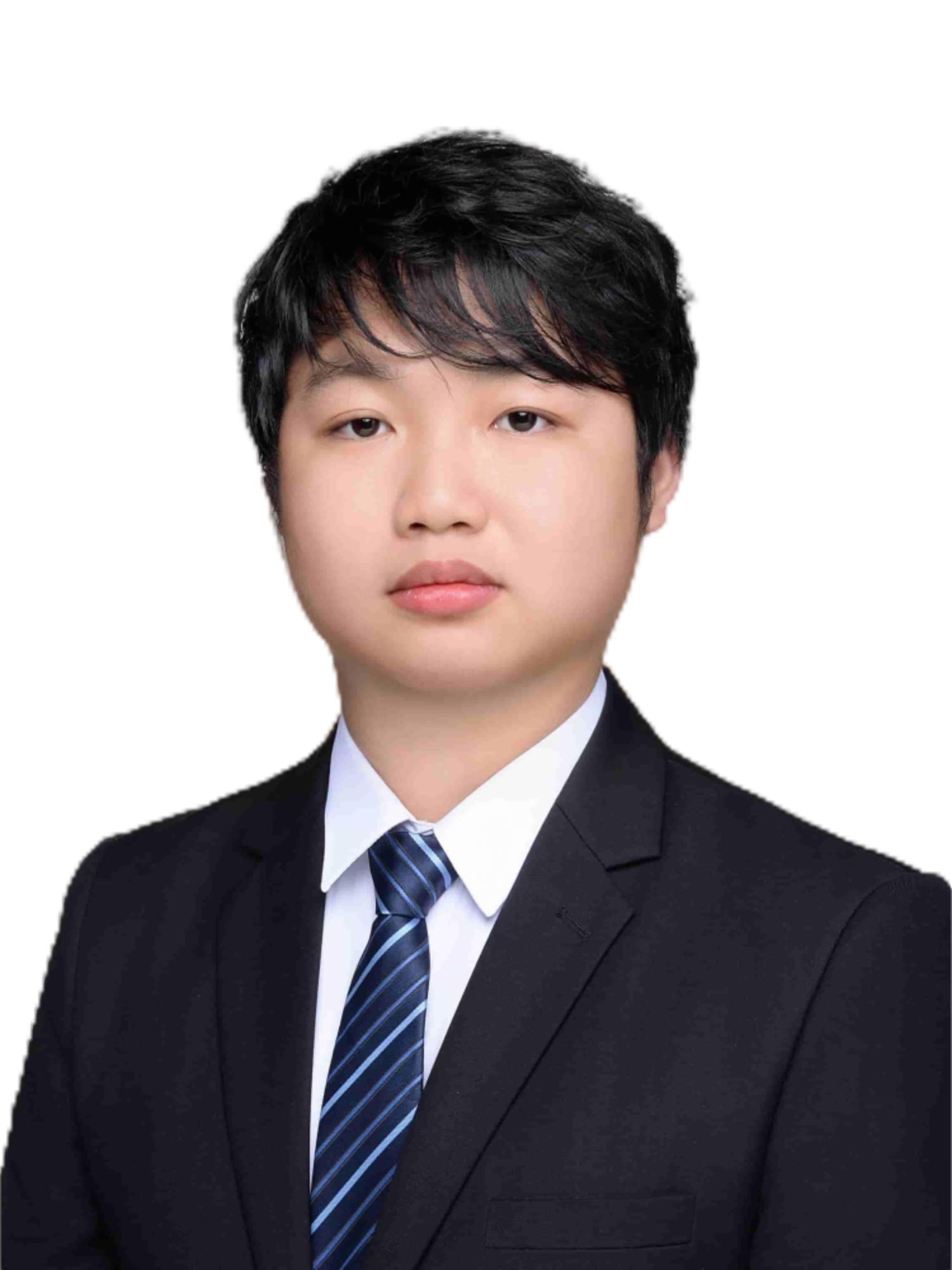}}]
{Qinde~Chen} is a Ph.D. student at the School of Software Engineering, Sun Yat-sen University. His research interests mainly include blockchain. 
\end{IEEEbiography}

\begin{IEEEbiography}[{\includegraphics[width=1in,height=1.25in,clip,keepaspectratio]{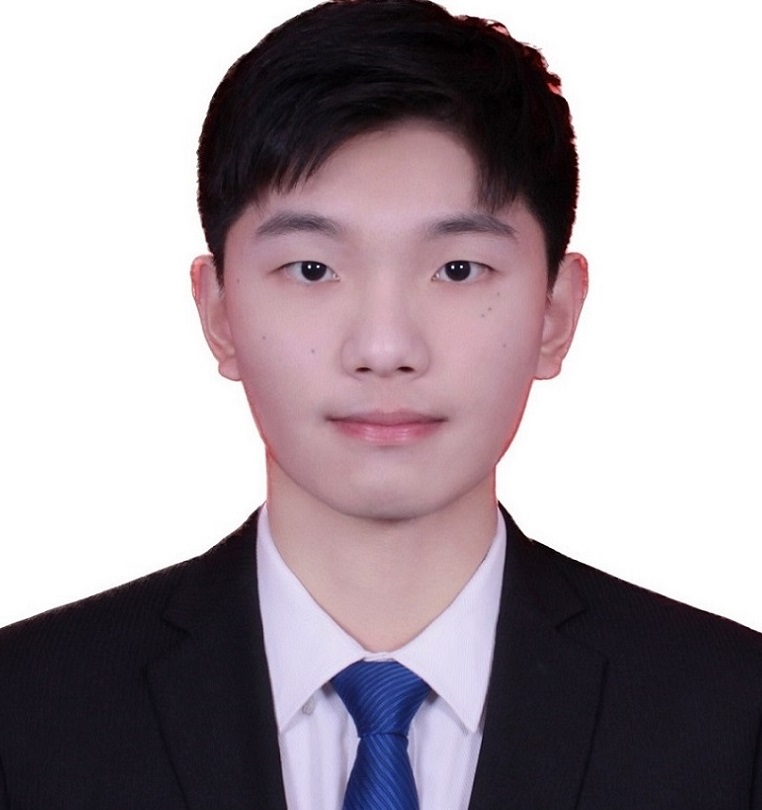}}]
{Guang~Ye} is currently a master-program student at the School of Software Engineering, Sun Yat-Sen University. His research interests mainly include Blockchain.
\end{IEEEbiography}

\begin{IEEEbiography}[{\includegraphics[width=1in,height=1.25in,clip,keepaspectratio]{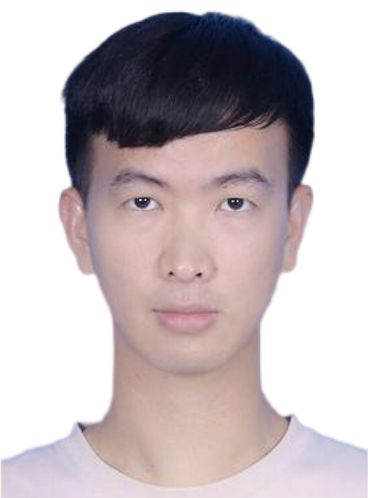}}]{Xiaowen~Peng} is currently a student pursuing his M.Sc. degree at the School of Computer Science and Engineering, Sun Yat-Sen University, China. His research interest mainly focuses on blockchain.
\end{IEEEbiography}

\begin{IEEEbiography}[{\includegraphics[width=1in,height=1.25in,clip,keepaspectratio]{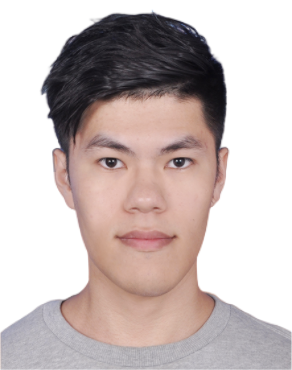}}]
{Yue~Lin} is currently a master-program student at the School of Computer Science and Engineering, Sun Yat-Sen University. His research interests mainly include Blockchain.
\end{IEEEbiography}

\begin{IEEEbiography}[{\includegraphics[width=1in,height=1.25in,clip,keepaspectratio]{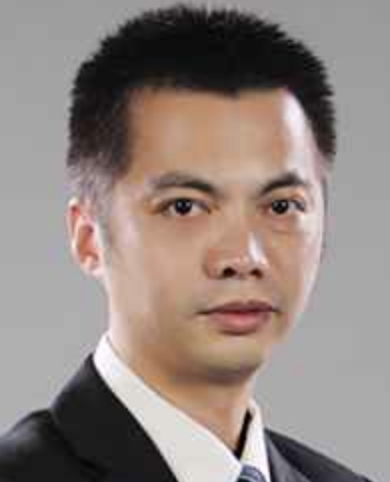}}]
{Zibin~Zheng} (SM'16-F'23)  received a Ph.D. degree from the Chinese University of Hong Kong, Hong Kong, in 2012. He is a Professor at the School of Software Engineering, Sun Yat-Sen University, China. His current research interests include service computing, blockchain, and cloud computing. Prof. Zheng was a recipient of the Outstanding Ph.D. Dissertation Award of the Chinese University of Hong Kong in 2012, the ACM SIGSOFT Distinguished Paper Award at ICSE in 2010, the Best Student Paper Award at ICWS2010, and the IBM Ph.D. Fellowship Award in 2010. He served as a PC member for IEEE CLOUD, ICWS, SCC, ICSOC, and SOSE. 
\end{IEEEbiography}

\begin{IEEEbiography}[{\includegraphics[width=1in,height=1.25in,clip,keepaspectratio]{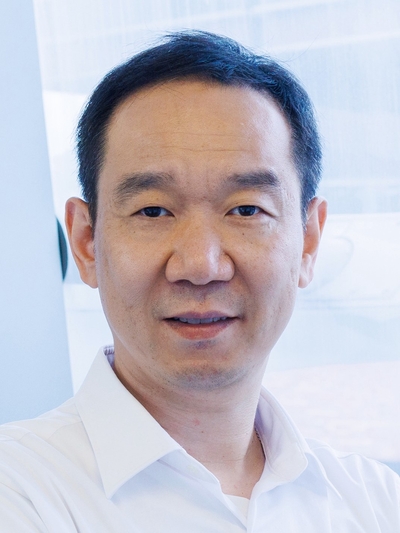}}]{Song~Guo} (M'02-SM'11-F'20) is a full professor in the Department of Computer Science and Engineering (CSE) at the Hong Kong University of Science and Technology (HKUST). Before joining HKUST in 2023, he was a Professor at The Hong Kong Polytechnic University. He received his PhD degree in computer science from University of Ottawa. He has authored or co-authored over 450 papers in major conferences and journals. His current research interests include big data, cloud and edge computing, mobile computing, and distributed systems. Prof. Guo was a recipient of the 2019 TCBD Best Conference Paper Award, the 2018 IEEE TCGCC Best Magazine Paper Award, the 2017 IEEE SYSTEMS JOURNAL Annual Best Paper Award, and six other Best Paper Awards from IEEE/ACM conferences. He was an IEEE Communications Society Distinguished Lecturer. He has served as an Associate Editor of IEEE TPDS, IEEE TCC, IEEE TETC, etc. He also served as the general and program chair for numerous IEEE conferences. He currently serves in the Board of Governors of the IEEE Communications Society.
\end{IEEEbiography}

\end{document}